\documentclass[
  aps,            % American Physical Society
  prx,            % Phys. Rev. X
  reprint,        % two-column, journal-like output (use "preprint" for a draft)
  superscriptaddress,
  amsmath,amssymb,
  footinbib,      % Place text footnotes in the reference list
  longbibliography
]{revtex4-2}

\usepackage{graphicx}   % figures
\usepackage{tikz}       % diagrams
\usepackage{tikz-cd}
\usepackage{placeins}   % keep SI figures within their owning subsections
\usepackage{bbm} % blackboard

\usetikzlibrary{arrows.meta,positioning,calc,decorations.pathreplacing,fit,shapes.geometric,backgrounds,matrix,tikzmark}

\definecolor{basecol}{HTML}{1F4E79}   % base code C_*
\definecolor{anccol}{HTML}{B03A2E}    % ancilla system E_*
\definecolor{mapcol}{HTML}{1E7A4D}    % chain map gamma
\definecolor{logcol}{HTML}{6C3483}    % logical space L_*
\definecolor{infocol}{HTML}{D5B443}   % information orbits (Figs. 1, 2 and the standard form in Sec. 3)
\definecolor{pivotcol}{HTML}{5D8EB8}  % Z-stabilizer pivot orbits
\definecolor{zerocol}{HTML}{D77870}   % X-stabilizer pivot orbits

\tikzset{
  cxdiag/.style={
    >={Stealth[length=2.2mm,width=1.6mm]},
    node distance=15mm and 26mm,
    obj/.style     ={rounded corners=2pt, minimum width=12mm,
                     minimum height=9mm, inner sep=2pt},
    base/.style    ={obj, draw=basecol!65, fill=basecol!7},
    anc/.style     ={obj, draw=anccol!65,  fill=anccol!7},
    log/.style     ={obj, draw=logcol!65,  fill=logcol!7, dashed},
    meas/.style    ={obj, draw=anccol!65,  fill=anccol!7},
    ar/.style      ={->, semithick, shorten >=1.5pt, shorten <=1.5pt},
    dar/.style     ={ar, dashed},
    lab/.style     ={font=\small, inner sep=1.5pt},
    diaglab/.style ={lab, pos=0.5, anchor=south east, xshift=2pt, yshift=-1pt},
    rowlab/.style  ={font=\small\itshape, align=right, inner sep=0pt},
  },
}

\definecolor{shiftcol}{HTML}{B7791F}   % the translated check g -> gx
\definecolor{axiscol}{HTML}{2B2B2B}    % off-black, for the group-action axis
\colorlet{xcol}{anccol}                % X-checks
\colorlet{qcol}{basecol}               % qubits
\colorlet{zcol}{mapcol}                % Z-checks

\newcommand{\symtsz}{0.70}   % tile size (cm)
\newcommand{\symdx}{0.22}    % depth offset, x
\newcommand{\symdy}{0.26}    % depth offset, y

\tikzset{
  symtile/.style={rounded corners=1.6pt, draw=#1!50, fill=#1!7,
                  minimum width=\symtsz cm, minimum height=\symtsz cm,
                  inner sep=0pt},
  symfront/.style={draw=#1, fill=#1!25, line width=0.8pt},
  symlnk/.style={-{Stealth[length=1.9mm,width=1.4mm]}, line width=0.7pt,
                 shorten >=0pt, shorten <=0pt},
  symtlnk/.style={symlnk, line width=0.55pt, opacity=0.6},
  symlab/.style={font=\small, inner sep=1pt, fill=white},
  symblab/.style={font=\small},
  symmod/.style={rounded corners=2pt, draw=#1!65, fill=#1!8, text=#1,
                 minimum width=11mm, minimum height=8mm, inner sep=2pt},
}

\newcommand{\symstack}[4]{%
  \foreach \i [evaluate=\i as \j using {int(7-\i)}] in {0,...,7}{%
    \node[symtile=#4] (#1-\j) at ({#2 + \j*\symdx}, {#3 + \j*\symdy}) {};
  }%
}
\usepackage{bm}         % bold math
\usepackage{braket}     % Dirac notation
\usepackage[colorlinks=true,allcolors=blue]{hyperref}

\usepackage{amsthm}
\newtheorem{theorem}{Theorem}
\newtheorem{lemma}[theorem]{Lemma}
\newtheorem{proposition}[theorem]{Proposition}
\newtheorem{corollary}[theorem]{Corollary}
\theoremstyle{definition}
\newtheorem{definition}{Definition}

\theoremstyle{remark}

\usepackage[capitalize,noabbrev]{cleveref}
\crefname{definition}{Definition}{Definitions}
\Crefname{definition}{Definition}{Definitions}

\usepackage{soul}
\usepackage[normalem]{ulem}

\definecolor{ehcol}{HTML}{2B2B2B}   % internal differential, horizontal
\definecolor{egcol}{HTML}{1E7A4D}   % gamma, the short diagonal within a layer
\definecolor{etcol}{HTML}{B7791F}   % coupling to either adjacent round

\tikzset{
  stobj/.style={rounded corners=2pt, minimum width=8mm, minimum height=6mm,
                inner sep=1pt, font=\footnotesize},
  stbase/.style={stobj, draw=basecol!65, fill=basecol!8},
  stanc/.style ={stobj, draw=anccol!65,  fill=anccol!8},
  star/.style  ={-{Stealth[length=1.7mm,width=1.25mm]}, line width=0.6pt,
                 shorten >=1.5pt, shorten <=1.5pt},
  stvirt/.style={star, densely dotted, line width=0.7pt},
  sttl/.style  ={star, line width=0.5pt, opacity=0.75},
  sttlv/.style ={stvirt, line width=0.6pt, opacity=0.75},
  stlab/.style ={font=\scriptsize, inner sep=1pt},
  stbrace/.style={decorate, decoration={brace, amplitude=4pt}, black!45,
                  line width=0.6pt},
  stport/.style={dashed, dash pattern=on 2pt off 1.6pt, line width=0.7pt,
                 black!50, shorten <=1.5pt},
  stbox/.style ={rounded corners=3pt, draw=black!35, dashed,
                 dash pattern=on 2.5pt off 2pt, line width=0.6pt},
  stincl/.style ={-{Stealth[length=2.1mm,width=1.7mm]}, basecol, densely dotted,
                  line width=1.0pt, shorten >=2pt, shorten <=2pt, opacity=0.8},
  stinclm/.style={-{Stealth[length=2.8mm,width=2.2mm]}, basecol,
                  line width=1.8pt, shorten >=2pt, shorten <=2pt},
}

\newcommand{\Ftwo}{\mathbb{F}_2}
\newcommand{\group}{G}
\newcommand{\code}[3]{[[#1,#2,#3]]}
\newcommand{\bdry}{\partial}

\DeclareMathOperator{\im}{im}
\begin{document}

\ifdefined\appendixonly
  % The appendix-only driver skips the title and main text, and begins directly
  % with the single-column appendix material below.
  \onecolumngrid
\else

\title{Efficient Logic with Ultra-High-Rate Quantum Codes}
\author{Nishad Maskara}
\email{nishadma@mit.edu}
\affiliation{Center for Theoretical Physics---a Leinweber Institute, Massachusetts Institute of Technology, Cambridge, MA 02139, USA}

\author{Rohan Mehta}
\affiliation{Department of Physics, Harvard University, Cambridge, MA 02138, USA}

\author{Zhiyang He}
\affiliation{Department of Mathematics, Massachusetts Institute of Technology, Cambridge, MA 02139, USA}

\author{Varun Menon}
\affiliation{Department of Physics, Harvard University, Cambridge, MA 02138, USA}

\author{J. Pablo Bonilla Ataides}
\affiliation{Department of Physics, Harvard University, Cambridge, MA 02138, USA}

\author{Mikhail D. Lukin}
\affiliation{Department of Physics, Harvard University, Cambridge, MA 02138, USA}

\author{Hengyun Zhou}
\affiliation{Department of Electrical Engineering and Computer Science, Massachusetts Institute of Technology, Cambridge, MA 02139, USA}

\begin{abstract}
Quantum low-density parity-check (qLDPC) codes offer a promising route to reducing the qubit overhead of fault tolerance, with recent ultra-high-rate constructions achieving encoding rates above one half.
In this work, we develop fast and parallel logical computation methods for ultra-high-rate codes. 
We first utilize the group symmetries of the codes to partition the logical qubits into orbits, and then apply parallel surgery methods to implement group-invariant collections of logical Pauli-product measurements.
These measurements are highly programmable; we may arbitrarily couple logical orbits as long as each measurement is repeated across the whole orbit(s).
To reduce the space and time overheads, we develop a technique called \textit{chain-map convolution}, which uses the group structure to boost both the spatial and temporal distances of our protocol in a highly efficient manner.
As case studies, we construct distance-preserving parallel measurement gadgets for the [[256, 68, 12]], [[576, 148, 18]] pair-partition codes, and a [[1152, 580, 12]] rate-$1/2$ code, with only $\sim 2.5\times$ space overhead and $3-4$ rounds of syndrome extraction each.
The same methods also enable parallel injection of cultivated magic states into high-rate codes.
Finally, we extend programmability through subgroup restrictions, revealing a tradeoff between logical control and the spacetime resources required for fault tolerance and highlighting the crucial role structured parallelism plays in reducing overhead.
These results provide new methods for designing qLDPC logical operations with low spacetime overhead, opening a route toward more efficient fault-tolerant quantum architectures.
\end{abstract}

\maketitle

% ==========================================================================
% Main text
% ==========================================================================

\section{Introduction}
\label{sec:intro}

\begin{figure*}[t]
\centering
\includegraphics[width=\textwidth]{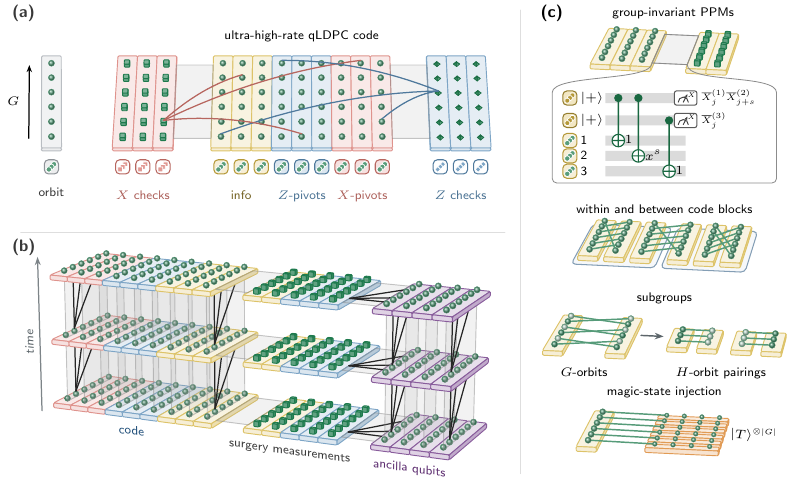}
\caption{\textbf{Efficient parallel logical operations with ultra-high-rate qLDPC codes.}
(a) Illustration of a qLDPC code organized by its group symmetry $G$, with each bar representing one orbit of qubits or checks. The data qubits are divided into three roles based on the logical standard form: the information coordinates (yellow), and $Z$- and $X$-stabilizer pivot coordinates (blue and red). Logical operators are specified by a single information coordinates and completed on the complementary pivot coordinates. Lines indicate the qubits involved in an example $X$ and $Z$ checks; most connections are omitted for clarity.
(b) Illustration of the spacetime surgery construction. Ancillary qubits and additional check measurements are introduced to measure the desired logical operators, with successive measurement rounds shown vertically. Chain-map convolution couples shifted positions along the symmetry orbits in space and time, increasing the distance with low overhead. Representative convolution connections are shown in black and appear along the transparent sheets.
(c) Examples of available logical operations. The top circuit measures pairs of logical $X$ operators from orbits 1 and 2 with a relative shift $s$, together with individual logical $X$ operators from orbit 3. The same measurement pattern is applied across each orbit in parallel.
Depicting PPMs as green edges between logical qubits, we can perform simultaneous pairwise measurements within and between code blocks.
Finer control is obtained by choosing a subgroup $H$ of $G$: the illustrated six-site orbits each split into two three-site orbits, and the pairwise measurements separate into two families that can be programmed independently. 
Finally, magic states can be injected by applying PPMs between a logical orbit and a stack of auxiliary magic-state blocks. 
}
\label{fig:overview}
\end{figure*}

Advances in quantum low-density parity-check (qLDPC) codes have substantially reduced the physical qubit cost of fault tolerance~\cite{panteleev2022asymptotically,leverrier2022quantum,bravyi2024high,breuckmann2022ldpc,xu2024constant,yoder2025tour,webster2026pinnacle,cain2026shor,he2026_qldpc_talk}, with recent ultra-high-rate constructions requiring only two physical qubits per logical qubit, while maintaining practically-relevant logical error rates and block sizes~\cite{zhao2026towards,kasai2026breaking}.
These constructions are commonly built out of group lifts of a smaller protograph, forming tunable families of code instances.
To reach the full potential of such ultra-high-rate constructions, we need fault-tolerant and efficient methods to operate on the densely encoded logical information.
Ideally, such fault-tolerant logical operations should combine
(i) \textit{low space overhead},
(ii) \textit{low time overhead},
(iii) \textit{high logical parallelism}, and
(iv) \textit{programmability suited to target algorithmic structures}~\footnote{We use ``programmability'' to describe the property that a logical computation scheme can implement flexible programs of logical gates and operations. A closely related term is ``addressability'', which captures the property that the scheme can address arbitrary logical qubit(s) or operator(s) with gates~\cite{he2025asymptotically}. As our schemes implements group-invariant operations that are flexible but not fully addressable, we use the term programmable throughout.}.

Nevertheless, these goals are often in tension with one another.
Consider for instance quantum code surgery, a general and flexible scheme for implementing fault-tolerant computation by measuring logical operators of quantum codes~\cite{yuan2026parsimonious,webster2025explicit,blue2026full,swaroop2024universal}.
Early studies showed how to measure arbitrary Pauli operators with low space overhead, but required $O(d)$ time overhead by rounds of repeated measurement and lacked parallelism~\cite{horsman2012surface,cohen2022low,cross2024improved,ide2025homological,williamson2026gauging,he2025extractors}.
Recent works have proposed parallel surgery methods~\cite{zhang2025time,cowtan2025parallel,zheng2025high,gu2026qgpu,zheng2026logical,bhardwaj2026high,cowtan2026near,hirasaki2026lifting}, which showcase a tradeoff between the complexity of measured operators and the space overhead. Similar to sequential surgery schemes, parallel surgery schemes generally use $O(d)$ rounds of repeated syndrome measurement to guarantee fault tolerance.
Fast surgery methods~\cite{cowtan2025fast,baspin2026,chang2026constant,berent2026lifted} are able to reduce this time overhead, but impose further constraints on the implementable set of measurements. 
Indeed, these challenges are particularly stark in the ultra-high-rate regime with $kd>n$, where many logical operators necessarily overlap. 
Applying existing methods directly to these codes will result in a significant space and time footprint as we perform a large number of operations. 

To address these challenges, we focus on an important class of operations motivated by both algorithms and hardware: those that respect a common symmetry.
Such operations are algorithmically useful, yet provide enough structure that facilitates low spacetime overhead and high parallelism.
Intuitively, these correspond to operations that are repeated multiple times in a circuit.
For example, in quantum simulation of lattice models, gates between adjacent sites are repeated, enabling a compilation into symmetric operations~\cite{xu2025batched,ismail2026transversal}.
In arithmetic, fundamental operations like ripple-carry adders~\cite{cuccaro2004new,xu2025fast} can similarly be written with repeated gate patterns across bit positions, allowing a symmetric compilation, with the only non-invariant part being the cascaded reactive measurements.
Symmetry is also central to proposals for hardware-efficient quantum error correction with neutral atoms and trapped ions~\cite{zhao2026towards,bhardwaj2026high,yang2026designer,tripier2026fault,xu2024constant,viszlai2023matching}, making it a promising ingredient to incorporate into practical realizations.

In this work, we introduce a general technique for performing efficient logic with ultra-high-rate codes, which achieves the aforementioned goals by generating symmetric logical operations and physical circuits.
Our techniques leverage the fact that state-of-the-art qLDPC code constructions
including lifted product codes~\cite{panteleev2022almostlinear,panteleev2022asymptotically,bhardwaj2026high,hong2026quantum,hirasaki2026lifting}, affine-permutation-matrix codes~\cite{kasai2026breaking,zhao2026towards,yang2026designer}, and pair-partition codes~\cite{okada2026pairpartition,lee2026design},
are naturally built from group algebras. 
By exploiting symmetries in the codes and logical operations, our constructions effectively amortize the cost of fault-tolerant quantum logic across group-translated copies of identical gates, similar to batched constructions~\cite{xu2025batched,golowich2025constant,chang2026constant,tan2025single} but acting within a code block.
As such, our constructions combine the space advantage of ultra-high-rate codes with low time overhead, high logical parallelism, and programmability at the level of group-invariant logical orbits. 

In more detail, for qLDPC codes built from a group algebra, we introduce a group-symmetry-adapted standard form~\cite{gottesman1997stabilizer} that organizes logical qubits into orbits transforming covariantly with the symmetry (Fig.~\ref{fig:overview}(a)).
Using this, we design parallel surgery gadgets for sets of Pauli-product measurements (PPM) that respect the group symmetry, i.e., each chosen Pauli-product measurement is performed together with all of its translations (Fig.~\ref{fig:overview}(c)). 
Our starting construction, while having low space footprint, has low spatial distance as well. 
From here, the standard approach is to repeat the construction in both space (with a technique called \textit{thickening}~\cite{he2025extractors,cohen2022low}) and time (by repeated syndrome measurements) to boost the respective distances, thereby resulting in large overheads.
We instead develop a new technique called \textit{chain-map convolution}, which exploits the group symmetry and replaces the repetition in space and time by convolutions (Fig.~\ref{fig:overview}(b)). 
This leads to a significant reduction in overheads, giving constructions that are parallel, programmable, and efficient.
The same methods also enable parallel injection of magic states or other resource states with lower overhead than prior approaches.

Our construction further offers the freedom to tune the group symmetry used. 
Instead of the full base group symmetry, we may use a subgroup symmetry and obtain smaller orbits, which enables more programmable logical control.
However, using a subgroup symmetry limits the effectiveness of chain-map convolution, which results in higher overheads.
This provides a systematic way to move along the frontier between structure, programmability, and spacetime overhead.

Using these methods, we construct distance-preserving parallel measurement gadgets for [[256, 68, 12]], [[576, 148, 18]] pair-partition codes~\cite{lee2026design}, and a [[1152, 580, 12]] rate-$1/2$ code~\cite{kasai2026breaking,zhao2026towards}, finding that despite measuring $k/2$ pairs of logical operators across the system, around $2.5\times$ space overhead and $3-4$ rounds of syndrome extraction suffice. This is a substantial reduction compared to prior protocols that are designed for general parallel product measurements~\cite{cowtan2025parallel,zheng2026logical}. 
We find that the majority of the saving comes from using the larger group symmetry, highlighting how \emph{structured} parallelism---rather than just parallelism---can enable more efficient gadgets.

The remainder of this paper is organized as follows. Section~\ref{sec:background} reviews code surgery and its description in terms of symmetric chain complexes. Section~\ref{sec:convolution} develops chain-map convolution for parallel logical measurements and presents its spatial and temporal protection and circuit-level performance. Section~\ref{sec:reducing_structure} examines the trade-off between programmability and overhead under subgroup control. Section~\ref{sec:magic-state-teleportation} applies the construction to parallel magic-state teleportation. We conclude in Section~\ref{sec:outlook} with a discussion of future directions.

\section{Code Surgery and Symmetries}
\label{sec:background}

\begin{figure*}[t]
\centering
\includegraphics[width=1.0\textwidth]{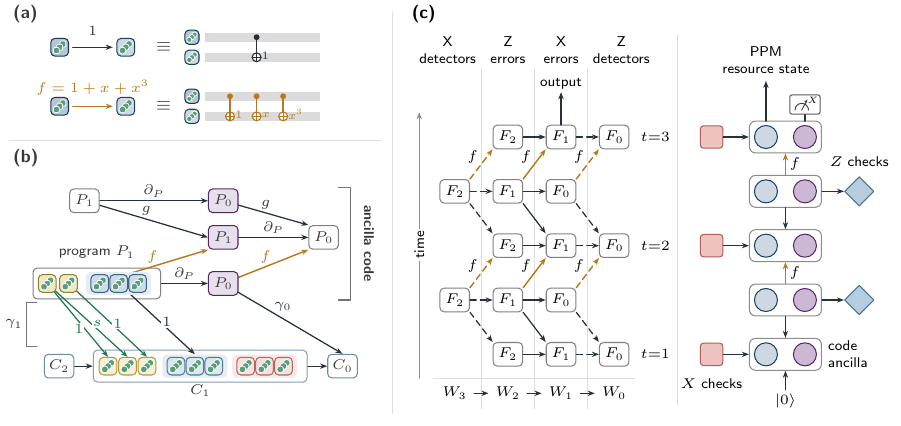}
\caption{\textbf{Programming and protecting parallel logical measurements.}
(a) Chain-map convolution between a check orbit and a qubit orbit. An identity connection corresponds to one layer of transversal CNOTs, and a non-trivial scalar like $f=1+x+x^3$ corresponds to three transversal CNOTs with relative cyclic shifts of zero, one, and three positions respectively.
(b) Ancilla system for parallel logical measurements. The complex $P_1\rightarrow P_0$ is attached to the data code through $\gamma_1$ and $\gamma_0$. Yellow orbits specify the measurements in Fig.~\ref{fig:overview}(c), while blue pivot coordinates complete the logical operators. Additional copies are introduced and coupled through $f$ and $g$ to increase the distance, with the paired maps commuting with $\partial_P$. Purple boxes denote ancillary data qubits.
(c) Temporal convolution and resource-state preparation. Left: fault complex for three rounds of the merged code, with time running upward. $F_2$, $F_1$, and $F_0$ denote $X$ checks, data qubits, and $Z$ checks of the static merged code, respectively. Solid connections specify physical couplings and dashed connections define detector parities. The second term $W_2$ encodes all of the ancilla qubits for $X$-check measurement, as well as MBQC data qubits that are sensitive to $Z$-errors, and vice versa for $W_1$. 
Shifted couplings $f$ (orange) between successive layers add protection against faults in time. Right: corresponding measurement-based quantum computation (MBQC) resource, with code and ancillary qubits in blue and purple. Edges represent CNOT couplings. Intermediate qubits are measured from bottom to top, followed by the remaining ancillary qubits, leaving the logic resource state on the output code qubits.
}
\label{fig:construction-pipeline}
\end{figure*}

Quantum code surgery implements logical Pauli-product measurements (PPMs) by extending the check structure of the original code. We can describe the data code by a chain complex
\begin{equation}\label{eq:Chain complex}
C_{\bullet} = C_2\xrightarrow{\,H_X^{\mathsf T}\,}C_1
\xrightarrow{\,H_Z\,}C_0.
\end{equation}
Physically, a set of $X$-checks is represented by a binary-valued vector $c \in C_2$, where $c_i = 1$ means the $i$-th $X$-check is part of the set. Analogously, sets of physical qubits and $Z$-checks are vectors in $C_1$ and $C_0$ respectively. 
Further, we will often refer to chain complexes using operational and physical terminology. Arrows between vector spaces correspond physically to CNOT circuits between collections of qubits (see Fig.~\ref{fig:construction-pipeline}(a)). For example, $C_2\xrightarrow{\,H_X^{\mathsf T}\,}C_1$ corresponds to an $X$-syndrome extraction circuit, and is implemented by applying a CNOT circuit with controls on $C_2$, targets on $C_1$, and connectivity set by $H_X$. 

To measure the target $X$-basis logical PPMs, the logical encoding is modified from the bare code $C_{\bullet}$ to a merged code $F_{\bullet}$, constructed by attaching an auxiliary complex $P_1\xrightarrow{\partial_P}P_0$ through maps $\gamma_1:P_1\to C_1$ and $\gamma_0:P_0\to C_0$, as illustrated in Fig.~\ref{fig:construction-pipeline}(b). In the resulting merged code, $P_0$ corresponds to ancilla data qubits and $P_1$ corresponds to ancilla $X$-checks spanning both the data and ancilla systems. The attachment also extends the original $Z$-checks onto the ancilla qubits. Commutation of the merged checks requires
\[
H_Z\gamma_1=\gamma_0\partial_P,
\]
which constrains how these couplings can be modified.
Products of the added checks $p \in P_1$ with $\partial_P p = 0$ (eq. $p\in\ker\partial_P$) cancel on the added data qubits, leaving an operator $\gamma_1p$ supported only on the original data qubits.  
Achieving high-logical parallelism corresponds to ensuring there are many products $p \in \ker{\partial_P}$, such that all of the target logical PPMs are included in the subspace generated by $\gamma_1 (\ket \partial_P)$, allowing their eigenvalues to be inferred from the check outcomes.

Constructing a valid merged code, however, is only the first step toward a fault-tolerant procedure. Errors in the check measurements can corrupt the inferred logical outcomes. Intuitively, the merged code distance only helps protect against data qubit errors, but not measurement errors.
The conventional approach addresses this by repeating syndrome extraction, adding protection to the check measurements by growing the distance in time, as illustrated in Fig.~\ref{fig:construction-pipeline}(c)~\cite{dennis2002topological}.
% Mathematically, the spacetime procedure is represented by a fault-complex $W_{\bullet}$ ~\cite{hillmann2025faultcomplex}, as illustrated in Fig.~\ref{fig:construction-pipeline}(c), that encodes the circuit structure of the full fault-tolerant procedure. 
% The distance of the full protocol is determined by the minimum number of faults that can change a logical measurement outcome or corrupt surviving logical information without detection, including faults that extend across both space and time.
% Intuitively, the distance of the X-checks is encoded in the ``code'' $W_3 \rightarrow W_2 \rightarrow W_1$, and by performing multiple rounds of measurement, the detectors $W_3 \rightarrow W_2$ form a repetition code between the merged code checks.
However, repetition requires $d$ rounds to increase the spacetime check-distance and match the spatial distance, generating a large spacetime overhead.

Efficient logical operations therefore pose two closely related design challenges: (i) implementing the desired logical measurements, and (ii) making them fault-tolerant with low space and time overhead. These challenges are particularly prominent for ultra-high-rate codes, which encode many logical qubits within a single code block. In the regime \(kd>n\), a full set of \(k\) nontrivial logical representatives must overlap, since each has weight at least \(d\). 
This makes it difficult to perform parallel logic using techniques that leverage disjoint low-weight logicals to build small surgery systems.
Instead, our method builds on constructions that share ancilla resources across overlapping logical measurements~\cite{zheng2025high,xu2025batched,chang2026constant,hirasaki2026lifting}.

Many recent high-rate qLDPC codes are built from objects defined over group algebras~\cite{panteleev2022almostlinear,xu2024constant,bhardwaj2026high,hong2026quantum,bravyi2024high,lin2023quantum, kasai2026breaking, zhao2026towards}, naturally imbuing the code with symmetric structures.
We exploit this structure by grouping checks, qubits, and operators into symmetry \textit{orbits} and expressing the code maps over the group algebra \(R=\mathbb F_2[G]\).
We refer to the entire code as a \textit{code block} and refer to collections of qubits or logical operators related by the symmetry as \textit{orbits}.

Here, we focus on Abelian groups built from circulant matrices of length $P$, as recent constructions based on them have been able to achieve state-of-the-art code parameters in pair-partition (PP) codes~\cite{lee2026design}.
For example, many of our constructions below are instantiated in a $[[576, 148, 18]]$ pair-partition code, which has a $G = \mathbb{Z}_{72}$ symmetry.
The corresponding vector spaces for this code have $C_2 = \mathbb{F}_2[G]^{3}$, $C_1 = \mathbb{F}_2[G]^{8}$, $C_0 = \mathbb{F}_2[G]^{3}$, since there are three orbits of X- and Z-checks and eight orbits of data qubits.
For a cyclic symmetry group $\mathbb{Z}_P$ of order $P$, the group algebra is equivalently the polynomial ring $\mathbb{F}_2[\mathbb{Z}_P] \cong \mathbb F_2[x]/(x^P-1)$ with $x$ representing a one-site cyclic shift.
Group elements define both positions within an orbit and also permutation operations. For example we can shift the $i$-th element with an orbit to position $i+1$ by multiplying by $x$, $x \cdot x^i = x^{i+1}$. Sums of group elements $f \in \mathbb F_2[G]$ describe multiple permutations, and can encode more complex symmetric circuits between orbits (see Fig.\ref{fig:construction-pipeline}(a)).
Working within the group algebra formalism ensures the operations preserve the physical symmetry, providing a compact representation to characterize the fault-tolerant gadget.

\section{Efficient Logic via Space and Time Convolution}
\label{sec:convolution}

We now describe the class of parallel logical measurements our approach supports, how to program them into an auxiliary complex, and protect them with low overhead in both space and time.
We first choose the measurements, construct a small ancilla system that implements them, and then modify its spatial and temporal couplings to improve protection.

\subsection{Programmable symmetric logical gates}

\paragraph*{Logical registers.}
In classical computers, registers of a fixed size provide a convenient unit for organizing information and executing parallel operations. Motivated by this structure, we consider quantum registers comprising multiple logical qubits, with operations applied in parallel across them.

For example, in a translation-invariant lattice Hamiltonian simulation, each Trotter layer applies the same gate to a set of equivalent, translated bonds. Encoding different sublattices into separate registers allows these interactions to be implemented by transversal gates together with relative cyclic shifts~\cite{xu2025batched,ismail2026fast}. Arithmetic circuits offer a related motivation: the repeated majority-function blocks of a ripple-carry adder suggest organizing qubits by their positions within these blocks, although carry propagation still imposes a sequential order~\cite{cuccaro2004new,xu2025fast}.

Recent proposals encode columns of lattice qubits into code blocks, and combine transversal gates with code automorphisms that act as simple logical translations across each register, restricting the choice of codes~\cite{xu2025batched,ismail2026fast}. For example, the bicycle chain codes of Ref.~\cite{xu2025batched,ismail2026fast} have a disjoint logical basis that bounds their encoding rate by $1/d$. To combine this algorithmic structure with the encoding advantage of ultra-high-rate codes, we instead use logical symmetry orbits within a larger code block as registers. Programming operations on selected orbits then provides control within the block while preserving the desired parallel logical action.

\paragraph*{A systematic logical basis.}
To support this programming, we generalize Gottesman's binary stabilizer standard form~\cite{gottesman1997stabilizer} to a logical basis that respects the cyclic symmetry. 
We work over the group algebra $R=\mathbb F_2[x]/(x^P-1)$, so that one coordinate represents an entire physical orbit. 
For the studied codes, we can obtain a systematic logical basis by row reduction and back-substitution over the group algebra, preserving the orbit structure (see Appendix~\ref{app:symmetry-logical-coordinates}).
Figure~\ref{fig:overview}(a) illustrates the construction: each logical basis generator has a single information bit and a completion on the appropriate stabilizer-pivot orbits. The $X$ and $Z$ generators are completed on the $Z$- and $X$-stabilizer pivot orbits, respectively. Their shared information bit and disjoint completion supports ensure the canonical logical commutation relations. Translating these representatives along the orbit columns in Fig.~\ref{fig:overview}(a) produces the remaining basis generators.
Many codes also have residual logical sectors~\cite{zheng2026logical,lee2026logical,lee2026design}, which lie outside the active full-orbit register; our constructions may also measure products confined to these residual sectors.
This standard form preserves the symmetry while providing explicit coordinates for logical programming. We optimize the pivot assignment to reduce the weight of the logical basis while preserving the information orbits. The systematic structure of this basis allows us to specify logical measurements through sparse couplings to the data, independent of the weight of the logical operators themselves in the high-parallelism regime.

\paragraph*{Programming logical measurements.}

Figure~\ref{fig:overview}(c) illustrates the logical operations enabled by symmetry.
For a cyclic group of order $P$, we label the logical operators by $X_{i,j}$ and $Z_{i,j}$, where $i$ identifies a full logical orbit and $j\in\mathbb Z_P$ identifies a qubit within that orbit~\footnote{Some logical qubits lie in symmetry-invariant sectors outside the full logical orbits considered here. These form the residual sector outside the active full-orbit register, and may undergo additional residual-sector measurements.}. Let $x$ denote the cyclic group generator, acting by $x\cdot X_{i,j}=X_{i,j+1}$ and $x\cdot Z_{i,j}=Z_{i,j+1}$ in every orbit, with indices understood modulo $P$. The dot denotes translation of the logical operator.
We choose seed Pauli products $M_1,\ldots,M_q$ and measure each seed together with all of its translations:
\begin{equation}
M_{\alpha,j}=x^j\cdot M_\alpha,
\qquad \alpha=1,\ldots,q,\quad j\in\mathbb Z_P.
\label{eq:cyclic-logical-measurements}
\end{equation}
Here $q$ is the number of seeds, $\alpha$ labels a seed, and $x^j$ shifts every factor in that seed by $j$ positions.
Each seed specifies the participating orbits and their relative shifts. 
Within each code block, all seeds use a common Pauli type, either $X$ or $Z$, while different blocks may use different types. 
All translated products therefore commute and define a measurement family $\mathcal M=\{M_{\alpha,j}\}$. 
By construction, all logical PPMs are repeated exactly $|G|$ times, so the group structure sets the degree of redundancy in the logical program.
In the following, we focus on one- and two-body logical measurements, which we implement by direct programming. Appendix~\ref{app:sparse-surgery-constructions} discusses more general constructions.

\subsection{Constructing the fault-tolerant circuit}

\begin{figure*}[t]
\centering
\includegraphics[width=\textwidth]{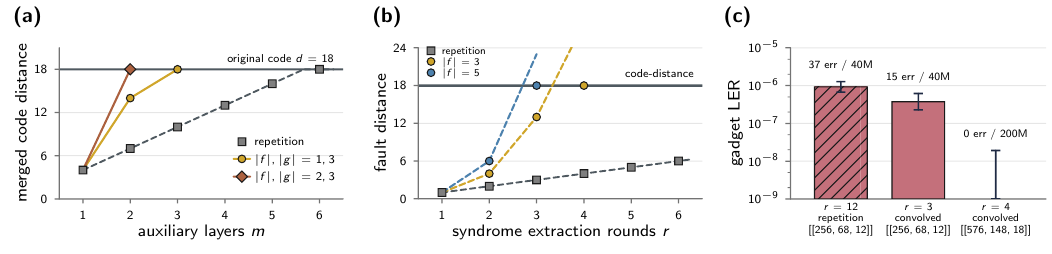}
\caption{\textbf{Simulation results.}
(a) Merged-code distance $d_X$ versus the number of auxiliary layers $m$ for the $[[576,148,18]]$ code. Each layer includes one primal–dual pair. Repetition (grey squares) results in spatial distance $d_X=\min(3m+1,18)$ (dashed). Convolution instead reaches distance with three layers for $(|f|,|g|)=(1,3)$ (yellow circles), or just two layers for $(2,3)$ (red diamond), where $f$ and $g$ are as labeled in Fig.~\ref{fig:construction-pipeline}(b).
(b) Temporal distance versus measurement rounds $r$ for readout of the same code. Repetition gives $d_Z=r$ (grey squares). Convolved profiles with temporal coupling weights $(1,3)$ and $(1,5)$ reach $d_Z=18$ in four and three rounds, respectively (yellow and blue circles). Both use shifts in $\mathbb Z_{36}\subset\mathbb Z_{72}$, with different scalars on successive edges. Matching dashed curves show the isolated classical temporal-code distances, which omit the spatial sectors and base-code boundary. The horizontal line marks the base-code distance.
(c) Circuit-level logical failure probability. The first two bars measure $32$ logical $X$ operators along an orbit and two orbit-invariant residual logical qubits on the $[[256,68,12]]$ code, while the last one measures $72$ logical $X$ operators on the $[[576,148,18]]$ code. Bars show the equal-weight average over $X$- and $Z$-basis experiments with 95\% confidence intervals.
}
\label{fig:convolution-performance}
\end{figure*}

\paragraph*{Constructing the ancilla system.}
For concreteness, throughout this section we use the seeds $M_u=X_{1,0}X_{2,s_0}$ and $M_v=X_{3,0}$, where $u$ and $v$ label the seeds and $s_0\in\mathbb Z_P$ is a chosen relative shift. 
The collection of measurements, after translation, include
$M_{u,j}=X_{1,j}X_{2,j+s_0}$ and $M_{v,j}=X_{3,j}$ for every $j\in\mathbb Z_P$. 
In other words, we measure shifted pairs from orbits 1 and 2 and, in parallel, the individual logical $X$ operators in orbit 3. The relative shift and participating orbits are programmable, providing finer control than a global transversal gate.
Later we also increase the degree of programmability by using subgroup symmetries to increase the number of logical orbits.

To construct a surgery ancilla system, we first build an auxiliary complex $P_1\xrightarrow{\partial_P}P_0$, and specify a chain map $\gamma$ which connects it to the base code complex. This complex specifies the measurements we perform.
We use two ancillary information orbits, labeled $u$ and $v$, to specify the two translated measurement families. Let the columns of $Q$ contain the logical coordinates of these seeds. Their coupling to the three data information orbits is
\begin{equation}
Q\begin{pmatrix}u\\v\end{pmatrix}
=\begin{pmatrix}1&0\\x^{s_0}&0\\0&1\end{pmatrix}
\begin{pmatrix}u\\v\end{pmatrix}
=\begin{pmatrix}u\\x^{s_0}u\\v\end{pmatrix}.
\label{eq:surgery-running-program}
\end{equation}
The first column sends $u$ to orbit 1 and its shift to orbit 2; the second sends $v$ to orbit 3, as shown in Figure~\ref{fig:construction-pipeline}(b). 
We include these two orbits in $P_1$, together with copies of the $Z$-stabilizer pivot coordinates needed to complete the selected logical operators.
Note that $P_1$ is the space of ancillary checks we use in surgery.
In $P_0$, we create a copy of each $Z$-check of the original code (elements of $C_0$) which are incident to the operators being measured.
In most cases we are acting on all the logical qubits and $P_0$ is a copy of $C_0$.
$\gamma_1$ maps $u, v$ from $P_1$ to the information coordinates in $C_1$, and maps the copied $Z$-stabilizer pivot coordinates to the respective coordinates in $C_1$ via identity. 
$\gamma_0$ is an identity map from $P_0$ to $C_0$.
The boundary $\partial_P$ follows the same connections as the data code, so that $H_Z\gamma_1=\gamma_0\partial_P$.
One can then verify that products of the corresponding merged checks cancel on the ancilla qubits and give precisely the translated products $M_{u,j}$ and $M_{v,j}$. 
Appendix~\ref{app:sparse-surgery-constructions} describes this construction in more detail and the conditions ensuring that no additional logical information is measured.

\paragraph*{Spatial convolution.}

This ancilla system implements the intended logical action, but the distance of the merged code can be substantially smaller than that of the original code.
A large logical operator in the original code can be deformed onto the ancilla system at much lower weight.
For example, for the $[[576,148,18]]$ instance in Fig.~\ref{fig:convolution-performance}(a), the $X$ spatial distance decreases from 18 to 4 with a single auxiliary layer. 
By coupling to the auxiliary complex, there are now undetectable error patterns of the form $(e, H_Z e)$ with $|e|=1$ and $|H_Z e| = 3$, that are supported on the original data qubits $C_1$ and ancilla qubits $P_0$ respectively.
These low-weight operators correspond to logical operators of the base code which are cleaned, or weight-reduced, by the new checks from $P_1$.

A conventional remedy is \textit{thickening}, which repeats the auxiliary complex in several coupled layers~\cite{cohen2022low}. 
By adding a layer, this introduces new $Z$-checks that add protection to the first ancillary $P_0$ space, forcing $X$ logical operators to extend further through the stack.
We characterize the rate of distance growth with repetition thickening alone in Fig.~\ref{fig:convolution-performance}(a), and observe the distance grows by 3 for every additional layer. 
Simply creating many layers of the auxiliary complex therefore introduces a significant space cost.
An effective modern method that boosts distance with lower overhead uses graph expansion~\cite{cross2024improved,williamson2026gauging,he2025extractors}.
However our ancilla complexes measure overlapping logical operators and correspond to hypergraphs, for which fewer efficiency guarantees exist~\cite{zheng2025high,cowtan2026near}. 
We therefore need new methods for efficient surgery constructions. %

Here, we use the group structure to improve distance. We replace the identity connections between adjacent layers by sums of translated connections that form a chain map, a procedure we call \emph{chain-map convolution}.
The lift group provides distinct shifts that can be combined to create additional connections between layers; this is in contrast to the binary field, where the only nonzero scalar is the identity.
For example, using $x$ for the one-position cyclic shift, we can replace an identity connection between two orbits by
\begin{equation}
1\quad\longrightarrow\quad f=1+x+x^3,
\label{eq:surgery-convolution-example}
\end{equation}
as illustrated in Fig.~\ref{fig:construction-pipeline}(a).
Physically, this grows the connectivity between checks and qubits, increasing their degree by two.
A coordinate $j$ now connects to $j$, $j+1$, and $j+3$ in the next layer.
When implementing as a syndrome extraction circuit, each term of $f$ corresponds to one transversal CNOT between the pair of check and qubit orbits.
Intuitively, these connections increase distance by requiring additional qubit support to cancel a detectable syndrome.
We call this \emph{scalar convolution}, as it applies the same group-algebra element $f$ independently to each existing connection between orbits; matrix-valued generalizations can also introduce new connections between different orbits, as discussed in Appendix~\ref{sec:chain-complex-convolution} and Fig.~\ref{fig:structure-overhead}.

The paired $f$ maps in Fig.~\ref{fig:construction-pipeline}(b) must still form a commuting square. For scalar convolution over the cyclic group algebra, applying the same polynomial to both sides of an interlayer square ensures
\begin{equation}
\partial_P(fI_{P_1})=(fI_{P_0})\partial_P.
\end{equation}
This follows from the fact that we are using an Abelian group, since all elements of $R$ commute. 
Further selecting unit scalars is sufficient to ensure the programmed measurements are preserved by this procedure.

We directly verify the code distance of the resulting merged codes using a branch-and-bound search for lower bounds~\cite{dumer2017distance} and witnesses for upper bounds.
Although there are auxiliary (or gauge) logical degrees of freedom, we still benchmark the distance of the entire code.

Notably, convolution increases the local degree, but the starting auxiliary complex often has fewer incident couplings than the original code, because it only includes the coordinates needed for the chosen measurements. This can leave room for several translated connections with only a small increase in the maximum degree of the merged complex. The construction also has a repeated coupling pattern across each orbit, whose cyclic structure may simplify implementations in reconfigurable neutral-atom or trapped-ion systems~\cite{zhao2026towards,tripier2026fault,xu2024constant,bhardwaj2026high,yang2026designer}, particularly compared to constructions based on complex graphs~\cite{williamson2026gauging}.

Figure~\ref{fig:convolution-performance}(a) compares repetition and convolved couplings for a fixed measurement family, showing the spatial distance attained at each auxiliary thickness. Here, we use the convention that a primal and dual layer count as one layer, so that the original construction in Ref.~\cite{cohen2022low} involves $d$ layers. For the $[[576,148,18]]$ code, a $(|f|,|g|)=(2,3)$ coupling raises the static $X$ distance from $4$ to the base-code distance $18$ with two auxiliary layers, while a $(1,3)$ coupling reaches $14$ at two layers and $18$ at three. Standard repetition has distance $7$ at two layers.

\paragraph*{Temporal convolution.}

Remarkably the same principles apply in the time direction as well. 
We represent the full spacetime procedure by a fault-complex $W_{\bullet}$ ~\cite{hillmann2025faultcomplex}, as illustrated in Fig.~\ref{fig:construction-pipeline}(c).
This complex directly maps onto a measurement-based quantum circuit for the fault-tolerant procedure, that prepares a resource state for the target PPM~\cite{bolt2016foliated}.
The four terms in this diagram from left to right are $W_3 \rightarrow W_2 \rightarrow W_1 \rightarrow W_0$. In order they correspond to
\begin{enumerate}
    \item $W_3$: $X$-basis detectors
    \item $W_2$: Qubits initialized in $\vert + \rangle$ and measured in $X$.
    \item $W_1$: Qubits initialized in $\vert 0 \rangle$ and measured in $Z$.
    \item $W_0$: $Z$-basis detectors.
\end{enumerate}
Solid lines in the diagram are physical CNOT gates applied between qubits. Dotted lines are parities computed from measurements as part of the classical decoding apparatus. 
While the fault-complex captures the connectivity of the CNOT circuits that comprise the spacetime procedure, it does not capture the schedule and order the gates are applied.
In principle, the entire circuit could be applied in a single time step, using a number of qubits proportional to the size of the entire complex.
However, throughout this work, we interpret the fault-complex as a procedure that is applied in layers, from bottom to top.
Since all connections are local in this orientation, the entire circuit can be executed with no more than $2n_{\text{merged}}$ qubits at any given time, where $n_{\text{merged}}$ is the number of data qubits in the merged code.

We now walk through the procedure in detail to prepare a surgery resource state. Starting from the bottom of Fig.~\ref{fig:construction-pipeline}(c), the first layer corresponds to initializing all qubits $F_1$ in $\vert 0 \rangle$, and then using an ancilla system $F_2$ to measure the $X$-stabilizers of the merged code. 
Moving from the first layer to the second layer involves introducing a second copy of $F_1$, applying a layer of transversal CNOTs (coupling $I$), and then measuring out the first copy, performing a standard quantum teleportation circuit~\cite{zhou2000methodology}.
Then in the second layer, the $Z$-stabilizers are measured using an ancilla system $F_0$.
Without convolution, the circuit would then proceed using another standard quantum teleportation circuit.

When convolution is applied to the fault-complex, it corresponds to multiplying all of the edges between two layers by a scalar $f \in \mathbb F_2[G]$. Unless otherwise stated, in our constructions we apply a minimal amount of temporal convolution, after every $Z$-stabilizer measurement round and using a weight-3 scalar. 
Physically, the only change in the circuit is a small modification to the quantum teleportation procedure.
For example, given $f = 1 + x + x^3$, we apply the transversal CNOT gate three times between the two copies of $F_1$, once in the standard position, once with the second copy shifted by $x$, and once with the second copy shifted by $x^3$.
This entire procedure constitutes one round of syndrome extraction (SE), and the effect of convolution is to increase the circuit depth by a modest $|f|-1$ layers of CNOTs per round.

After multiple rounds of SE such that the $X$-stabilizers have been fault-tolerantly measured, all logical operators in the merged code are in the logical $\vert 0 \rangle$ state, and all $X$-stabilizers have an established logical frame. 
The final step of the state-preparation procedure involves measuring the ancilla block in the $Z$-basis.
After this step, all that remains is the bare code $C_{\bullet}$. Assuming we apply Pauli corrections to correctly orient the logical frame, all unmeasured logical $Z$ operators are in the $\vert 0 \rangle$ state, while all measured logical $X$ operators are in the $\vert + \rangle$ state. This is a resource state for the PPM measurement. To see this, note that we can use it to perform a PPM on a copy of $C_{\bullet}$, simply by applying a CNOT gate from the resource state to the copy, and measuring the resource state in the $X$-basis. 
Therefore, preparing a resource state is sufficient for generating logical PPMs.

The distance of the full protocol is determined by the minimum number of faults that can change a logical measurement outcome or corrupt surviving logical information without detection, including faults that extend across both space and time. Although we discuss the physical implementation in terms of individual gates, the distances we compute are in the \textit{phenomenological} model, where faults only occur during state-preparation or readout.
Informally, the protected $X$ information sensitive to $Z$ errors is encoded in the ``quantum code'' $W_3 \rightarrow W_2 \rightarrow W_1$, and the protected $Z$ information sensitive to $X$ errors in $W_2 \rightarrow W_1 \rightarrow W_0$ (see \cref{app:efficient-fault-tolerance} for more details).

In Fig.~\ref{fig:convolution-performance}(b), we show the distance-boosting performance of temporal convolution on a $d=18$ pair-partition code. 
For these full fault processes, we count only undetectable patterns that change a measured logical outcome or the surviving logical frame
\cite{delfosse2023spacetime_v2,williamson2026gauging} (see Appendix~\ref{sec:ft-process-requirements} for more details).
The standard repeated syndrome extraction has a time-like fault distance that is equal to the number of measurement rounds, requiring $r=18$ rounds to achieve the code distance. Empirically, the degree of distance boosting grows exponentially at first as $|f|^r$. Therefore, with as few as $r=4$ measurement rounds, temporal convolution with $|f|=3$ connections can already reach $d=18$, and as few as $r=3$ for $|f|=5$ connections, yielding a significant reduction in spacetime volume.

\subsection{Performance and scaling}

\paragraph*{Example constructions.}
We construct explicit surgery gadgets for symmetric pairwise PPMs on three representative codes: $[[256,68,12]]$, $[[576,148,18]]$, and $[[1152,580,12]]$.
Table~\ref{tab:representative-merged-codes} summarizes the merged-code parameters and syndrome-extraction rounds required to generate a fault-tolerant procedure. The target PPM count includes only the intended pairwise logical products. Each construction also measures two parities of residual logicals. The total footprint is reported as the number of data qubits plus the number of checks.
The reported check weights and qubit degrees are the maxima across all elements, and the SE round counts refer to the number of layers of an $X/Z$ syndrome measurement plus one temporal convolution.
We see in the studied cases, a two-layer construction is sufficient to achieve the target distance, generating a rate-independent $2.5 \times$ overhead. 

\begin{table}[t]
\caption{\textbf{Examples of merged codes and time overheads.}
Code parameters only include data qubits. Total $N$ counts data qubits and checks. Target PPMs are fully parallel pairwise PPMs between orbits.}
\label{tab:representative-merged-codes}
\centering
\footnotesize
\setlength{\tabcolsep}{2pt}
\renewcommand{\arraystretch}{1.15}
\begin{tabular*}{\columnwidth}{@{\extracolsep{\fill}}lccc@{}}
\hline\hline
 & PP256 & PP576 & APM1152 \\
\hline
Base code & $[[256,68,12]]$ & $[[576,148,18]]$ & $[[1152,580,12]]$ \\
Merged code & $[[576,34,12]]$ & $[[1296,74,18]]$ & $[[2304,290,12]]$ \\
Target PPMs & $32$ & $72$ & $288$ \\
Check weight $X/Z$ & $9/9$ & $10/10$ & $12/13$ \\
Qubit degree $X/Z$ & $5/4$ & $5/6$ & $9/6$ \\
$N_{\mathrm{merge}}/N_{\mathrm{base}}$ & $2.5$ & $2.5$ & $2.5$ \\
SE rounds & $3$ & $4$ & $3$ \\
\hline\hline
\end{tabular*}
\end{table}

\paragraph*{Circuit-level performance.}
To test whether the time overhead reduction survives under circuit-level noise, we simulate block-readout experiments for the $[[256,68,12]]$ code using the repetition and convolved gadgets in Fig.~\ref{fig:convolution-performance}(c). These experiments serve as a proxy for surgery performance. We include two-qubit-gate and state-preparation-and-measurement (SPAM) errors at $p=0.1\%$, with no idle noise. The circuit corresponds to measurement-based quantum computation (MBQC), in which the SPAM errors also play the role of additional idling errors; therefore, the absolute logical error rates are not directly comparable with standard circuit-based simulations. Both gadgets use the same hierarchical decoder, with multiple tiers of belief propagation followed by a time-limited integer-programming fallback for unresolved shots or conflicting logical predictions.

With $20$ million shots in each of the $X$ and $Z$ bases, we observe $37$ failures for the $r=12$ repetition gadget and $15$ for the $r=3$ convolved gadget, including both memory and logical measurement failures. Averaging equally over the two bases gives gadget LERs of $(9\pm2)\times10^{-7}$ and $(4\pm1)\times10^{-7}$ per complete experiment, respectively, where the uncertainties are one standard error, calculated from the independent binomial samples. Thus, in this comparison, the convolved gadget uses four times fewer measurement rounds while achieving a lower observed gadget LER.
Figure~\ref{fig:convolution-performance}(c) also shows a larger instance $[[576,148,18]]$ with $r=4$ rounds, where we observed no failures in $100$ million shots per basis.

\paragraph*{Lifted-product ancilla codes.}

In addition to the practical finite-size performance we demonstrate numerically, we also show the structures generated by convolution support good asymptotic properties.
The standard repetition based approach to spatial thickening corresponds to taking the hypergraph product of the ancilla program $P_1 \rightarrow P_0$ with a distance $\Theta(d)$ repetition code. %
However, by organizing the objects into symmetry orbits, and applying convolution between repetition layers, the resulting structure becomes equivalent to taking the lifted product of the starter block $P_1 \rightarrow P_0$ with a quasi-cyclic code obtained by lifting a much shorter repetition code (\cref{app:sparse-surgery-constructions}).
Replacing the transverse factor by a lifted product of two Panteleev–Kalachev (PK) expander codes, for an abelian symmetry group of order $g = 2^k$, yields a merged code with distance $d$ and $O(\max\{1,\log^2 |G|,(d/|G|)^2\})$ spatial overhead.
For $|G|=\Omega(d)$ this reduces to $O(\log^2 |G|)$ overhead, compared to $O(d)$ for a repetition code. The construction and proof are given in Appendix~\ref{si-pk:sec:two-page-proof}.

To analyze the asymptotic spacetime cost of fault-tolerant surgery with symmetries, we also study more general time-like repetitions. In particular, connections between widely separated time-like layers are also possible~\cite{xu2025batched,xu2026framework}, and will be treated in more detail in \cite{mehta2026unifying}.
With this added flexibility, we show that there exist a class of quantum codes whose spacetime surgery complexes can be built from high-dimensional expanders~\cite{dinur2024cubical_doi, nguyen2025quantum}, giving $O(\log^4 |G|)$ spacetime overhead for symmetric logical PPMs. The precise guarantees and proof are given in Appendix~\ref{sec:pk-temporal-distance}.
A common feature of these asymptotic results is that the symmetry group $G$ shared by the base code $C_{\bullet}$ and starter surgery complex $P_{\bullet}$ needs to be sufficiently large to benefit from the asymptotic savings. 
In the following section, we explore the relationship between overhead and the size of the symmetry in practice.

\section{Increasing programmability and overhead tradeoff}
\label{sec:reducing_structure}

The preceding construction uses group symmetry to perform parallel logical measurements with low space and time overhead. 
We now consider restricting this symmetry to a subgroup, in order to generate logic gates with finer control.
This divides logical registers into smaller, independently addressable orbits, while reducing the structure available for distance boosting. We can therefore chart out a trade-off between programmability and the resources required for protection.

\paragraph*{Increasing programmability via subgroup symmetry.}
Restricting to a subgroup $H\subseteq G$ divides each full $G$-orbit into $|G|/|H|$ smaller individually addressable orbits.
Figure~\ref{fig:structure-overhead}(a) illustrates this control for $72$ parallel two-body Pauli product measurements on the regular $G=\mathbb Z_{72}$ orbits of the $[[576,148,18]]$ code. We consider random pairwise matchings compatible with the subgroup symmetry: each logical qubit belongs to exactly one pair, and the resulting $72$ PPMs are invariant as a family under $H$. 
We treat the cost of implementing these random matchings as an indicator for a typical high-parallelism circuit.
Since each logical operation is repeated exactly $|H|$ times, we refer to it as the \textit{batch size}. The total number of independently programmable blocks scales inversely with the batch size as $72 / |H|$. At $|H|=18$, there are four such families; reducing to $|H|=6$ gives twelve, while $|H|=1$ allows all $72$ measurements to be programmed independently.
However, when we build the ancilla system associated with a generic logical program, the chain complex symmetry is reduced to $H$, which makes scalar convolutional distance boosting less effective. 
Intuitively, $|H|$ sets the number of positions over which convolution can spread logical operators within an orbit. At the trivial subgroup, $|H|=1$, there are no nontrivial translations within an orbit, and we recover standard parallel surgery~\cite{zhang2025time,cowtan2025parallel,zheng2025high,cowtan2026near}.

Figures~\ref{fig:structure-overhead}(b,c) compare the resources required for these random pairwise matchings at different batch sizes. At each batch size, we sample three matchings and optimize complexes using scalar convolution alone (blue), or a heuristic matrix-convolution construction with targeted ancilla addition and check splitting (red); see Appendix~\ref{sec:chain-complex-convolution} for the matrix-valued coupling framework. Both constructions preserve the full $H$ symmetry throughout. For each method, we plot the largest overhead among the complexes obtained for the three matchings; the relative spread is generally small. Fig.~\ref{fig:structure-overhead}(b) shows the total physical-qubit footprint, including merged-code data qubits and one measurement ancilla per listed check, normalized by the $144$ participating logical qubits. For scalar convolution, this space cost decreases from $45.5$ physical qubits per participating logical qubit at $|H|=1$ to $17.5$ at $|H|=18$. Matrix convolution further reduces the space cost at several of the tested subgroup sizes, with the largest reduction at the trivial subgroup.

Figure~\ref{fig:structure-overhead}(c) estimates the spacetime cost $r n_{\mathrm{phys}}/144$, where $n_{\mathrm{phys}}$ is the total number of data and measurement-ancilla qubits and $r$ is the number of rounds required to reach $d=18$. The scalar-convolution spacetime cost decreases from $819$ physical-qubit rounds per participating logical qubit at $|H|=1$ to $70$ at $|H|=18$, over an order of magnitude reduction. For $|H|\geq4$, the estimates are approximately described by the fit $1258.1/|H|$. Matrix convolution lowers the estimated spacetime cost at several subgroup sizes, although the improvement is smaller where the scalar construction is already efficient. These results illustrate the resource savings available under a larger subgroup symmetry, at the cost of fewer independently programmable measurement families.

\begin{figure*}
\centering
\includegraphics[width=\textwidth]{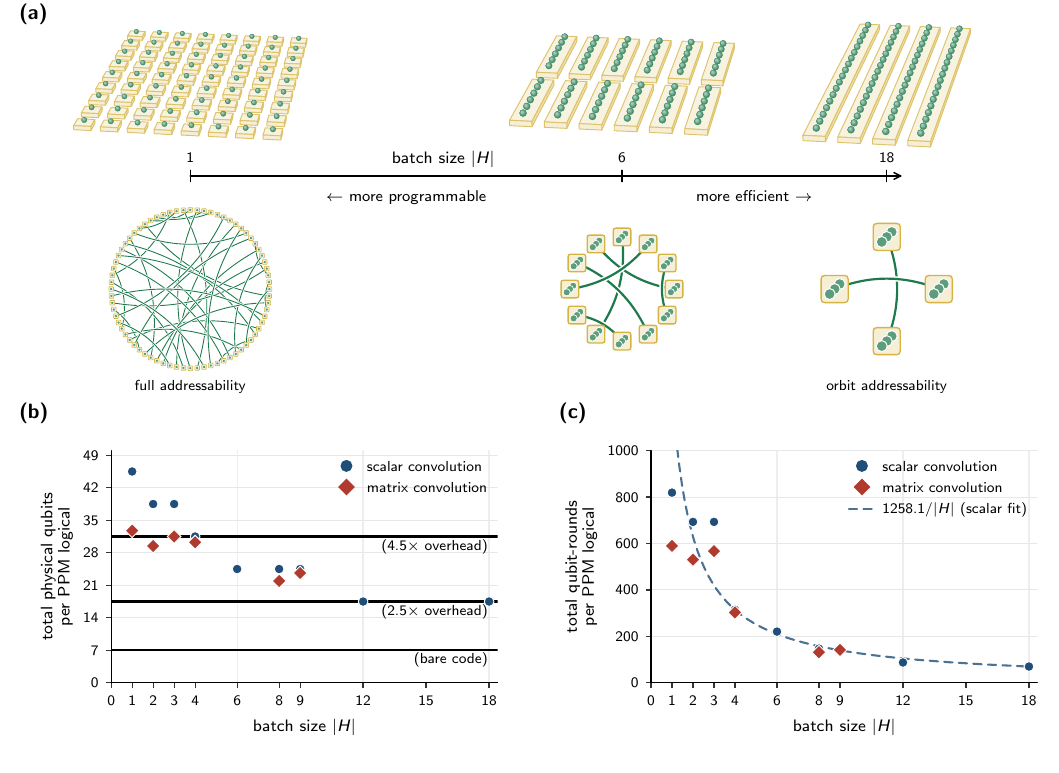}
\caption{\textbf{Trade-off between programmability and overhead.}
(a) Subgroup control for $72$ parallel pairwise Pauli-product measurements on $144$ logical qubits of the $[[576,148,18]]$ code. Each measurement is repeated across a subgroup orbit of size $|H|$, giving $72/|H|$ independently programmable pairwise measurements. Smaller subgroups provide finer orbit addressability, reaching full addressability with all measurements independently programmable at $|H|=1$.
(b) Space overhead per participating logical qubit for scalar convolution (blue circles) and matrix convolution with ancillary qubits and check-weight reduction (red diamonds). Both preserve the subgroup symmetry. Each point gives the largest overhead among the sampled matchings, using the best certified design found for each.
The cost $n_{\mathrm{phys}}/144$ includes merged-code data qubits and one measurement ancilla per check, including redundant checks.
(c) Estimated spacetime overhead per participating logical qubit,
$r n_{\mathrm{phys}}/144$. We observe a scaling in the cost reduction with batch size that is approximately $\approx 1/|H|$. To highlight this, we add an empirical fit of
$1258.1/|H|$
to the scalar-convolution results for $4\leq|H|\leq18$, that is extrapolated toward smaller batches. 
}
\label{fig:structure-overhead}
\end{figure*}

\paragraph*{Programmability and overhead.}
Our results expose an intriguing trade-off between the degree of programmability of logical measurements, and the spacetime overhead required to protect them. In particular, we observe that within our construction, increasing the batch size allows protection resources to be shared much more effectively, reducing the overhead per logical operation.

This perspective also connects to several recent and distinct techniques for performing parallel and low-overhead fault-tolerant logic. 
In Ref.~\cite{xu2025batched}, batching across code blocks was used to achieve provable constant spacetime overhead for identical Clifford operations (assuming fast classical decoding). The key idea was to implement a batch of Clifford gates by code-switching into a different representation, where such operations could be performed transversally. Therefore, to preserve constant overhead, the required batch size has to grow with the target distance, since the number of batched operations is directly related to the number of encoded logicals in an asymptotically code $[[n=\Theta(d), k=\Theta(d), d]]$ quantum code~\cite{panteleev2022asymptotically}.
The same batching of protection across code blocks also enables the low-overhead bulk state preparation and code-switching constructions of Ref.~\cite{golowich2025constant,nguyen2025quantum}.
Parallel surgery provides another setting in which multiple logical operations share protection resources.
The constructions from Refs.~\cite{gu2026qgpu,hirasaki2026lifting,zheng2025high,chang2026constant} all 
use shared ancilla systems for parallel logical measurements, while 
the construction from Ref.~\cite{chang2026constant} further implement $O(d)$ distinct parallel surgery operations in $O(d)$ time via constant-time surgery~\cite{cowtan2025fast, baspin2026}, achieving parallel surgery with amortized constant time overhead.
These constructions enable a form of in-block batching for product codes, where logical qubits are organized into columns which must perform identical operations. 

In that view, our construction extends these prior results to enable in-block batching with maximum parallelism, for ultra-high-rate codes based on group symmetries. Further, subgroup control allows the batch size to be varied within a single code, giving a tunable trade-off between programmability and overhead. 
Remarkably, we observe a simple scaling in the cost of fault-tolerant logic gates, with the amortized spacetime overhead decreasing approximately as $1 / |H|$. 
While it is clear that structure in the code and logic gates can be exploited to generate very low spacetime overheads, it remains an open question whether this tradeoff between logical programmability and spacetime overhead is fundamental, or merely a reflection of currently available techniques. 
Continued work is necessary to understand whether other structures can support more efficient sharing of error-correction resources without batching, i.e. while preserving full independent control.

\paragraph*{Resource comparison.}
We estimate the resource cost for each logical operation in units of physical-qubit syndrome rounds, and compare it against alternative schemes on high-rate codes with comparable check weight. We emphasize that a full spacetime comparison will require further matching the decoded logical error rates under a common noise model, accounting for the desired logical operation, and including physical implementation details such as the duration of a syndrome round, which we leave for future work.

We consider a perfect matching of the logical orbits under a subgroup of order $18$ for the $[[576,148,18]]$ code. Pairing the eight full subgroup orbits gives $72$ two-body logical Pauli-product measurements. We estimate a parallel measurement footprint of $2520$ physical qubits (including ancillas), and four syndrome-extraction rounds to achieve $d=18$, giving $140$ physical-qubit rounds per logical PPM.

An alternative is to encode registers in separate code blocks with shift automorphisms and transversal gates, as in bicycle-chain architectures~\cite{xu2025batched,ismail2026fast}. The smallest distance-$18$ bicycle-chain instance found in our search is $[[246,6,18]]$, with rate $1/41$. We estimate that with parallel transversal CNOTs, we require $24$ blocks and $11,808$ physical qubits to perform the $72$ logical gates, with roughly $160$ physical-qubit rounds per logical CNOT. The estimated spacetime costs per native logical operation are therefore comparable, but our method significantly reduces the spatial footprint.

We also compare against the parallel surgery gadgets for the $[[300,60,14]]$ mitten code, which measures $30$ independently addressable two-body logical Pauli products in parallel~\cite{bhardwaj2026high}. The gadget requires a total of $3118$ physical qubits; rescaling this footprint to 72 PPMs at fixed qubit cost per PPM gives approximately $7500$ qubits, about $3\times$ our space footprint. Assuming 14 syndrome extraction rounds, this gives approximately $1450$ physical-qubit rounds per PPM, about an order of magnitude above the preceding symmetric PPMs and comparable to the overhead we achieve in the fully addressable batch size $|H|=1$ limit (Fig.~\ref{fig:structure-overhead}(c)).

These estimates suggest that by leveraging structure in the logic gates, group-invariant surgery can achieve a qubit-round cost comparable to transversal gates between qLDPC codes with a smaller spatial footprint. The savings are even more striking when comparing with existing instantiations of arbitrary parallel surgery, highlighting how symmetry shared across simultaneous logical operations can further reduce the cost of protection.

\section{Parallel Magic State Teleportation}
\label{sec:magic-state-teleportation}

Universal quantum computation also requires non-Clifford resource states~\cite{bravyi2005magic}. One promising approach is to prepare individual magic states by cultivation in small code patches~\cite{gidney2024magic}, where postselection suppresses errors, and then transfer them into a high-rate code~\cite{zheng2026logical,bhardwaj2026high,xu2024constant,cain2026shor}.
Parallel transfer is particularly natural for algorithms that apply the same non-Clifford operation across many logical qubits~\cite{xu2025batched}.
However, conventional surgery-based approaches require $O(d)$ rounds for fault-tolerant transfer~\cite{cohen2022low}. Since the cultivated source patches dominate the loading footprint, this duration also contributes substantially to the spacetime cost.

Our construction provides a natural route for efficient loading.
Concretely, let us assume surface code patches are used as the cultivated source~\cite{sahay2025fold}.
Denote the surface code as $S_{\bullet}$, the high-rate code as $C_{\bullet}$, and let $G$ be a symmetry of the high-rate code.
We can then treat a stack of $|G|$ identical surface codes as a large code $S_{\bullet}^{\otimes |G|}$ with a full permutation symmetry, so regardless of what the symmetry $G$ is the stacked code is invariant under it.
Therefore, $G$ is a common symmetry of the joint system.

For chain-map convolution to be applied, it is also crucial that the logical action is symmetric under $G$. If we select a logical orbit in the high-rate data block, the pairwise measurements
\begin{equation}
    \overline Z^{\mathrm{LDPC}}_g
    \overline Z^{\mathrm{patch}}_g,
    \qquad g\in G,
\end{equation}
form a symmetry-invariant family that can be efficiently programmed and protected.
With the source patches prepared in $\ket{T}=T\ket{+}$, these joint measurements, followed by logical $X$ readout of the source patches and outcome-dependent Clifford corrections, implement $T$ gates directly on the logical data qubits in the high-rate block~\cite{zhou2000methodology}.

For resource-state loading, we consider parallel injection of $16$ magic states into a single logical orbit of the $[[256,68,12]]$ CPM-PP code. Table~\ref{tab:pp256-surface-injection-resources} also compares the space and time cost with the LP-transistor and mitten magic-port constructions~\cite{zheng2026logical,bhardwaj2026high}.
These schemes use an intermediate LP code block to mediate the transfer, whereas our example couples the source patches directly to the destination block through surgery bridges. 
Further, our construction does not rely on any product structure and accommodates rotated surface code patches. It also uses fewer interface qubits, reducing the allocated space per loaded state. Temporal convolution further generates a three-round loading schedule with sufficient fault-distance, compared with the fourteen-round baseline for the two-stage schemes.

While these estimates are promising and indicate the potential for lower spacetime overhead, in practice additional considerations may arise.
One consideration is that the selected basis for single orbit transfer is distinct from the basis for full-density parallel logic as considered in the rest of the paper.
Another concerns the implementation of temporal convolution, which is most natural when performed in the measurement-based setting and with resource-state preparation. 
We leave a detailed accounting of architecture and circuit design together with circuit-level simulations of the complete transfer process to future work.

\begin{table*}[t]
\caption{Comparison of parallel magic-state loading schemes at source distance seven. Qubit counts are per loaded state and include the source and destination blocks together with the loading interface, excluding syndrome-extraction ancillas and resource preparation. Syndrome rounds count the loading procedure only, excluding preparation and terminal readout. The two-stage injection is estimated at 14 syndrome rounds, corresponding to two fault-tolerant distance 7 surgeries~\cite{zheng2026logical,bhardwaj2026high}. Appendix~\ref{sec:magic-loading-accounting} gives the detailed resource accounting.}
\label{tab:pp256-surface-injection-resources}
\centering
\small
\begin{tabular}{lccc}
\hline\hline
 & Lifted direct injection (this work) & LP transistor~\cite{zheng2026logical} & Mitten magic port~\cite{bhardwaj2026high} \\
\hline
Destination code & $[[256,68,12]]$ & $[[1122,148,\leq20]]$ & $[[200,40,12]]$ \\
States loaded per batch & $16$ & $66$ & $40$ \\
Source patch & Rotated $[[49,1,7]]$ & Unrotated $[[85,1,7]]$ & Unrotated $[[85,1,7]]$ \\
Total allocated qubits per state & $89$ & $136$ & $117$ \\
Syndrome rounds for loading & $3$ & $14$ & $14$ \\
\hline\hline
\end{tabular}
\end{table*}

\section{Discussion and Outlook}
\label{sec:outlook}

We have introduced a framework for parallel logical measurements that uses the symmetries of high-rate quantum codes to achieve efficient logical operations with low spacetime overhead. A symmetry-adapted logical basis enables us to efficiently encode the measurement program, while chain-map convolution enables low-overhead spatial and temporal protection. Subgroup control provides tunable programmability, and parallel magic state transfer provides a low spacetime overhead route to achieving universality.

The highly-structured coupling patterns of the symmetric gadgets may also simplify their physical implementation. Recent proposals for reconfigurable atom arrays and trapped-ion systems use cyclic structure to organize qubit transport and syndrome extraction~\cite{zhao2026towards,tripier2026fault}. Our constructions suggest a way to extend these control strategies to logical operations. The ancilla system starts from a structured subcode and a simple repetition protograph, with convolution adding translated copies of its couplings. We expect this regularity in both the protograph and lift to simplify qubit transport, and we leave more detailed analysis of this to future work.

Our work opens many directions of future research and improvements. At the construction level, the repetition protographs used here are only one choice. Combining graph expansion techniques and randomized heuristics for code surgery~\cite{cross2024improved,williamson2026gauging,ide2025homological,zheng2025high, mehta2026unifying} with group lifts may yield stronger protection with fewer ancilla qubits. 
Similarly, our temporal convolutions connect only neighboring layers, allowing completed layers to be measured out promptly and the corresponding qubits to be recycled. Couplings between more widely separated layers enlarge the design space for spacetime protection~\cite{mehta2026unifying,xu2026framework,xu2025batched} and may further reduce the required temporal depth and spacetime volume, at the cost of a larger instantaneous space cost.

These results also motivate rethinking algorithm design and compilation strategies. Algorithms with repeated subroutines provide a natural setting for organizing logical measurements into symmetry-compatible batches. Determining which computational tasks can make effective use of this structure is an important direction for translating the savings of individual gadgets into lower overhead for complete algorithms.
In cases where such structure is absent, other approaches will likely be necessary to best facilitate low-overhead code surgery. For this reason, in upcoming work \cite{mehta2026unifying}, we develop a complementary toolkit for constructing fast, high-rate surgery gadgets that apply beyond the symmetric settings considered here.

Finally, the observed trade-off between programmability and overhead has interesting connections to recent asymptotic progress on low spacetime overhead computation~\cite{nguyen2025quantum}, and it would be interesting to understand to what extent these trade-offs are fundamental.

\begin{acknowledgments}
We thank J.~Blue, A.~Gu, K.~Kasai, N.~U.~Koyluoglu, J.~Y.~Lee, D.~B.~Tan, C.~Zhao for helpful discussions.
NM was supported by NSF grant PHY-2325080, the U.S. Department of Energy, Office of Science, National Quantum Information Science Research Centers, Quantum Systems Accelerator (Award No. DE-SCL0000121), and by the MIT Center for Theoretical Physics (report number MIT-CTP/6027).
This work was also supported by the National Science Foundation (CCF 2313084), the CUA Physics Frontiers Center (PHY-2317134) and NQVL (PHY-2410716), IARPA and the Army Research Office, under the Entangled Logical Qubits program (W911NF 23-2-0219), the U.S. Department of Energy C2QA program (DE-SC0012704), DARPA MeasQuIT (grant numbers HR0011-26-3-E122), the DARPA HARQ program, DOE Quantum Systems Accelerator Center, grant number DE- AC02- 05CH11231, QuEra Computing, and the MIT-IBM Watson AI Lab.

\textit{Competing interests.} N.M. is an employee and shareholder at Farfield Inc. M.D.L. is a co-founder, shareholder, and Chief Scientist at QuEra Computing. H.Z. is a shareholder of QuEra Computing.

\textit{AI use disclosure.} The authors acknowledge the use of AI tools. The main ideas and numerical experiments of the paper were conceived and designed by the authors. 
The AI coding assistants Codex and Claude Code, with various underlying models, were used to implement and assist with calculations, simulations, figure generation, proof writing, and editing of the manuscript. 
The proof strategies in \cref{app:pk-transverse-products} were developed primarily by AI.
The authors have reviewed all content and take full responsibility of the results.

\textit{Note added.} While completing this work, we became aware of several related but generally complementary results pertaining to improving code surgery, especially in the presence of group structure \cite{berent2026lifted,cowtan2026near,zhou2026parallel,hirasakiUnpublished}.

\end{acknowledgments}

% References follow the main text in two columns, before the appendices.
\makeatletter
\renewcommand{\bibsection}{%
  \begingroup
  \let\addcontentsline\@gobblethree
  \section*{\refname}%
  \endgroup
  \@nobreaktrue
}%
\makeatother
\bibliography{refs}

% ==========================================================================
% The appendices are typeset in a SINGLE column (the main text and references
% above stay two-column).  \onecolumngrid must come before \appendix.
% Under the "preprint" class option the whole document is already one column
% and this command is a harmless no-op.
% The page break also flushes this marker into the .toc before appendix entries.
\addtocontents{toc}{\protect\setcounter{tocdepth}{-1}}
\clearpage
\onecolumngrid
\fi

\appendix
% cleveref names \section after \appendix "Section" by default; alias the
% section counters onto an "appendix" type so cross-references read
% "Appendix C.3" rather than "Section C.3".
\crefname{appendix}{Appendix}{Appendices}
\Crefname{appendix}{Appendix}{Appendices}
\crefalias{section}{appendix}
\crefalias{subsection}{appendix}

% Give readers entering the technical appendices a map of their structure.
% Include subsections but omit the finer proof-level headings.
\ifdefined\appendixonly
  % No main-text entries exist in this job, so the ordinary contents command is
  % sufficient and produces a compact map of the included appendices.
  \setcounter{tocdepth}{2}
  \renewcommand{\tocname}{Appendix contents}
  \tableofcontents
\else
  \setcounter{tocdepth}{2}
  \addtocontents{toc}{\protect\setcounter{tocdepth}{2}}
  \begingroup
  \makeatletter
  % REVTeX does not use tocdepth while formatting its contents file, so hide
  % entries until the \appendix marker and restore only section/subsection rows.
  \let\appendixtocsection\l@section
  \let\appendixtocsubsection\l@subsection
  \let\l@section\@gobble@tw@
  \let\l@subsection\@gobble@tw@
  \let\l@subsubsection\@gobble@tw@
  \def\appendix@toc{%
    \let\l@section\appendixtocsection
    \let\l@subsection\appendixtocsubsection
  }%
  \makeatother
  \renewcommand{\tocname}{Appendix contents}
  \tableofcontents
  \endgroup
  \clearpage
\fi

% Technical appendices begin with Appendix A.

\section{Symmetry and logical coordinates}
\label{app:symmetry-logical-coordinates}

We describe symmetric CSS codes using chain complexes over the group algebra
and introduce symplectic logical coordinates organized into full symmetry
orbits. We obtain the standard form used by our surgery constructions by
identifying a compatible partition of physical qubit orbits. For cyclic
symmetry, we also prove
that a surplus of qubit orbits over check orbits guarantees a corresponding
free logical sector. Appendix~\ref{app:symmetric-logical-measurements} uses
these coordinates to specify logical measurements and their surgery realizations.

\subsection{Symmetric CSS complexes}
\label{sec:symmetric-css-complexes}

A symmetry lets us describe an entire orbit of qubits or checks by one
coordinate over a group algebra. We introduce this description here, keeping
track of the physical basis in which sparsity and distance are measured.

A CSS code is represented by the chain complex
\begin{equation}
  C_\bullet:\qquad
  C_2 \xrightarrow{\ \bdry^C_2=H_X^{\mathsf T}\ } C_1
      \xrightarrow{\ \bdry^C_1=H_Z\ } C_0,
  \qquad
  \bdry^C_1\bdry^C_2=0.
  \label{eq:app-a-base-complex}
\end{equation}
The bases of $C_2$, $C_1$, and $C_0$ label the $X$-checks, physical qubits,
and $Z$-checks, respectively; the chain condition is the CSS commutation
relation~\cite{gottesman1997stabilizer,breuckmann2022ldpc}. We call this a
\emph{based} complex to emphasize these physical coordinates. In our
convention, $H_1(C_\bullet)$ is the $X$-logical space and $H^1(C_\bullet)$
is the $Z$-logical space. Throughout Appendices~A--C, PPM means an $X$-type
logical Pauli-product measurement unless stated otherwise; exchanging
$X$ and $Z$ gives the dual construction.

A finite group $\group$ is a symmetry of this description if it acts by
permutations on qubits and checks of each type and satisfies
\begin{equation}
  \bdry^C_i(g\cdot c)=g\cdot\bdry^C_i(c),
  \qquad g\in\group,\ c\in C_i,\ i\in\{1,2\}.
  \label{eq:boundary-equivariance}
\end{equation}
The boundary maps are then \emph{equivariant}.
Extending the group action linearly makes each $C_i$ a module over
$R=\Ftwo[\group]$: the element $f=\sum_g f_g g$ acts by
$f\cdot c=\sum_g f_g(g\cdot c)$. Equation~\eqref{eq:boundary-equivariance}
then says that the boundary maps are $R$-linear. Thus group-algebra
coefficients describe sums of translated couplings.

We assume that the group acts freely on the physical qubit and check
bases, so every orbit has exactly $|\group|$ distinct elements.
Then, checks and qubits can be divided into orbits, and boundary maps can be represented by matrices over $R$~\cite{breuckmann2021balanced}.
If $b_i$ is the number of orbits in degree $i$, the code becomes
\begin{equation}
  C_\bullet:\qquad
  R^{b_2}\xrightarrow{\ \bdry_2^C\ }R^{b_1}
  \xrightarrow{\ \bdry_1^C\ }R^{b_0}.
  \label{eq:free-module-presentation}
\end{equation}
For cyclic $\group=\langle x\rangle\cong\mathbb Z_P$, we have
$R\cong\Ftwo[x]/(x^P-1)$, and the boundary maps are matrices over this
commutative ring. Multiplication by $x$ shifts one position along an orbit;
for example, $1+x+x^3$ specifies couplings at three relative shifts.
Figure~\ref{fig:symmetry-orbit-compression} illustrates this representation.
Expanding the group-algebra coordinates recovers the original binary
matrices and physical Hamming weights. A general change of module basis
need not preserve those weights.

\begin{figure}[htbp]
\centering
\begin{tikzpicture}
  \def\cA{0}
  \def\cB{5.0}
  \def\cC{10.0}
  \def\rU{2.3}
  \def\rL{-2.3}

  \symstack{QA}{\cB}{\rU}{qcol}
  \symstack{QB}{\cB}{\rL}{qcol}
  \symstack{X}{\cA}{0}{xcol}
  \symstack{Z}{\cC}{0}{zcol}

  \node[symtile=shiftcol, symfront=shiftcol] at (X-1) {};
  \node[symtile=xcol, symfront=xcol] at (X-0) {};
  \node[symtile=qcol, symfront=qcol] at (QA-0) {};
  \node[symtile=qcol, symfront=qcol] at (QB-0) {};

  \foreach \s in {0,1,3}{
    \pgfmathtruncatemacro{\t}{\s+1}
    \draw[symtlnk, shiftcol] (X-1.east) -- (QA-\t.north west);
  }
  \foreach \s in {0,2,5}{
    \pgfmathtruncatemacro{\t}{\s+1}
    \draw[symtlnk, shiftcol] (X-1.east) -- (QB-\t.north west);
  }

  \foreach \s in {0,1}{
    \draw[symlnk, xcol] (X-0.east) -- (QA-\s.north west);
  }
  \draw[symlnk, xcol] (X-0.east) --
    node[symlab, text=xcol, pos=0.5, above=6pt, sloped] {$1+x+x^3$}
    (QA-3.north west);
  \foreach \s/\lb/\p in {0/$1$/0.88, 2/$x^2$/0.88, 5/$x^5$/0.88}{
    \draw[symlnk, xcol] (X-0.east) --
      node[symlab, text=xcol, pos=\p] {\lb} (QB-\s.north west);
  }

  \foreach \s/\lb/\p in {0/$1$/0.55, 2/$x^2$/0.66, 5/$x^5$/0.77}{
    \draw[symlnk, zcol] (QA-0.east) --
      node[symlab, text=zcol, pos=\p] {\lb} (Z-\s.north west);
  }
  \foreach \s in {1,3}{
    \draw[symlnk, zcol] (QB-0.east) -- (Z-\s.north west);
  }
  \draw[symlnk, zcol] (QB-0.east) --
    node[symlab, text=zcol, pos=0.5, below=1pt, sloped] {$1+x+x^3$}
    (Z-0.north west);

  \draw[-{Stealth[length=2.2mm,width=1.6mm]}, axiscol, line width=0.9pt]
    ($(X-0.center)+(-0.80,0.62)$) --
    node[symlab, text=axiscol, sloped, above, inner sep=1.5pt] {$x$}
    ($(X-0.center)+({-0.80+7*\symdx},{0.62+7*\symdy})$);

  \node[symblab, xcol, below=3mm of X-0] {$X$-checks};
  \node[symblab, qcol, above=3mm of QA-7] {qubit orbit $A$};
  \node[symblab, qcol, below=3mm of QB-0] {qubit orbit $B$};
  \node[symblab, zcol, below=3mm of Z-0] {$Z$-checks};

  \def\yC{-5.3}
  \def\yM{-4.1}
  \def\yD{-6.1}
  \def\kX{0.98}
  \def\kQ{5.98}
  \def\kZ{10.98}

  \draw[xcol!45, dotted, line width=0.7pt]
    (\kX,-1.35) -- (\kX,{\yC+0.45});
  \draw[qcol!45, dotted, line width=0.7pt]
    (\kQ,-3.65) -- (\kQ,{\yC+0.45});
  \draw[zcol!45, dotted, line width=0.7pt]
    (\kZ,-1.35) -- (\kZ,{\yC+0.45});

  \node[symmod=xcol] (Rx) at (\kX,\yC) {$R$};
  \node[symmod=qcol] (Rq) at (\kQ,\yC) {$R^2$};
  \node[symmod=zcol] (Rz) at (\kZ,\yC) {$R$};

  \draw[symlnk, black!70] (Rx) --
    node[font=\small, above=1pt] {$\bdry_2$} (Rq);
  \draw[symlnk, black!70] (Rq) --
    node[font=\small, above=1pt] {$\bdry_1$} (Rz);

  \node[font=\small] at ({(\kX+\kQ)/2},\yM)
    {$\bdry_2=\begin{pmatrix}1+x+x^3\\[2pt]1+x^2+x^5\end{pmatrix}$};
  \node[font=\small] at ({(\kQ+\kZ)/2},\yM)
    {$\bdry_1=\begin{pmatrix}1+x^2+x^5&1+x+x^3\end{pmatrix}$};

  \node[font=\scriptsize, text=black!60] at (\kX,\yD)
    {$\dim_{\Ftwo}=8$};
  \node[font=\scriptsize, text=black!60] at (\kQ,\yD)
    {$\dim_{\Ftwo}=16$};
  \node[font=\scriptsize, text=black!60] at (\kZ,\yD)
    {$\dim_{\Ftwo}=8$};
  \node[font=\small, anchor=north] at (\kQ,{\yD-0.55})
    {$R=\Ftwo[\mathbb Z_8]=\Ftwo[x]/(x^8-1)$};
\end{tikzpicture}
\caption{A $\mathbb Z_8$-symmetric CSS complex before and after orbit
compression.  Each stack contains one basis element for every group element,
and multiplication by $x$ moves one step through a stack.  The highlighted
$X$-check and its translate have the same support pattern.  The lower row
represents the same $16$-qubit complex using one $R$-coordinate per orbit.}
\label{fig:symmetry-orbit-compression}
\end{figure}
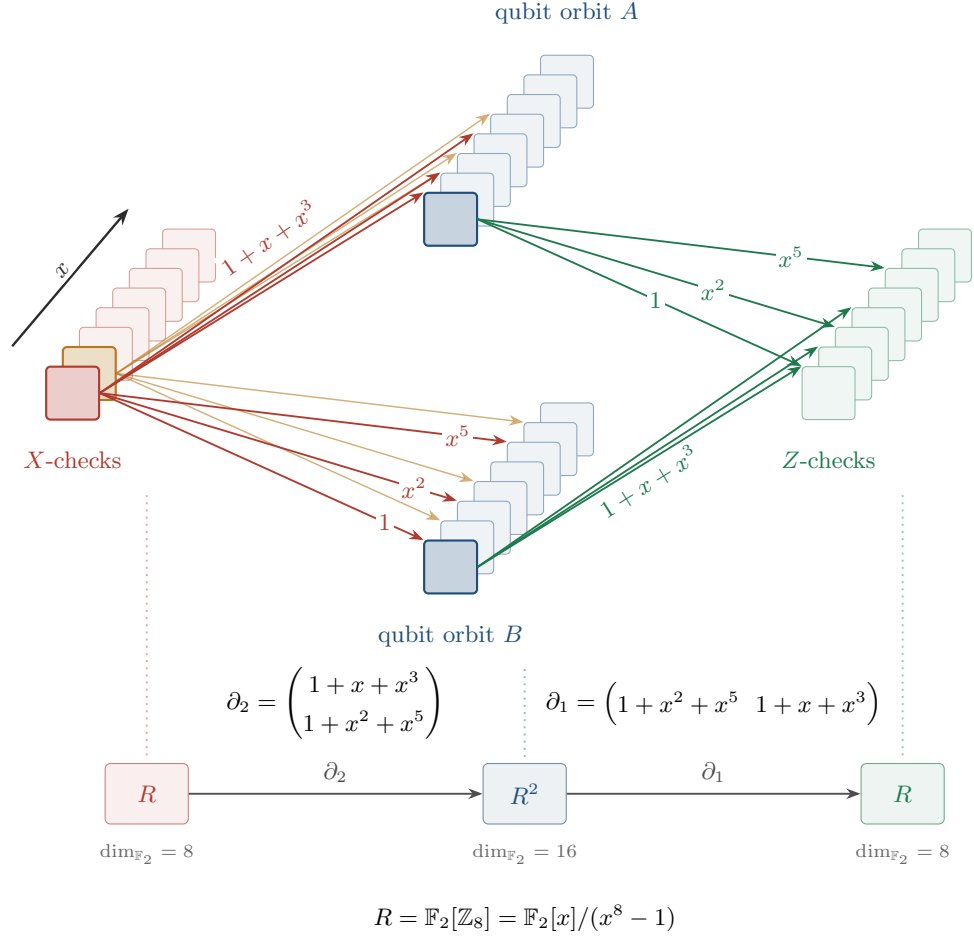

The symmetry also acts on logical operators: it preserves cycles and
stabilizers, so it descends to their quotient
\begin{equation}
  H_1(C_\bullet)=\ker\bdry_1^C/\im\bdry_2^C.
\end{equation}
This logical space is therefore also an $R$-module.
However, it need not be free, even when the physical action is free.
Translating a logical operator may return to the same operator up to
stabilizers before traversing the full physical orbit. More generally,
distinct translates need not be linearly independent as logical operators.
For our principal examples, however, almost all logical qubits can be
organized into full, independent logical orbits. The next subsection
develops coordinates for these orbits, which we then use to program
symmetric logical measurements.

\par\medskip
\subsection{Symmetry-adapted logical basis}
\label{sec:symmetry-adapted-logical-coordinates}

For the cyclic symmetries used below, a small set of logical seeds describes
an entire family of translated logical qubits. Write
\begin{equation}
  \group=\langle x\rangle\cong\mathbb Z_P,
  \qquad R=\Ftwo[x]/(x^P-1),
  \label{eq:cyclic-group-algebra}
\end{equation}
and choose paired $X$- and $Z$-type logical seeds with translates
\begin{equation}
  X_{i,j}=x^j\cdot X_{i,0},
  \qquad
  Z_{i,j}=x^j\cdot Z_{i,0},
  \qquad i=1,\ldots,\ell,\quad j\in\mathbb Z_P.
  \label{eq:free-logical-coordinates}
\end{equation}
Here $i$ labels an orbit and $j$ a position within it; the dot denotes
translation by the physical qubit permutation. We require the symplectic
pairing
\begin{equation}
  \bigl\langle X_{i,j},
  Z_{i^{\prime},k}\bigr\rangle
  =\delta_{i i^{\prime}}\delta_{jk}.
  \label{eq:paired-logical-orbit-commutation}
\end{equation}
Thus each $X$ operator anticommutes with its paired $Z$ and commutes with
all the others.
Translation preserves this pairing, so it suffices to check each $X$ seed
against the translates of the $Z$ seeds.

We call these operators a \emph{symmetry-adapted symplectic logical basis}
for the sector they span. 
Each seed generates a \emph{full logical orbit}
of $P$ independent translates, and each Pauli sector is isomorphic to
$R^\ell$. These orbits account for $\ell P$ logical qubits.
We take these paired full orbits as the \emph{active} logical subsystem
and their complement as the \emph{residual}
subsystem, writing
\begin{equation}
  H_1(C_\bullet)=\mathcal L_X\oplus\mathcal R_X,
  \qquad
  H^1(C_\bullet)=\mathcal L_Z\oplus\mathcal R_Z.
  \label{eq:active-residual-logical-sectors}
\end{equation}
Here $\mathcal L_X,\mathcal L_Z$ are the active sectors and
$\mathcal R_X,\mathcal R_Z$ are the residual sectors: the intended operation may
include additional measurements confined to these sectors, while preserving
unmeasured active logicals.

\paragraph*{Symmetric standard form.}

Our surgery gadgets use a standard form with explicit information and pivot
coordinates. This requires a common partition of the physical qubit orbits,
\begin{equation}
  \{1,\ldots,b_1\}=S_X\sqcup S_Z\sqcup L,
  \label{eq:compatible-pivot-partition}
\end{equation}
into $X$-stabilizer pivots, $Z$-stabilizer pivots, and information orbits.
This extends the information/pivot decomposition of Gottesman's binary
stabilizer standard form~\cite{gottesman1997stabilizer} while keeping whole
symmetry orbits together. For a compatible partition, restriction to $S_X$
respectively $S_Z$ detects every nonzero stabilizer of that type, and the
logical-cycle equations admit representatives
\begin{equation}
  G_X^{\mathrm{free}}=\begin{pmatrix}0&C_X&I_\ell\end{pmatrix},
  \qquad
  G_Z^{\mathrm{free}}=\begin{pmatrix}C_Z&0&I_\ell\end{pmatrix},
  \qquad \ell=|L|,
  \label{eq:symmetry-adapted-standard-form}
\end{equation}
with columns ordered as $(S_X\mid S_Z\mid L)$. Each row is a seed: it has
one information coordinate and a completion on the opposite stabilizer-pivot
orbits.
To construct the standard form, we start by identifying a set of valid pivots, and then solve for the $C_X$ and $C_Z$ completions via back-substitution.

\par\medskip
\subsection{Existence of a large free logical sector}
\label{sec:large-free-logical-sector}

Remarkably for cyclic symmetries, the surplus of qubit orbits over check orbits guarantees a corresponding number of full logical orbits.
All of the ultra-high-rate codes we study in this work achieve their high-rate via this mechanism, and the following proposition partially explains why we are also able to find symmetry-adapted symplectic logical bases for them.

\begin{proposition}[Free logical orbits from a check deficit]
\label{prop:cyclic-free-logical-sector}
Suppose $\mathbb Z_P$ acts freely on the qubits and listed checks of a CSS
code, with $b_1$ qubit orbits and $b_2,b_0$ check orbits of $X,Z$ type.
If $\ell_0=b_1-b_2-b_0\ge0$, the code admits at least $\ell_0$ paired
full logical orbits, accounting for $\ell_0P=n-m_X-m_Z$ logical qubits.
The listed checks may be dependent, and $P$ may be even.
\end{proposition}

\begin{proof}
Matrices over $R=\Ftwo[x]/(x^P-1)$ admit invertible basis changes that
diagonalize them: apply Smith reduction over $\Ftwo[x]$ and reduce modulo
$x^P-1$~\cite{nicholson2022algebra}.
First diagonalize $\bdry_1^C$. Since it has $b_0$ rows, at least $b_1-b_0$
middle coordinates have zero boundary. Now diagonalize the corresponding
rows of $\bdry_2^C$, changing basis only within this block and among the
$X$ checks. Since $\bdry_2^C$ has $b_2$ columns, at least
$b_1-b_0-b_2=\ell_0$ coordinates have neither incoming nor outgoing
boundaries. Each is a free copy of $R$, giving $P$ independent
$X$ logicals. Their binary dual coordinate vectors give the paired
$Z$ logicals, and translation cyclically permutes both bases.
\end{proof}

\ifdefined\appendixaonly
\else

\section{Symmetric logical measurements and surgery}
\label{app:symmetric-logical-measurements}

Using the logical coordinates of Appendix~\ref{app:symmetry-logical-coordinates},
we specify symmetric families of Pauli-product measurements and describe
how an auxiliary complex realizes them by surgery. We identify which logical
information the merged code measures and which survives. We then state the
sparsity, overhead, and distance requirements that guide the constructions in
Appendix~\ref{app:sparse-surgery-constructions}. All content in this section are
definitions and background for the constructions in the following appendices,
and are not original to this work.

\subsection{Symmetric logical Pauli-product measurements}
\label{sec:symmetric-logical-ppm}

We express the seed measurements of the main text in the logical coordinates
of Appendix~\ref{app:symmetry-logical-coordinates}.

\begin{definition}[Target measured logical subspace]
\label{def:measured-logical-subspace}
Let $Q\in R^{\ell\times q}$ be a matrix.
Let column $\alpha$ of $Q$ contain the coordinates of the $X$-type seed
$M_\alpha$ on the $\ell$ active logical orbits. For $q$ seeds, the
\emph{target measured logical subspace} is
\begin{equation}
  M:=\operatorname{im}_R Q
  \leq\mathcal L_X\cong R^\ell.
  \label{eq:logical-measurement-presentation}
\end{equation}
\end{definition}

Taking the column span over $R$ includes all translations and products of the seeds.
The number of independent target measurements is $\dim_{\Ftwo}M$; redundant seeds
or translates do not add independent measurements.
For measurements between multiple code blocks, take
$C_\bullet=\bigoplus_\nu C_\bullet^{(\nu)}$ with a common cyclic symmetry
action, and allow columns of $Q$ to involve logical orbits from different blocks.

Mixed-type products between blocks, such as logical $X\otimes Z$, fit the
same framework by choosing a fixed Pauli type in each block for all seeds.
For a block participating through $Z$ operators, exchange its $X$- and
$Z$-checks and work with the dual CSS complex $C_\bullet^\vee$: its homology
$H_1(C_\bullet^\vee)\cong H^1(C_\bullet)$ describes the original $Z$ logicals.
Taking the direct sum of the chosen original and dual complexes therefore
reduces the mixed-type measurement to the $X$-type construction above, whose
operators are then interpreted in the original Pauli labels. All translated
products commute because every block uses a common Pauli type.

\par\medskip
\subsection{Auxiliary complex and the merged code}
\label{sec:auxiliary-complex-mapping-cone}

This section explains what conditions an auxiliary complex and merged code should satisfy
to realize a given target measured logical subspace $M$.
We work within the formalism of quantum code surgery.
Adding physical representatives of the desired logical operators as new
$X$-checks directly implements the measurement~\cite{zheng2025high,cowtan2025parallel},
but these checks can have large weight. To obtain sparse checks, we introduce
ancillas through the surgery construction of
Refs.~\cite{cowtan2024css,ide2025homological,cowtan2025parallel}.

\paragraph*{The merged code.}
Attach an auxiliary complex over $R$,
\begin{equation}
  E_\bullet:\qquad
  E_1\xrightarrow{\ \bdry^E_1\ }E_0
  \xrightarrow{\ \bdry^E_0\ }E_{-1},
  \qquad
  \bdry^E_0\bdry^E_1=0.
  \label{eq:auxiliary-complex}
\end{equation}
Here $E_1$ labels new $X$-checks, $E_0$ ancilla qubits, and $E_{-1}$ new
$Z$-checks. The coupling $\gamma_1:E_1\to C_1$ specifies the support of the
new $X$-checks on the base qubits, while $\gamma_0:E_0\to C_0$ extends the
old $Z$-checks onto the ancillas. 
In main text we have the special case
$E_\bullet=P_\bullet$ with $E_{-1}=0$ and $\bdry_1^E=\partial_P$;
$E_\bullet$ also describes the enlarged auxiliaries obtained by thickening.

\begin{figure}[htbp]
\centering
\begin{tikzpicture}[cxdiag]
  \node[base] (C2) {$C_2$};
  \node[base, right=of C2] (C1) {$C_1$};
  \node[base, right=of C1] (C0) {$C_0$};

  \node[anc, below=of C2] (E1) {$E_1$};
  \node[anc, right=of E1] (E0) {$E_0$};
  \node[anc, right=of E0] (Em1) {$E_{-1}$};

  \node[draw=black!28, densely dashed, rounded corners=3pt,
        fit=(C2)(E1), inner sep=5pt,
        label={[font=\small]below:$F_2=C_2\oplus E_1$}] (F2box) {};
  \node[draw=black!28, densely dashed, rounded corners=3pt,
        fit=(C1)(E0), inner sep=5pt,
        label={[font=\small]below:$F_1=C_1\oplus E_0$}] (F1box) {};
  \node[draw=black!28, densely dashed, rounded corners=3pt,
        fit=(C0)(Em1), inner sep=5pt,
        label={[font=\small]below:$F_0=C_0\oplus E_{-1}$}] (F0box) {};

  \draw[ar, basecol] (C2) -- node[lab, above] {$\bdry^C_2$} (C1);
  \draw[ar, basecol] (C1) -- node[lab, above] {$\bdry^C_1$} (C0);
  \draw[ar, anccol] (E1) -- node[lab, above] {$\bdry^E_1$} (E0);
  \draw[ar, anccol] (E0) -- node[lab, above] {$\bdry^E_0$} (Em1);
  \draw[ar, mapcol] (E1) -- node[diaglab] {$\gamma_1$} (C1.south west);
  \draw[ar, mapcol] (E0) -- node[diaglab] {$\gamma_0$} (C0.south west);

  \node[rowlab, left=8mm of C2, basecol] {base code};
  \node[rowlab, left=8mm of E1, anccol] {auxiliary complex};
\end{tikzpicture}
\caption{Constructing the merged code by attaching an auxiliary complex. Blue nodes
belong to the base code, red nodes to the auxiliary complex, and green arrows
give the coupling map. Each dashed column is one space of the merged code's
CSS complex.
Thus $E_1$ labels new $X$-checks, $E_0$ labels ancilla qubits, and
$E_{-1}$ labels new $Z$-checks.}
\label{fig:auxiliary-mapping-cone}
\end{figure}
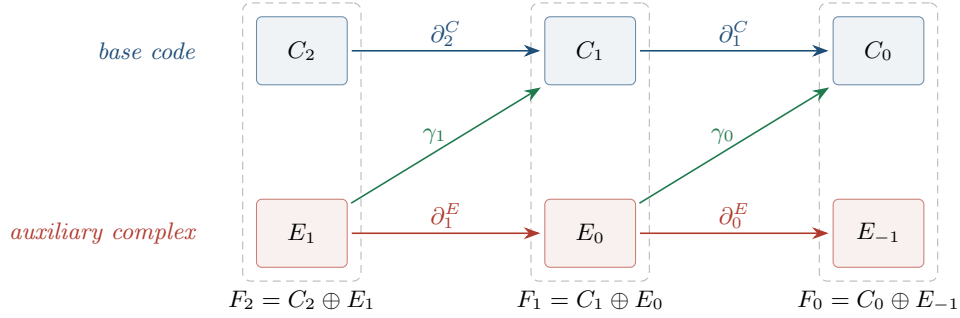

Grouping the spaces as in \cref{fig:auxiliary-mapping-cone} gives the merged
code, represented by the mapping cone $F_\bullet$ with boundary maps
\begin{equation}
  \bdry^F_2=
  \begin{pmatrix}
    \bdry^C_2&\gamma_1\\
    0&\bdry^E_1
  \end{pmatrix},
  \qquad
  \bdry^F_1=
  \begin{pmatrix}
    \bdry^C_1&\gamma_0\\
    0&\bdry^E_0
  \end{pmatrix}.
  \label{eq:merged-complex-boundaries}
\end{equation}
The checks commute precisely when the coupling satisfies the chain-map condition
\begin{equation}
  \bdry^C_1\gamma_1=\gamma_0\bdry^E_1.
  \label{eq:auxiliary-chain-map-condition}
\end{equation}
Taking $\bdry_1^C=H_Z$ and $\bdry_1^E=\partial_P$ recovers the
main-text condition $H_Z\gamma_1=\gamma_0\partial_P$. 
Working over $R$ automatically preserves the symmetry of the base code.

\FloatBarrier
\paragraph*{The logical action.}
To assess logical correctness, we must identify what the construction measures
and which logical information survives. The following criterion imposes the
target measurements, permits additional measurements confined to the residual
sector, and excludes new logical qubits introduced by the ancillas.
Write $\gamma_{*,j}:H_j(E_\bullet)\to H_j(C_\bullet)$ for the induced
maps on homology. Physically, for $j=1$, a cycle $z\in E_1$ specifies a
product of new $X$-checks whose ancilla support cancels
($\partial_1^E z=0$). The map $\gamma_{*,1}$ returns the base-code logical
class $[\gamma_1z]$ measured by that product; its image
$J=\im\gamma_{*,1}$ is therefore the total measured subspace.

For $j=0$, a class $[a]\in H_0(E_\bullet)$ represents an ancilla $X$ operator
that commutes with the new $Z$-checks labeled by $E_{-1}$
($\partial_0^E a=0$), modulo the ancilla supports of new $X$-checks.
These are only one set of merged-code $Z$-checks: the extended base-code
checks, labeled by $C_0$, also touch the ancillas and can have nonzero syndrome.
The full syndrome of the ancilla-only operator is
\[
  \partial_1^F(0,a)=(\gamma_0a,\partial_0^E a)=(\gamma_0a,0)
  \in C_0\oplus E_{-1}.
\]
The map $\gamma_{*,0}[a]=[\gamma_0a]$ records the first component modulo
syndromes producible by $X$ operators on the base qubits. Thus
$[\gamma_0a]=0$ means that there exists a base-qubit $X$ operator $x\in C_1$
with $\partial_1^C x=\gamma_0a$. Its support need not be local or low-weight.

\begin{definition}[Logical correctness]
\label{def:logical-correctness}
We call the realization \emph{logically faithful} to the active target $M$ if
the following two requirements hold:
\begin{equation}
  \begin{array}{@{}r@{\quad}l@{}c@{\quad}l@{}}
    \text{(1)} & \text{Measured subspace} & :
      &J=M\oplus M_{\mathrm{res}}
        \quad\text{for some }M_{\mathrm{res}}\subseteq\mathcal R_X.\\
    \text{(2)} & \text{No new logical qubits} & :
      &\gamma_{*,0}\text{ is injective}.
  \end{array}
  \label{eq:auxiliary-realizes-measurement}
\end{equation}
\end{definition}

Requirement (1) ensures that every target outcome can be recovered by
multiplying outcomes of the new $X$-check measurements, accounting for any
known base-code stabilizer signs, and that additional measurements act only
on residual logicals.
In particular, measuring an active operator multiplied by an unknown residual
operator is insufficient unless the active outcome can be recovered from the
full set of measured products.

Requirement (2) ensures that every merged $X$ logical is equivalent to one
supported entirely on the base qubits. To see this, let $(x,a)$ represent a
merged logical. Its zero syndrome gives
\[
  \partial_1^C x+\gamma_0a=0,
  \qquad \partial_0^E a=0,
\]
so $[a]\in\ker\gamma_{*,0}$. If $\gamma_{*,0}$ is injective, then $[a]=0$,
meaning $a=\partial_1^E z$ for some $z\in E_1$. Multiplying by the new
$X$-checks specified by $z$ adds the stabilizer
$(\gamma_1z,\partial_1^E z)$ and leaves $(x+\gamma_1z,0)$, removing all
ancilla support.

Conversely, if requirement (2) fails, choose a nonzero
$[a]\in\ker\gamma_{*,0}$. As explained above, a base-qubit operator $x$
can cancel its syndrome, so $(x,a)$ commutes with all merged-code checks. But $[a]\ne0$
means $a\notin\im\partial_1^E$, whereas every $X$ stabilizer has ancilla
support in $\im\partial_1^E$. Thus no stabilizer product can remove $a$:
this logical is additional to those inherited from the base code.

Under the criterion, all merged $X$ logicals are therefore inherited from the
base code, and precisely the $X$ logicals in $J$ become stabilizers.
The surviving $X$-logical space is the quotient by $J$; its binary dual is
the space of base-code $Z$ logicals commuting with $J$. Hence
\begin{equation}
  \begin{aligned}
    H_1(F_\bullet)&\cong H_1(C_\bullet)/J
      \cong(\mathcal L_X/M)\oplus(\mathcal R_X/M_{\mathrm{res}}),\\
    H^1(F_\bullet)&\cong J^\perp.
  \end{aligned}
  \label{eq:logically-faithful-realization}
\end{equation}
Here $J^\perp$ consists of all $Z$ logicals commuting with the total measured
subspace, including the residual measurements. Their representatives in the
merged code may require ancilla support to commute with the individual new
$X$-checks; restricting back to the base qubits identifies their original
logical classes. For example, measuring $X_1X_2$ on two logical qubits makes
$X_1$ and $X_2$ equivalent and leaves $Z_1Z_2$ as their surviving $Z$ partner.
These conditions are sufficient and simplify analysis, and the constructions 
we consider satisfy them. Other surgery protocols have been developed which do not 
require criteria (2)~\cite{cross2024improved, ide2025homological, zheng2025high}. 

\par\medskip
\subsection{Sparsity, overhead, and distance}
\label{sec:sparsity-overhead-distance}

Logical faithfulness specifies which logical information survives the merged
code, but it does not control the physical cost of the realization.  A useful
surgery gadget must also have sparse checks, bounded qubit degree, modest
ancilla overhead, and sufficient distance.  We now state these quantitative
requirements.  They will serve as the targets for the constructions in the
following appendices.

\paragraph*{Sparsity.}
We report maximum check weights and qubit degrees separately for $X$ and $Z$:
the number of qubits in a check and the number of checks of the specified type
incident on a qubit, respectively. Weights are counted after binary expansion.
For an LDPC family, these quantities remain bounded
as the code grows. By \cref{eq:merged-complex-boundaries}, sparse auxiliary
and coupling maps preserve this property when the base code is LDPC.

\paragraph*{Ancilla overhead.}
Let $n_C=\dim_{\Ftwo}C_1$ and $n_{\mathrm{anc}}=\dim_{\Ftwo}E_0$.
The relative overhead in code qubits and the number of independent target
measurements per added code qubit are
\begin{equation}
  \eta_q:=\frac{n_{\mathrm{anc}}}{n_C},
  \qquad
  \rho_{\mathrm{meas}}:=\frac{\dim_{\Ftwo}M}{n_{\mathrm{anc}}}.
  \label{eq:ancilla-overhead-and-measurement-rate}
\end{equation}
We also count one syndrome-extraction ancilla per measured check, so that
both added code qubits and added checks contribute to the physical cost.
With $N_C=\sum_{j=0}^2\dim_{\Ftwo}C_j$ the corresponding base-code cost,
\begin{equation}
  N_{\mathrm{anc}}:=\sum_{j=-1}^{1}\dim_{\Ftwo}E_j,
  \qquad
  \eta_{\mathrm{tot}}:=\frac{N_{\mathrm{anc}}}{N_C}.
  \label{eq:total-ancilla-overhead}
\end{equation}
Thus the total qubit count including check ancillas is
$N_F=N_C+N_{\mathrm{anc}}$.
Check counts include redundant generators when they are measured, whereas
$\dim_{\Ftwo}M$ counts only independent target measurements.

\paragraph*{Distance.}
The static stabilizer-code distances $d_X(F_\bullet)$ and $d_Z(F_\bullet)$
are the minimum physical weights of nontrivial $X$ and $Z$ logical operators;
$d(F_\bullet)=\min\{d_X(F_\bullet),d_Z(F_\bullet)\}$.
These minima include all surviving active and residual logicals.
For a logically faithful realization, every nontrivial merged $Z$ logical
restricts to a nontrivial base-code logical with no greater weight, giving
\begin{equation}
  d_Z(F_\bullet)\geq d_Z(C_\bullet)
  \label{eq:automatic-z-distance-bound}
\end{equation}
whenever logical qubits survive. However, the $X$-distance requires additional
protection because the ancillas can provide lighter representatives of
surviving logicals.
When evaluating our constructions we check the $X$-distance of the merged code numerically.
Appendix~\ref{app:sparse-surgery-constructions} develops
constructions that improve this distance while controlling sparsity and
overhead. Measurement rounds and fault distance are treated separately in
Appendix~\ref{app:efficient-fault-tolerance}.

\section{Sparse surgery constructions and convolution}
\label{app:sparse-surgery-constructions}

Once the target logical measurement as been specified (Appendix~\ref{app:symmetric-logical-measurements}), the next step is to construct a sparse auxiliary complex that facilitates the operation.
We start by programming the measurements in a starter complex, and then protect the surviving
logical information by repetition thickening and chain-map convolution.
We give their physical maps and resource counts, a sufficient condition for
transferring the surviving base-code distance to the merged code, and
explicit upper bounds on the distance achievable by scalar convolution.
The asymptotic two-factor construction and its squared-logarithmic spatial
overhead guarantee are given in Appendix~\ref{si-pk:sec:two-page-proof}.

The starter and thickening constructions build on homological measurement
and qLDPC surgery~\cite{cowtan2024css,ide2025homological,cross2024improved,
cowtan2025parallel,zheng2025high}. Our main additional ingredient is
\emph{chain-map convolution}: replacing interlayer identity maps by sums
of symmetry-translated couplings, inspired by code lifting~\cite{panteleev2022almostlinear,
panteleev2022asymptotically,guemard2024lifts,hirasaki2026lifting}.
The chain-map condition ensures commuting checks. Logical correctness and distance require the additional conditions stated below.

\subsection{Sparse starter complexes}
\label{sec:sparse-starter-complexes}

We construct a two-term starter complex
$P_1\xrightarrow{\partial_P}P_0$, with $P_{-1}=0$, and coupling maps
$\gamma_1:P_1\to C_1$ and $\gamma_0:P_0\to C_0$ satisfying the logical
correctness criterion of \cref{def:logical-correctness}, before we add the layers that
protect their distance. Physically, $P_1$ labels new $X$-checks
and $P_0$ labels added qubits. 
We first give the more efficient direct programming
construction illustrated in Fig.~\ref{fig:construction-pipeline}, then describe information-bit merging as a route for programming larger products while maintaining sparsity.

\subsubsection{Direct information-bit programming}
\label{sec:direct-information-bit-programming}
Work over $R=\mathbb F_2[x]/(x^P-1)$ and use the compatible standard-form
partition from \cref{eq:compatible-pivot-partition},
\[
 C_1=R^{S_X}\oplus R^{S_Z}\oplus R^L,
 \qquad
 H_Z=\begin{pmatrix}H_{Z,S_X}&H_{Z,S_Z}&H_{Z,L}\end{pmatrix}.
\]
The columns of $Q\in R^{\ell\times q}$ specify the measurement seeds in
the active logical coordinates, so that $M=\im_R Q$. Introduce one
information-check orbit per column of $Q$, retain all $Z$-pivot orbits,
and copy the base-code $Z$-check coordinates into $P_0$:
\begin{equation}
 \begin{gathered}
 P_1=R^{S_Z}\oplus R^q,\qquad P_0=C_0,\\[2pt]
 \gamma_1=
 \begin{pmatrix}0&0\\ I&0\\0&Q\end{pmatrix},
 \qquad \gamma_0=I,
 \qquad
 \partial_P=\begin{pmatrix}H_{Z,S_Z}&H_{Z,L}Q\end{pmatrix}.
 \end{gathered}
 \label{eq:starter-direct-matrices}
\end{equation}
Thus pivot checks couple to their matching data coordinates, while each
program check couples to the information coordinates selected by its column
of $Q$. The ancilla support of every new check is its base-code $Z$ syndrome:
$\partial_P=H_Z\gamma_1$. This ensures that the new $X$-checks commute
with the extended base-code $Z$-checks.

For example, the single column
\begin{equation}
 Q=\mathbf v=e_a+x^s e_b+x^t e_c
 \label{eq:starter-running-triple}
\end{equation}
programs the translated products $X_{a,j}X_{b,j+s}X_{c,j+t}$.
Its information-check coordinate couples to three data qubits per translate;
the pivot checks supply the remaining support when these checks are
multiplied to recover the logical measurement.

\paragraph*{Logical correctness.}
A product of new checks is specified by $(p,u)\in R^{S_Z}\oplus R^q$.
Its ancilla support cancels precisely when
\begin{equation}
 H_{Z,S_Z}p+H_{Z,L}Qu=0.
 \label{eq:starter-cycle-condition}
\end{equation}
Standard form provides a pivot completion $p$ for every information vector
$Qu$. The resulting base operator is $(0,p,Qu)$, so every target product
in $M$ is measured. Conversely, every such product has active logical
coordinates $Qu$. Choosing the completion from the active logical basis
gives the target operator itself; any difference between completions has
support only on $S_Z$ and hence has zero pairing with every active
$Z$ logical. Additional measured operators therefore lie entirely in the
residual sector, giving $J=M\oplus M_{\mathrm{res}}$ as required in
\cref{def:logical-correctness}.

The same pivot coordinates also prevent new logical qubits from being
introduced by the ancillas. Compatibility of the pivot partition implies
that the columns of $H_{Z,S_Z}$ span the full syndrome image of $H_Z$.
Since these columns are retained in $\partial_P$,
\begin{equation}
 \im\partial_P=\im H_Z.
 \label{eq:starter-syndrome-image}
\end{equation}
With $\gamma_0=I$, the induced map
$C_0/\im\partial_P\to C_0/\im H_Z$ is therefore injective, which is
the second logical correctness requirement. This argument allows dependent
base-code checks; it does not require $\partial_P$ to be surjective.

\paragraph*{Sparsity and overhead.}
The direct starter adds $\dim_{\mathbb F_2}C_0$ code qubits and
$P(|S_Z|+q)$ listed $X$-checks, with no new $Z$-checks. These counts can
be combined using the total qubit-plus-check convention of \cref{sec:sparsity-overhead-distance}.
After binary expansion, a column weight of $Q$ counts the logical factors
in a seed, while a row weight counts how many translated seeds involve a
given information qubit. Both must remain bounded to keep the interface
sparse. Together with bounded base-code check weights and qubit degrees,
these bounds also keep $\partial_P$ sparse. In particular, the construction
uses the sparse base-code checks without placing the potentially dense
logical completion matrices into the coupling maps.
However, every logical factor in a seed adds a data coupling to the same
program check. A cap on that check's degree (its physical check weight)
therefore limits the weight of the logical PPM that can be programmed
directly, with its ancilla support consuming part of the same budget.
This motivates information-bit merging, which distributes a larger product
across several checks connected by pairwise relations.

\subsubsection{Information-bit merging}
\label{sec:information-bit-merging}

In order to reach higher-weight logical pauli products, while maintaining sparsity, we need an alternate approach.
In this section, we show how to keep separate check coordinates for the factors in a
product and connect them by pairwise relations. For the triple in
Eq.~\eqref{eq:starter-running-triple}, let $w=(u_a,u_b,u_c)$ and set
\begin{equation}
 B=\begin{pmatrix}x^s&1&0\\0&x^{t-s}&1\end{pmatrix},
 \qquad
 Bw=0\ \Longleftrightarrow\ w=(u,x^su,x^tu).
 \label{eq:starter-triple-merge-boundary}
\end{equation}
Figure~\ref{fig:information-bit-merging} shows how the two relations distribute
the product across three checks.
\begin{figure}[htbp]
\centering
\begin{tikzpicture}[
  font=\small,
  data/.style={circle,draw=blue!65!black,fill=blue!7,minimum size=10mm},
  check/.style={rectangle,draw=black!75,fill=black!5,minimum size=10mm},
  ancilla/.style={circle,draw=violet!75!black,fill=violet!10,minimum size=9mm},
  block/.style={rectangle,rounded corners=2pt,draw,align=center,
    minimum width=23mm,minimum height=10mm,font=\footnotesize},
  coupling/.style={thick,draw=blue!65!black},
  syndrome/.style={thin,draw=black!40},
  relation/.style={thick,draw=violet!75!black},
  edgelabel/.style={fill=white,inner sep=1pt,font=\footnotesize}
]
\draw[rounded corners=3pt,draw=blue!25,fill=blue!2]
  (-0.65,1.15) rectangle (9.4,2.4);
\draw[rounded corners=3pt,draw=black!25,fill=black!2]
  (-0.65,-0.65) rectangle (9.4,0.65);
\draw[rounded corners=3pt,draw=violet!30,fill=violet!3]
  (-0.65,-2.9) rectangle (9.4,-1.4);
\node[anchor=east,align=right] at (-0.85,1.8)
  {$C_1$\\[-1pt]\scriptsize Data qubits};
\node[anchor=east,align=right] at (-0.85,0)
  {$P_1$\\[-1pt]\scriptsize New $X$-checks};
\node[anchor=east,align=right] at (-0.85,-2.1)
  {$P_0$\\[-1pt]\scriptsize Added qubits};
\node[data] (a) at (0,1.8) {$a$};
\node[data] (b) at (3,1.8) {$b$};
\node[data] (c) at (6,1.8) {$c$};
\node[block,draw=blue!65!black,fill=blue!7] (sz) at (8,1.8)
  {$Z$-pivot qubits\\$R^{S_Z}$};
\node[check] (ua) at (0,0) {$u_a$};
\node[check] (ub) at (3,0) {$u_b$};
\node[check] (uc) at (6,0) {$u_c$};
\node[block,fill=black!5] (piv) at (8,0)
  {Pivot checks\\$R^{S_Z}$};
\node[ancilla] (rab) at (1.5,-2) {$r_1$};
\node[ancilla] (rbc) at (4.5,-2) {$r_2$};
\node[block,draw=violet!75!black,fill=violet!10] (c0) at (8,-2.1)
  {Copied check\\coordinates $C_0$};
\draw[syndrome] (ua.south) -- (c0.north west);
\draw[syndrome] (ub.south) -- (c0.north west);
\draw[syndrome] (uc.south) -- (c0.north west);
\draw[syndrome] (piv.south) -- (c0.north);
\draw[coupling] (ua) -- node[left] {$1$} (a);
\draw[coupling] (ub) -- node[left] {$1$} (b);
\draw[coupling] (uc) -- node[left] {$1$} (c);
\draw[coupling] (piv) -- node[left] {$I$} (sz);
\draw[relation] (ua) -- node[edgelabel,pos=0.55] {$x^s$} (rab);
\draw[relation] (ub) -- node[edgelabel,pos=0.55] {$1$} (rab);
\draw[relation] (ub) -- node[edgelabel,pos=0.55] {$x^{t-s}$} (rbc);
\draw[relation] (uc) -- node[edgelabel,pos=0.55] {$1$} (rbc);
\node[font=\footnotesize] at (1.5,-2.65) {$u_b=x^su_a$};
\node[font=\footnotesize] at (4.5,-2.65) {$u_c=x^{t-s}u_b$};
\node[align=center] at (4.4,-3.55)
  {$P_0=C_0\oplus Rr_1\oplus Rr_2$\\[4pt]
   Measured product: $X_{a,j}X_{b,j+s}X_{c,j+t}$};
\end{tikzpicture}
\caption{Information-bit merging, grouped by chain degree. The $P_1$ row
contains the new $X$-checks; the $P_0$ row contains the added qubits.
Merging enlarges $P_0$ beyond a copy of $C_0$ by adding the two
relation-qubit orbits $r_1,r_2$. Their incident edges impose the relations
in $B$, while the upper edges give the data coupling $\gamma_1$.
Circles and squares represent individual qubit and check orbits;
rounded boxes collect the indicated coordinates. Gray lines schematically
show the couplings to the copied $C_0$ coordinates: $H_{Z,L}T$ for the
information checks and $H_{Z,S_Z}$ for the pivot checks. Each gray line
represents the corresponding syndrome support, rather than a single edge
of unit weight. The untouched $S_X$ data sector is omitted.}
\label{fig:information-bit-merging}
\end{figure}
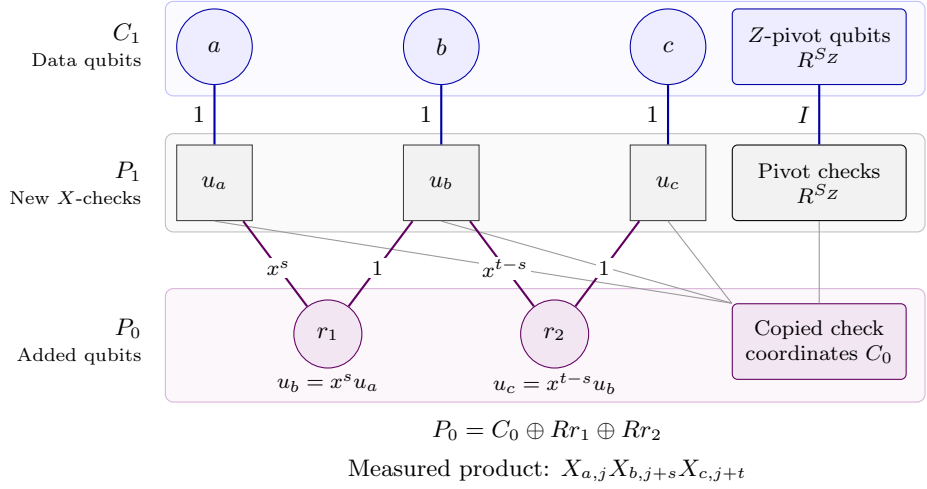
\par
Let $T:R^3\to R^L$ place these coordinates on information orbits $a,b,c$.
The full starter is
\begin{equation}
 \begin{gathered}
 P_1=R^{S_Z}\oplus R^3,\qquad P_0=C_0\oplus R^2,\\[2pt]
 \gamma_1(p,w)=(0,p,Tw),\qquad \gamma_0(a,b)=a,\\[2pt]
 \partial_P=
 \begin{pmatrix}
 H_{Z,S_Z}&H_{Z,L}T\\
 0&B
 \end{pmatrix}.
 \end{gathered}
 \label{eq:starter-merging-matrices}
\end{equation}
Thus the two relation-qubit orbits enlarge $P_0$ from $C_0$ to
$C_0\oplus R^2$ and impose the pairwise relations. A cycle must
satisfy both $Bw=0$ and the base-code cycle equation, so it measures the
same triple, again allowing only additional residual measurements.
The chain condition follows directly from Eq.~\eqref{eq:starter-merging-matrices}.
To check the absence of new logical qubits, consider a class in
$\ker\gamma_{*,0}$, represented by $(a,b)$ with $a\in\im H_Z$.
Since $B$ is surjective, adding a boundary from the information checks
can cancel $b$. The remaining first component still lies in $\im H_Z$
and can be canceled using the retained pivot checks, as in
Eq.~\eqref{eq:starter-syndrome-image}. The original class is therefore zero.

For a $k$-factor seed, the analogous path uses $k$ information-check orbits
and $k-1$ relation-qubit orbits, compared with one information-check orbit
for direct programming. It distributes the product across two-body
relations at the cost of these additional resources. Direct programming is
therefore a natural choice for the bounded-weight products used here;
merging provides an alternative when the support of an individual program
check would otherwise be too large. Both starters fix the logical action;
protecting the surviving logicals against short ancilla-supported errors
is the purpose of the subsequent constructions.

\FloatBarrier

\subsection{Repetition thickening}
\label{sec:thickening-and-ckbb}
\label{sec:repetition-thickening}
\label{sec:ckbb-thickening}
\label{sec:starter-low-weight-logicals}

The starter complex measures the desired products, but can shorten surviving logical
operators.
A standard approach to remedy this issue is \textit{thickening}~\cite{hastings2021quantum,cohen2022low,williamson2026gauging,he2025extractors}, which creates copies of the starter complex and couples them with identity interlayer maps.
We briefly review this technique, which is simple and effective but incurs a large space overhead.

\paragraph*{The short-error mechanism.}
For the direct starter complex with $P_0=C_0$ and $\gamma_0=I$, a base-qubit $X$ error
$e$ can cancel its syndrome using the ancilla pattern $H_Ze$:
\begin{equation}
  \partial_1^F(e,H_Ze)=H_Ze+H_Ze=0.
  \label{eq:starter-short-cycle}
\end{equation}
To see when this is a logical operator, take $e$ on an information coordinate which anti-commutes with an unmeasured logical $Z$ operator.
Multiplying the merged operator by the new $X$ checks removes its ancilla support and gives
$(\bar e,0)$, a logical $X$-operator. It is therefore nontrivial in the surviving quotient, although its weight is only $|e|+|H_Ze|$. This is the mechanism behind the short
representatives discussed in the main text.

\paragraph*{The layered construction.}
Let $A_m:R^m\to R^{m-1}$ be the repetition map,
$(A_m v)_i=v_i+v_{i+1}$, and place its domain and codomain in degrees $0,-1$.
The tensor product auxiliary $E_\bullet=P_\bullet\otimes_R\mathsf{Rep}^{(m)}_\bullet$
has groups
\begin{equation}
  E_1=P_1^m,\qquad
  E_0=P_0^m\oplus P_1^{m-1},\qquad
  E_{-1}=P_0^{m-1}.
  \label{eq:ckbb-graded-spaces}
\end{equation}
As before, $E_1$ labels new $X$-checks, $E_0$ added qubits, and $E_{-1}$ new
$Z$-checks. The spatial depth $m$ counts primal--dual layer pairs (with the
last intermediate layer omitted) and $m=1$ corresponds to just the starter complex.

Write $z=(z_i)\in E_1$ and $(a,b)=((a_i),(b_i))\in E_0$. The boundaries are
\begin{equation}
  \begin{aligned}
    \partial_1^E z
      &=\bigl((\partial_P z_i)_{i=1}^m,
               (z_i+z_{i+1})_{i=1}^{m-1}\bigr),\\
    \partial_0^E(a,b)
      &=\bigl(a_i+a_{i+1}+\partial_P b_i\bigr)_{i=1}^{m-1}.
  \end{aligned}
  \label{eq:ckbb-explicit-boundaries}
\end{equation}
Only the first layer attaches to the base:
$\gamma_1^{\mathrm{rep}}(z)=\gamma_1z_1$ and
$\gamma_0^{\mathrm{rep}}(a,b)=\gamma_0a_1$.
The equations make both the chain condition and the neighboring-layer
couplings explicit.

\paragraph*{Logical action and cost.}
A product of new $X$-checks with no ancilla support has
$z_1=\cdots=z_m$ and $\partial_Pz_1=0$, so it measures exactly a starter
cycle. Projection onto the attached layer also identifies
$H_0(E_\bullet)$ with $H_0(P_\bullet)$: the intermediate components can be
removed successively using the identity links. Both induced coupling maps
therefore agree with those of the starter, preserving the two requirements
of \cref{def:logical-correctness}, including when the listed checks are
dependent.

Put $n_P=\dim_{\mathbb F_2}P_0$ and $c_P=\dim_{\mathbb F_2}P_1$.
There are $mc_P$ new $X$-checks and $(m-1)n_P$ new $Z$-checks. The added
code-qubit count and the total added qubit-plus-check count are
\begin{equation}
  n_{\mathrm{anc}}=mn_P+(m-1)c_P,
  \qquad N_{\mathrm{anc}}=(2m-1)(n_P+c_P).
  \label{eq:repetition-resource-counts}
\end{equation}
Every layer meets only its neighbors, so check weights and qubit degrees
remain bounded as $m$ grows.

For the direct starter, the short logical above extends to
\begin{equation}
  \bigl(e;\underbrace{H_Ze,\ldots,H_Ze}_{m\text{ copies}};0\bigr),
  \qquad \text{weight}=|e|+m|H_Ze|,
  \label{eq:repetition-short-logical}
\end{equation}
where the final zero denotes the intermediate $P_1$ qubits. The neighboring
syndromes cancel, and multiplying by the pivot checks in every layer again
recovers $(\bar e,0)$. This is an upper-bound witness for the distance.
A lower bound must also control errors on the intermediate qubits, as shown
in \cref{sec:construction-distance-transfer}.

\subsection{Chain-map convolution}
\label{sec:chain-complex-convolution}

Convolution changes the coupling maps between layers.
For a fixed set of layers, this changes the checks without adding qubits.
The simplest modification replaces a repetition link $(1\;1)$ by
$(f\;g)$ with $f,g\in R$. In \cref{eq:ckbb-explicit-boundaries}, the two
corresponding terms become
\begin{equation}
  z_i+z_{i+1}\ \longrightarrow\ f_i z_i+g_i z_{i+1},
  \qquad
  a_i+a_{i+1}\ \longrightarrow\ f_i a_i+g_i a_{i+1}.
  \label{eq:scalar-convolution-links}
\end{equation}
Number the starter layers away from the base, beginning with the attached
layer $1$. Relation $i$ joins layers $i$ and $i+1$: $f_i$ multiplies the
couplings from layer $i$ into that relation, while $g_i$ multiplies those
from layer $i+1$. Each replacement acts on both the $P_1$ and $P_0$
coordinates. In Fig.~\ref{fig:construction-pipeline}(b), these are the paired
diagonal $f$ arrows from the nearer layer and the paired diagonal $g$ arrows
from the farther layer. The horizontal $P_1\xrightarrow{\partial_P}P_0$
maps remain unchanged. Using the same polynomial on the two diagonals of
each pair ensures commutation with $\partial_P$, so the modified boundaries
still compose to zero.

Each nonzero coefficient of $f_i$ or $g_i$ adds a translated coupling.
For example, $1+x+x^3$ connects positions shifted by $0,1,3$ within the
same orbit. The reported weights $(|f_i|,|g_i|)$ count these coefficients.
The qubit and check counts in \cref{eq:repetition-resource-counts} remain
unchanged, while the interlayer incidence grows with these weights. Bounded
polynomial weights therefore preserve LDPC sparsity, subject to the desired
physical degree cap.

\paragraph*{Preserving the measurement.}
Commuting checks are necessary but do not alone guarantee the intended
logical action. A useful sufficient condition for convolved repetition attached at its
first layer is that every outward coefficient $g_i$ is a unit in $R$.
Then every starter cycle extends uniquely through the layers by
\begin{equation}
  z_{i+1}=g_i^{-1}f_i z_i.
  \label{eq:convolution-cycle-extension}
\end{equation}
Successive elimination using $g_i^{-1}$ also identifies the degree-zero
homology with that of the starter. Evaluation at the first layer thus
preserves the total measured subspace $J$ and the injectivity of
$\gamma_{*,0}$. Here $g_i$ is a unit precisely when it is coprime to
$x^P-1$. Inverses are used only for this logical argument; the physical
couplings are the sparse polynomials $f_i,g_i$.
For other choices, the two requirements of \cref{def:logical-correctness}
must be verified directly. No invertibility condition by itself guarantees
an improvement in Hamming distance.

\paragraph*{Matrix-valued couplings.}
To mix different orbits, replace a scalar arm by maps $S_1:P_1\to P_1$
and $S_0:P_0\to P_0$ satisfying
\begin{equation}
  S_0\partial_P=\partial_P S_1.
  \label{eq:general-convolution-square}
\end{equation}
Use $S_1$ on the check coordinates and $S_0$ on the qubit coordinates of
the same arm. This again preserves commutation. Scalar convolution is the
special case $S_j=fI_{P_j}$. For a subgroup $H\subseteq G$, sparse matrices
over $\mathbb F_2[H]$ preserve the $H$ symmetry while mixing its orbit
coordinates; binary matrices permit couplings without a prescribed symmetry.
Logical correctness is checked for the assembled auxiliary, and row and
column weights set the physical incidence budgets. Ancilla addition or check
splitting, when used alongside these replacements, contributes separately
to the overhead.

\subsubsection{Scalar convolution as a lifted product}
\label{sec:scalar-convolution-lifted-product}
\label{sec:general-transverse-factors}

Scalar convolution has a direct interpretation in terms of a classical
quasi-cyclic (QC) code. Let $m$ be the number of starter layers and $s$
the number of transverse parity-check relations. Assemble their couplings
into $H(x)\in R^{s\times m}$, where $R=\mathbb F_2[x]/(x^P-1)$:
columns label layers and rows label relations, which need not be independent.
For repetition and its scalar-convolved version, $s=m-1$; a general transverse
QC factor can have a different number of relations. Replacing
$x$ by the cyclic shift $T:e_j\mapsto e_{j+1\bmod P}$ gives a binary
parity-check matrix $\widehat H=H(T)$ and a classical code
$\mathcal C_H=\ker\widehat H$ of length $mP$. For example, three layers
with scalar-convolved links give
\[
 H(x)=\begin{pmatrix}f_1&g_1&0\\0&f_2&g_2\end{pmatrix}.
\]
Repetition is the choice $H=A_m$; a general sparse QC matrix permits other
transverse relations within the same construction.

Set $D=\partial_P$ and regard $H$ as the two-term complex
$K_\bullet=(R^m\xrightarrow{H}R^s)$ in degrees $0,-1$. The auxiliary is
\begin{equation}
 E_\bullet=P_\bullet\otimes_R K_\bullet,
 \qquad E_1=P_1^m,\quad E_0=P_0^m\oplus P_1^s,\quad E_{-1}=P_0^s.
 \label{eq:general-transverse-total-complex}
\end{equation}
This is the QC lifted product of the starter and the classical check
complex~\cite[Secs.~III.D--III.E]{panteleev2022almostlinear}. The product is
formed over $R$ before binary expansion: both factors share one cyclic
coordinate, so its size contributes a factor $P$, rather than the $P^2$
of an ordinary binary tensor product of the separately expanded factors.
The diagram shows a small QC factor and recalls how its couplings act in
both starter degrees.

\begin{center}
\begin{tikzpicture}[
 font=\footnotesize,
 bit/.style={circle,draw=basecol,fill=basecol!7,minimum size=5mm,inner sep=1pt},
 parity/.style={rectangle,draw=mapcol,fill=mapcol!8,minimum size=4.5mm,inner sep=1pt},
 coord/.style={rounded corners=2pt,draw,align=center,
               minimum width=21mm,minimum height=9mm,inner sep=2pt},
 xcoord/.style={coord,draw=anccol,fill=anccol!7},
 qcoord/.style={coord,draw=basecol,fill=basecol!7},
 zcoord/.style={coord,draw=mapcol,fill=mapcol!7},
 arr/.style={-{Stealth[length=1.7mm]},semithick},
 lab/.style={fill=white,inner sep=1pt},
 plain/.style={draw=black!45},
 shifted/.style={draw=shiftcol,semithick,dashed}
]
\node[font=\small] at (2,2.05) {(a) QC factor: $H=(1+x\;\;1)$, $P=5$};
\foreach \j in {0,...,4}{
 \node[bit] (u\j) at (\j,0.9) {$u_{\j}$};
 \node[parity] (r\j) at (\j,0) {$r_{\j}$};
 \node[bit] (v\j) at (\j,-0.9) {$v_{\j}$};
 \draw[plain] (u\j) -- (r\j);
 \draw[plain] (v\j) -- (r\j);
}
\foreach \j [evaluate=\j as \k using {int(\j+1)}] in {0,...,3}{
 \draw[shifted] (u\j) -- (r\k);
}
\draw[shifted] (u4.north) to[out=115,in=150,looseness=1.15] (r0.west);
\node[align=center] at (2,-1.55)
 {Solid: $1$; dashed: $x$\\$r_j=u_j+u_{j-1}+v_j$ (indices mod $5$)};

\node[font=\small] at (9.4,2.05) {(b) Product with the starter $D$};
\node at (6.5,1.4) {$E_1=P_1^m$};
\node at (9.4,1.4) {$E_0=P_0^m\oplus P_1^s$};
\node at (12.3,1.4) {$E_{-1}=P_0^s$};
\node[xcoord] (z) at (6.5,0)
 {$z_i\in P_1$\\New $X$-checks};
\node[qcoord] (a) at (9.4,0.65)
 {$a_i\in P_0$\\Layer qubits};
\node[qcoord] (b) at (9.4,-0.65)
 {$b_\alpha\in P_1$\\Relation qubits};
\node[zcoord] (c) at (12.3,0)
 {$P_0$ at $\alpha$\\New $Z$-checks};
\draw[arr,draw=anccol] (z) -- node[lab,above] {$D$} (a);
\draw[arr,draw=anccol] (z) -- node[lab,below] {$H_{\alpha i}$} (b);
\draw[arr,draw=mapcol] (a) -- node[lab,above] {$H_{\alpha i}$} (c);
\draw[arr,draw=mapcol] (b) -- node[lab,below] {$D$} (c);
\node[align=center] at (9.4,-1.55)
 {Layer $i=1,\ldots,m$; relation $\alpha=1,\ldots,s$\\
 The same polynomial $H_{\alpha i}$ acts on both degrees.};
\end{tikzpicture}
\end{center}

Writing $z=(z_i)$ and $(a,b)=((a_i),(b_\alpha))$, the boundaries are
\begin{equation}
 \begin{aligned}
 \partial_1^E&=\begin{pmatrix}I_m\otimes D\\H\otimes I_{P_1}\end{pmatrix},
 &\partial_1^E z&=\bigl((Dz_i)_i,(\textstyle\sum_iH_{\alpha i}z_i)_\alpha\bigr),\\
 \partial_0^E&=\begin{pmatrix}H\otimes I_{P_0}&I_s\otimes D\end{pmatrix},
 &\partial_0^E(a,b)&=\bigl(\textstyle\sum_iH_{\alpha i}a_i+Db_\alpha\bigr)_\alpha.
 \end{aligned}
 \label{eq:general-transverse-boundaries}
\end{equation}
The two compositions cancel because $D$ commutes with the entries of $H$.
Thus scalar chain-map convolution implements the lifted product directly;
bounded binary row and column weights of $D$ and $H$ preserve sparsity.
There are $mn_P+sc_P$ added qubits. Attachment, for example
$\gamma_1^E(z)=\gamma_1z_1$ and $\gamma_0^E(a,b)=\gamma_0a_1$, must still
satisfy the logical correctness conditions. The classical distance of
$\mathcal C_H$ alone does not establish a distance bound for the merged code.
Crucially, although the scalar-convolved ancilla complex is a lifted-product
code, the bare code to which it attaches need not be, extending the use of
lifted-product ancillas beyond the lifted-product data-code setting of prior
surgery constructions~\cite{zheng2026logical,hirasaki2026lifting,zhou2026parallel}.

\subsection{Distance bounds and guarantees}
\label{sec:construction-distance-transfer}

The relevant distance is the static stabilizer-code distance of the complete
merged code, with binary Hamming weight on $C_1\oplus E_0$ and all surviving
active and residual logicals counted. Logical faithfulness already gives the
$Z$-distance bound in \cref{eq:automatic-z-distance-bound}. The remaining
issue is whether a surviving $X$ logical can become short by using the
ancillas.

Write $J=\im\gamma_{*,1}$ and define the surviving base-code distance by
\begin{equation}
 d_X(C_\bullet;J)
   =\min\{\,|x|:x\in\ker\partial_1^C,\ [x]\notin J\,\}.
 \label{eq:surviving-base-x-distance}
\end{equation}
This is at least $d_X(C_\bullet)$; the minimum is infinite if no logical
class survives. A lower bound follows if removing small ancilla support
requires only a controlled correction on the base qubits.

\begin{proposition}[Distance from bounded ancilla removal]
\label{prop:static-distance-repair}
Suppose the auxiliary is logically faithful. Assume that for every
$a\in\im\partial_1^E$ with $|a|<\Delta$, there are $z\in E_1$ and
$c\in C_2$ such that
\begin{equation}
 \partial_1^Ez=a,
 \qquad |\gamma_1z+\partial_2^Cc|\le\lambda|a|,
 \qquad \lambda\ge1.
 \label{eq:static-ancilla-removal}
\end{equation}
Then
\begin{equation}
 d_X(F_\bullet)\ge
 \min\left\{\Delta,\frac{d_X(C_\bullet;J)}{\lambda}\right\}.
 \label{eq:static-distance-transfer}
\end{equation}
\end{proposition}
\begin{proof}
Let $(x,a)$ be a nontrivial merged-code $X$ logical of weight
$w=|x|+|a|<\Delta$. Logical faithfulness implies
$a\in\im\partial_1^E$. Adding the stabilizers specified by $z,c$ removes
the ancilla support, leaving the base cycle
$x'=x+\gamma_1z+\partial_2^Cc$. Its class lies outside $J$, since the
original logical was nontrivial. Consequently,
\[
 d_X(C_\bullet;J)\le |x'|\le |x|+\lambda|a|\le\lambda w.
\]
A nontrivial logical therefore has weight at least $\Delta$ or at least
$d_X(C_\bullet;J)/\lambda$.
\end{proof}

The condition in \cref{eq:static-ancilla-removal} bounds the physical cost
of cleaning the ancillas; it is often called a relative cofilling bound.
It permits dressing by base stabilizers and measured logicals through the
choices of $c$ and $z$. Logical faithfulness ensures that a removal exists,
but does not bound its weight. For convolution, establishing a useful
$\Delta$ and $\lambda$ is the additional distance argument; bounded degree
and a large classical kernel distance alone do not supply it.

\subsubsection{Distance lower bounds for repetition}
\label{sec:repetition-distance-lower-bounds}

\paragraph*{A lower bound for repetition.}
For the repeated construction, let $\lambda_\gamma$ be the larger of one
and the maximum binary column weight of the starter coupling $\gamma_1$.
The removal condition holds with $\Delta=m$ and
$\lambda=\lambda_\gamma$. Indeed, an ancilla pattern $(a,b)$ of total
weight less than $m$ has some empty $P_0$ layer, say $a_k=0$. Set $z_k=0$
and solve $z_i+z_{i+1}=b_i$ in both directions. The cycle equations in
\cref{eq:ckbb-explicit-boundaries} then give $\partial_Pz_i=a_i$, so these
checks remove the entire ancilla pattern. Their base action satisfies
\[
 |\gamma_1z_1|
   =\left|\gamma_1\sum_{i<k}b_i\right|
   \le\lambda_\gamma\sum_i|b_i|.
\]
Thus
\begin{equation}
 d_X(F_\bullet^{\mathrm{rep}})\ge
 \min\left\{m,\frac{d_X(C_\bullet;J)}{\lambda_\gamma}\right\}.
 \label{eq:repetition-distance-lower-bound}
\end{equation}
This conservative lower bound applies to either starter in
\cref{sec:sparse-starter-complexes}. The explicit
witness in \cref{eq:repetition-short-logical} supplies an upper bound for
the direct starter; the two need not coincide. More generally, every surviving
base-code $X$ logical is also a merged-code representative, giving
$d_X(F_\bullet)\le d_X(C_\bullet;J)$.

\subsubsection{Distance upper bounds for convolved complexes}
\label{sec:convolution-distance-upper-bounds}

As discussed in \cref{sec:scalar-convolution-lifted-product}, the
scalar-convolved ancilla complex is the lifted product
$E_\bullet=P_\bullet\otimes_R K_\bullet$ of the starter complex $P_\bullet$
and the classical QC check complex
$K_\bullet=(R^m\xrightarrow{H}R^s)$, with $m$ starter layers and $s$
transverse relations ($s=m-1$ for convolved repetition). Its binary code is
$\mathcal C_H=\ker\widehat H$. Here we show how a short logical operator
of the starter merged code can be spread across the ancilla layers using a
classical codeword, giving an upper bound on the distance boosting that
convolution can achieve.

The physical qubits of the merged code split into three groups:
\[
 F_1=
 \underbrace{C_1}_{\text{base qubits}}\ \oplus\
 \underbrace{P_0^m}_{\text{layer qubits}}\ \oplus\
 \underbrace{P_1^s}_{\text{relation qubits}}.
\]
We write their $X$-operator support as $(x;(a_i);(b_\alpha))$, with
$x\in C_1$, $a_i\in P_0$, and $b_\alpha\in P_1$. We attach the auxiliary
to the base at layer one and take all maps to be $R$-linear.
Write a classical codeword as
\[
 q=(q_1,\ldots,q_m)^{\mathsf T}\in K_0=R^m,
 \qquad q_i\in R,\qquad Hq=0.
\]
Each polynomial component $q_i$ specifies the cyclic shifts applied in
layer $i$; in particular, $q_1$ is the component at the layer attached to
the base. Expanding all $m$ components into their $P$ binary coefficients
gives a word in $\mathcal C_H$. Its total binary weight is
$|q|=\sum_{i=1}^m|q_i|$; all other weights below are also binary.

\begin{proposition}[A logical operator from a starter logical and a classical word]
\label{prop:convolution-product-upper-bound}
Let $(e,a)\in C_1\oplus P_0$ be a nontrivial $X$ logical operator of the
starter merged code. For any classical word $q$ with $Hq=0$ and $q_1=1$,
\[
 \omega=(e;q_1a,\ldots,q_ma;0)
       \in C_1\oplus P_0^m\oplus P_1^s
\]
is a nontrivial $X$ logical operator of the convolved merged code, and
\begin{equation}
 d_X(F_\bullet)\leq |\omega|
       =|e|+\sum_{i=1}^m|q_i a|\leq |e|+|q|\,|a|.
 \label{eq:convolution-product-upper-bound}
\end{equation}
\end{proposition}
\begin{proof}
The base syndrome cancels because $(e,a)$ is a starter logical and $q_1=1$;
the transverse syndrome vanishes because $Hq=0$. To check nontriviality,
project onto the base and first layer, discarding all other qubits.
Every convolved $X$-stabilizer projects to a starter $X$-stabilizer: a
product of new checks with coefficients $(z_i)$ projects to the starter
check product $(\gamma_1z_1,\partial_Pz_1)$, and base checks are unchanged.
But $\omega$ projects to the nontrivial starter logical $(e,a)$, so it
cannot be a stabilizer. Finally, each $q_i a$ is a sum of cyclic shifts
of $a$, giving $|q_i a|\leq |q_i|\,|a|$.
\end{proof}

For the direct starter, use the short logical $(e,H_Ze)$ constructed in
\cref{sec:starter-low-weight-logicals}, where $e$ is a single
information-qubit error with a surviving logical completion.
More generally, when $q_1\ne1$, the same proof applies to
$(q_1e;q_1a,\ldots,q_ma;0)$ provided $q_1(e,a)$ is still a nontrivial
starter logical. In particular, this holds whenever $q_1$ is invertible
in $R$, and gives
\begin{equation}
 d_X(F_\bullet)\leq |q_1e|+\sum_i|q_i a|
                  \leq |q_1|\,|e|+|q|\,|a|.
 \label{eq:convolution-attached-word-upper-bound}
\end{equation}
Thus a lighter classical word gives a smaller upper bound, provided it
preserves the logical action at the attachment. An arbitrary minimum-weight
classical word need not do so. Throughout, the independent base-code ceiling
$d_X(F_\bullet)\leq d_X(C_\bullet;J)$ also applies.

\paragraph*{Convolving the repetition factor.}
Ordinary repetition copies the starter error $a$ into all $m$ layers,
giving $(e;a,\ldots,a;0)$ of weight $|e|+m|a|$. Scalar convolution replaces
the two identity couplings into each of the $m-1$ relations by $f_i$ and
$g_i$: $f_i$ acts from the nearer layer $i$, and $g_i$ from the farther
layer $i+1$. These are the paired diagonal arrows in
Fig.~\ref{fig:construction-pipeline}(b); the starter maps $\partial_P$
remain unchanged. The classical map $H:R^m\to R^{m-1}$ imposes
$f_iq_i+g_iq_{i+1}=0$. An explicit classical word $q\in R^m$ satisfying
these relations is
\begin{equation}
 q_i=(f_1\cdots f_{i-1})(g_i\cdots g_{m-1}),\qquad 1\leq i\leq m,
 \label{eq:convolution-two-arm-word}
\end{equation}
with products taken in $R$ and empty products equal to $1$. For each
$i=1,\ldots,m-1$,
\[
 f_iq_i=(f_1\cdots f_i)(g_i\cdots g_{m-1})=g_iq_{i+1},
 \qquad (Hq)_i=f_iq_i+g_iq_{i+1}=0.
\]
Thus $q$ satisfies every convolved repetition relation. This specializes the
classical QC kernel-word
construction of Ref.~\cite[Lemma~6]{smarandache2012qc} to the repetition
connectivity. If every $g_i$ is invertible in $R$, then so is
$q_1=g_1\cdots g_{m-1}$, and the general construction in
\cref{eq:convolution-attached-word-upper-bound} remains a nontrivial logical.
If each convolution polynomial has weight at most $w\geq1$, then
$|q_i|\leq w^{m-1}$ because $q_i$ contains $m-1$ factors. Hence
\begin{equation}
 \boxed{d_X(F_\bullet)\leq w^{m-1}\bigl(|e|+m|a|\bigr).}
 \label{eq:convolution-two-arm-upper-bound}
\end{equation}
Setting all $f_i=g_i=1$ recovers the repetition bound $|e|+m|a|$;
polynomial cancellations can only lower the witness weight. The general
bound assumes the repetition connectivity and a nontrivial starter logical. A general transverse QC matrix need not admit the same word, and
a noninvertible $q_1$ can destroy its logical nontriviality.

Therefore, we see in the upper bound, \textbf{convolution weight sets the base
of the exponential, and the
number of layers sets its exponent.} For a short starter witness,
the distance ceiling is $O(mw^{m-1})$: at most exponential in the layer
count $m$ for fixed $w>1$, and polynomial of degree $m-1$ in the convolution
weight $w$ at fixed depth.

\section{Efficient fault tolerance}
\label{app:efficient-fault-tolerance}

Appendix~\ref{app:symmetric-logical-measurements} identifies the logical action of a
surgery gadget, and Appendix~\ref{app:sparse-surgery-constructions} constructs
sparse static complexes with that action. 
However, in a physical implementation, errors occur during syndrome measurements too. This appendix explains the additional ingredients required to protect against both data and syndrome faults simultaneously.

We focus on the offline preparation of a parallel PPM resource state used
in the main text. We first specify the protected correlations, then describe
the resource state and its gate-teleportation readout.
The asymptotic joint spatial and temporal distance guarantees are given in
Appendix~\ref{sec:pk-temporal-distance}.
Appendix~\ref{app:family-level-scaling} discusses some of the additional conditions that are
required when readout gadgets are composed.

The fault-complex formalism represents fault locations, detectors, fault
equivalences, and protected logical correlations in a single chain complex
\cite{hillmann2025faultcomplex}.  Mapping telescopes and complementary temporal
boundaries extend this language to repeated syndrome extraction and homological
surgery processes
\cite{cowtan2024css,ide2025homological,cowtan2025fast}.  Our use of these tools
is relatively straightforward, we apply convolution in the time direction, keep track of its
logical transport, and state the process-level conditions needed by the
surgery process.

\subsection{Fault-tolerance requirements of the surgery process}
\label{sec:ft-process-requirements}

Let $J=M\oplus M_{\mathrm{res}}$ be the total measured subspace of \cref{sec:symmetric-logical-ppm}.  In the
absence of faults, the process should report the eigenvalues of an independent
generating set of $J$ and project the data into the corresponding joint
eigenspace.  In the convention fixed in
\cref{sec:symmetric-css-complexes}, the surviving logical sectors are
$H_1(C_\bullet)/J$ on the $X$ side and $J^\perp\subseteq H^1(C_\bullet)$ on
the $Z$ side.

Therefore, there are two kinds of failures that can occur during surgery:
(1) A fault that flips a reported PPM eigenvalue, or (2) a fault that flips a surviving logical operator~\cite{hillmann2025faultcomplex}.  
In practice, when numerically characterizing the distance of a surgery process, we also need to specify different boundary conditions depending on the basis of the relevant logical operators.
Therefore, (2) splits into separate $X$ and $Z$ components, and in practice we work with three different notations of distance.

Another subtlety when analyzing the fault-complex, is there are often products of faults that appear which are logically non-trivial.
We define the three distinct detector error models (DEMs) we study as $a\in\{\mathrm{PPM},X\text{-state},Z\text{-state}\}$.
Let $D_a$ and $B_a$ be the detector and fault-equivalence maps supplied by the fault complex, with $D_aB_a=0$. The space of undetected fault classes per DEM is
\begin{equation}
 \mathfrak F_a:=\ker D_a/\operatorname{im}B_a.
 \label{eq:semantic-fault-class-space}
\end{equation}
Let $L_a e$ record the flips of the protected logical correlations caused by a fault vector $e$, including the prescribed Pauli-frame updates. We require $L_aB_a=0$, so equivalent faults have the same protected effect. An undetected fault is harmful precisely when $L_a e\ne0$.

\begin{definition}[Process-distance]
\label{def:semantic-process-distance}
For a fixed finite surgery schedule, define
\begin{equation}
 d_a:=\min\bigl\{\operatorname{wt}(e):
     D_a e=0,\ L_a e\ne0\bigr\},
 \qquad
 a\in\{\mathrm{PPM},X\text{-state},Z\text{-state}\},
 \label{eq:semantic-distance-definition}
\end{equation}
where the minimum is over physical fault-location representatives and is
$\infty$ if no harmful undetected fault exists.
The process or fault distance is
\begin{equation}
 d_{\mathrm{proc}}
   :=\min\{d_{\mathrm{PPM}}^X,d_{\mathrm{state}}^X,
            d_{\mathrm{state}}^Z\}.
 \label{eq:process-distance}
\end{equation}
\end{definition}

Note that the a definition of distance based on only undetected and harmful faults has been discussed in both the fault-complex, surgery, and subsystem code settings~\cite{hillmann2025faultcomplex,cowtan2025fast, vuillot2019gaugefixing,cowtan2024ssip, delfosse2023spacetime_v2, williamson2026gauging}.
The three distances in \cref{eq:process-distance} and the name
\emph{process distance} are terminology used here to identify the operational
meaning of each protected correlation.

\paragraph*{Boundary completions for numerical verification.}

The three components of the process-distance profile need to be tested using
different preparation and readout settings, in order to ensure there are deterministic observeables which can detect faults during the circuit.
For an $X$-type PPM
process, Table~\ref{tab:process-distance-boundaries} specifies the
bases we use for certifying the fault-distance. 
Note in the MBQC setting, changing the input/output basis corresponds to appending an additional half-layer in the alternating fault-complex structure~\cite{bolt2016foliated}.

\begin{table}[htbp]
\caption{Boundary conditions for verifying an $X$-type PPM
process. Target labels identify the protected observables, whereas the
superscripts on the surviving-state distances identify the logical Pauli
errors that change those observables.}
\label{tab:process-distance-boundaries}
\centering
\begin{tabular}{lcccc}
\hline\hline
Protected target & Input preparation & Output readout & Fault sector & Distance\\
\hline
Measured $X$-PPM outcomes & $X$ & $Z$ & $Z$ & $d_{\mathrm{PPM}}^X$\\
All logical $X$ observables & $X$ & $X$ & $Z$ & $d_{\mathrm{state}}^Z$\\
Surviving $Z$ observables & $Z$ & $Z$ & $X$ & $d_{\mathrm{state}}^X$\\
\hline\hline
\end{tabular}
\end{table}

Note that when we discuss the fault-distance, we are measuring the phenomenological distance.
That means in the MBQC implementation of the fault-complex, and the corresponding detector error model, we only apply errors during qubit measurement.
We also always work with bounded-degree complexes, the setting in which the phenomenological distance is likely to reflect true circuit-level protection as well~\cite{christandl2025generalnoise}.

For our circuit-level numerics in Fig.~\ref{fig:convolution-performance}c, the same general approach is used to characterize logical successes and failures, but applying a two-qubit depolarizing error channel after every CNOT gate.

\FloatBarrier

\subsection{Offline resource preparation and readout}
\label{sec:offline-steane-surgery}

The factory prepares an encoded resource offline and couples it to the data
through a single transversal gate, implementing the PPM by gate teleportation~\cite{gottesman1999teleportation,huang2023homomorphic}
(see Fig.~\ref{fig:offline-steane-surgery}). Its temporal convolution need
not be undone on the data; the recorded signs are expressed in the desired
logical generator basis.

Let $M_X$ denote the programmed $X$-type PPM, possibly acting across several
logical orbits; a commuting family is treated in parallel.  A matching
encoded ancilla $A$, with $q$ logical qubits, begins in the logical $Z$-basis
state $\ket{0}_A^{\otimes q}$.  Measuring $M_X$ on this ancilla offline, with
outcome $\mu\in\{+1,-1\}$, prepares
\begin{equation}
 \ket{\Omega_{M_X,\mu}}_A
   \propto \Pi^{M_X}_{\mu,A}\ket{0}_A^{\otimes q}.
 \label{eq:offline-ppm-ancilla-resource}
\end{equation}
For the full programmed family, the resource is fixed by the measured
$X$ products in $J$, with their recorded signs, and the logical $Z$
stabilizers in $J^\perp$ inherited from $\ket{0}^{\otimes q}$.
The measured subspace $J$ is closed under the symmetry action, so unit
scalar convolution preserves $J$.  Its inverse-transpose action on $Z$ operators therefore preserves
$J^\perp$.  These latter stabilizers all have sign $+1$, which is unchanged
by taking products.  Thus convolution preserves the resource-state family:
it can mix the measured generators and their recorded signs, but the signs
in the desired generator basis are recovered by classical processing.
Below, $\mu$ denotes the outcome in that basis.

The unmeasured factory logicals are initialized in a fixed logical
$Z$-basis state. Unit scalar convolution therefore needs no additional
constraint that the accumulated convolution act as the identity on those logicals.
The terminal resource implements the desired measurement through the
circuit in \cref{fig:offline-steane-surgery}.

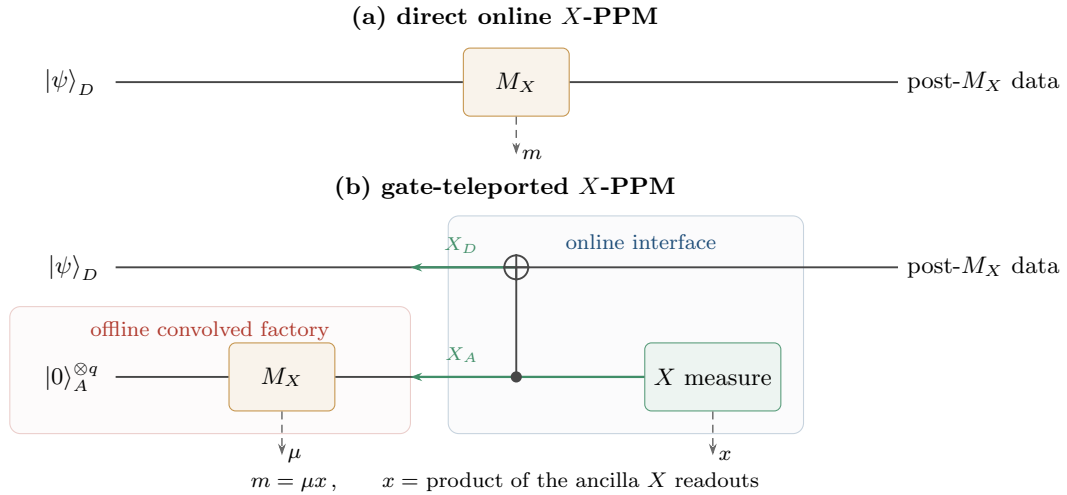
\begin{figure}[htbp]
\centering
\begin{tikzpicture}[
  >={Stealth[length=1.9mm,width=1.35mm]},
  wire/.style={draw=black!72, line width=0.72pt},
  cwire/.style={-{Stealth[length=1.55mm,width=1.05mm]}, densely dashed,
                draw=black!58, line width=0.55pt},
  opflow/.style={-{Stealth[length=1.55mm,width=1.10mm]},
                 draw=mapcol!85, line width=0.85pt},
  gate/.style={rounded corners=2.5pt, draw=shiftcol!80, fill=shiftcol!9,
               minimum width=14mm, minimum height=9mm, inner sep=2pt,
               font=\small},
  measure/.style={rounded corners=2.5pt, draw=mapcol!75, fill=mapcol!8,
                  minimum width=18mm, minimum height=9mm, inner sep=2pt,
                  font=\small, align=center},
  lab/.style={font=\footnotesize, inner sep=0.5pt, fill=white},
]
  \node[font=\small\bfseries] at (6.0,3.20)
    {(a) direct online $X$-PPM};
  \node[font=\small,anchor=east] at (0.70,2.35) {$\ket\psi_D$};
  \node[gate] (directmx) at (6.15,2.35) {$M_X$};
  \draw[wire] (0.85,2.35) -- (directmx.west);
  \draw[wire] (directmx.east) -- (11.20,2.35)
    node[font=\small,anchor=west] {post-$M_X$ data};
  \draw[cwire] (directmx.south) -- ++(0,-0.50)
    node[lab,anchor=west,xshift=1.5pt] {$m$};

  \node[font=\small\bfseries] at (6.0,0.95)
    {(b) gate-teleported $X$-PPM};
  \begin{scope}[on background layer]
    \filldraw[rounded corners=4pt,fill=anccol!2,draw=anccol!25]
      (-0.55,-2.30) rectangle (4.75,-0.62);
    \filldraw[rounded corners=4pt,fill=basecol!2,draw=basecol!25]
      (5.25,-2.30) rectangle (9.95,0.58);
  \end{scope}
  \node[font=\footnotesize,text=anccol,anchor=north] at (2.10,-0.70)
    {offline convolved factory};
  \node[font=\footnotesize,text=basecol,anchor=north] at (7.80,0.45)
    {online interface};

  \node[font=\small,anchor=east] at (0.75,-0.10) {$\ket\psi_D$};
  \draw[wire] (0.85,-0.10) -- (11.20,-0.10)
    node[font=\small,anchor=west] {post-$M_X$ data};

  \node[font=\small,anchor=east] at (0.75,-1.55)
    {$\ket{0}_A^{\otimes q}$};
  \node[gate] (offmx) at (3.05,-1.55) {$M_X$};
  \node[measure] (xmeas) at (8.75,-1.55) {$X$ measure};
  \draw[wire] (0.85,-1.55) -- (offmx.west);
  \draw[wire] (offmx.east) -- (xmeas.west);
  \draw[cwire] (offmx.south) -- ++(0,-0.58)
    node[lab,anchor=west,xshift=1.5pt] {$\mu$};

  \draw[opflow] (xmeas.west) -- (4.75,-1.55)
    node[font=\scriptsize,text=mapcol!85,above=2pt,pos=0.78] {$X_A$};
  \draw[opflow] (6.15,-0.10) -- (4.75,-0.10)
    node[font=\scriptsize,text=mapcol!85,above=2pt,pos=0.52] {$X_D$};

  \draw[wire] (6.15,-1.55) -- (6.15,-0.10);
  \fill[black!72] (6.15,-1.55) circle (2.1pt);
  \draw[black!72,line width=0.72pt] (6.15,-0.10) circle (4.5pt);
  \draw[black!72,line width=0.72pt] (6.15,-0.275) -- (6.15,0.075);
  \draw[black!72,line width=0.72pt] (5.975,-0.10) -- (6.325,-0.10);

  \draw[cwire] (xmeas.south) -- ++(0,-0.58)
    node[lab,anchor=west,xshift=1.5pt] {$x$};

  \node[font=\footnotesize] at (6.0,-2.94)
    {$m=\mu x$\,,\qquad $x=$ product of the ancilla $X$ readouts};
\end{tikzpicture}
\caption{Offline surgery as gate teleportation.  Panel~(a) is the desired
online $X$-type PPM instrument.  In panel~(b), $M_X$ is first measured on a
logical $\ket{0}$ ancilla in the offline factory, producing the stored outcome
$\mu$.  The online CNOTs are controlled by the ancilla and target the main-line
data.  Measuring the ancilla in the $X$ basis then projects the data by the same
$M_X$, with outcome $m=\mu x$, where $x$ is the product of the relevant
ancilla readouts.  This circuit identity is standard~\cite{huang2023homomorphic}; convolution changes how
the ancilla resource is protected and prepared, not its online logical action.
Green arrows show backward operator propagation through the CNOT:
$X_A^{\mathrm{out}}\mapsto X_A^{\mathrm{in}}X_D^{\mathrm{in}}$.}
\label{fig:offline-steane-surgery}
\end{figure}

The circuit cartoon suppresses two complete fault-complex histories.
The online computation is a mapping telescope of copies of $C_\bullet$, with
shift-$0$ rounds alternating with shift-$1$ intervals.  The resource factory
is a separate telescope of copies of the surgery complex $F_\bullet$.  Only
their terminal interface is shared.  This organization is shown in
\cref{fig:offline-online-fault-complex}.

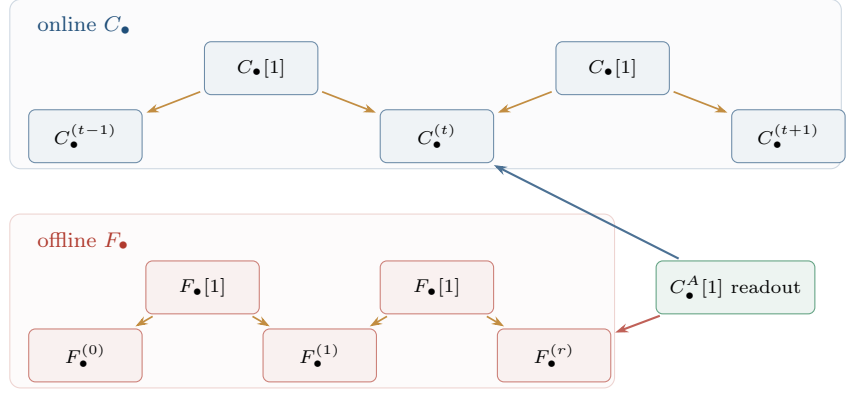
\begin{figure}[htbp]
\centering
\begin{tikzpicture}[
  >={Stealth[length=1.8mm,width=1.25mm]},
  arch/.style={rounded corners=2.5pt, minimum width=15mm,
               minimum height=7mm, inner sep=1.5pt, font=\scriptsize},
  cround/.style={arch,draw=basecol!65,fill=basecol!7},
  fround/.style={arch,draw=anccol!65,fill=anccol!7},
  readout/.style={arch,draw=mapcol!70,fill=mapcol!8,
                  minimum width=21mm},
  incidence/.style={-{Stealth[length=1.65mm,width=1.15mm]},
                     draw=shiftcol!85,line width=0.65pt,
                     shorten >=1.5pt,shorten <=1.5pt},
  factoryleg/.style={-{Stealth[length=1.75mm,width=1.2mm]},
                      draw=anccol!80,line width=0.8pt,
                      shorten >=1.5pt,shorten <=1.5pt},
  dataleg/.style={-{Stealth[length=1.75mm,width=1.2mm]},
                   draw=basecol!80,line width=0.8pt,
                   shorten >=1.5pt,shorten <=1.5pt},
  boundary/.style={-{Stealth[length=1.45mm,width=1.0mm]},
                   draw=black!65,line width=0.65pt,
                   shorten >=1.5pt,shorten <=1.5pt},
  gamma/.style={-{Stealth[length=1.45mm,width=1.0mm]},
                draw=mapcol!80,line width=0.7pt,
                shorten >=1.5pt,shorten <=1.5pt},
  maplab/.style={font=\scriptsize,inner sep=0.8pt,fill=white},
  groupbox/.style={rounded corners=4pt,draw=black!15,fill=black!1},
]
  \begin{scope}[on background layer]
    \filldraw[rounded corners=4pt,fill=basecol!2,draw=basecol!22]
      (-0.45,0.62) rectangle (10.55,2.86);
    \filldraw[rounded corners=4pt,fill=anccol!2,draw=anccol!22]
      (-0.45,-2.28) rectangle (7.55,0.02);
  \end{scope}
  \node[font=\footnotesize,text=basecol,anchor=north west]
    at (-0.20,2.75) {online $C_\bullet$};
  \node[font=\footnotesize,text=anccol,anchor=north west]
    at (-0.20,-0.10) {offline $F_\bullet$};

  \node[cround] (c0) at (0.55,1.05) {$C_\bullet^{(t-1)}$};
  \node[cround] (c1) at (5.20,1.05) {$C_\bullet^{(t)}$};
  \node[cround] (c2) at (9.85,1.05) {$C_\bullet^{(t+1)}$};
  \node[cround] (ce0) at (2.875,1.95) {$C_\bullet[1]$};
  \node[cround] (ce1) at (7.525,1.95) {$C_\bullet[1]$};
  \draw[incidence] (ce0) -- (c0);
  \draw[incidence] (ce0) -- (c1);
  \draw[incidence] (ce1) -- (c1);
  \draw[incidence] (ce1) -- (c2);

  \node[fround] (f0) at (0.55,-1.85) {$F_\bullet^{(0)}$};
  \node[fround] (f1) at (3.65,-1.85) {$F_\bullet^{(1)}$};
  \node[fround] (fr) at (6.75,-1.85) {$F_\bullet^{(r)}$};
  \node[fround] (fe0) at (2.10,-0.95) {$F_\bullet[1]$};
  \node[fround] (fe1) at (5.20,-0.95) {$F_\bullet[1]$};
  \draw[incidence] (fe0) -- (f0);
  \draw[incidence] (fe0) -- (f1);
  \draw[incidence] (fe1) -- (f1);
  \draw[incidence] (fe1) -- (fr);

  \node[readout] (ro) at (9.15,-0.95)
    {$C_\bullet^A[1]$ readout};
  \draw[factoryleg] (ro) -- (fr);
  \draw[dataleg] (ro) -- (c1);

\end{tikzpicture}
\caption{Independent online and offline histories joined at the terminal
readout.  The online data use copies of $C_\bullet$, while the resource
factory uses copies of $F_\bullet$.  After the ancilla-control CNOT, the
shifted ancilla readout copy attaches both to the final factory round and
to the online data round.  The factory attachment includes its base-code
coordinates; the data attachment is the identity contribution from
$X_A\mapsto X_AX_D$.  All arrows run from shifted intervals to rounds.}
\label{fig:offline-online-fault-complex}
\end{figure}

The terminal readout copy represents the prepared resource state.
Its attachments are the base-code
inclusion into $F_\bullet$ and the identity into the online $C_\bullet$.
Both attachments commute with the internal differentials. Consequently,
the off-diagonal terms in the square of the joined differential cancel,
and the interface defines a chain complex.

Logical correctness follows directly from the circuit identity.  Let
$\ket{x}_A$ denote an $X$-basis readout string and let
$x\in\{+1,-1\}$ be the product of its entries on the support of $M_X$.  The
ancilla-control CNOT gives
\begin{equation}
 {}_A\!\bra{x}\,\operatorname{CNOT}_{A\to D}
       \ket{\Omega_{M_X,\mu}}_A
   \propto \Pi^{M_X}_{\mu x,D}.
 \label{eq:ancilla-mediated-ppm-identity}
\end{equation}
Thus the main line is projected by $M_X$, while the offline outcome $\mu$ and
the terminal readout parity $x$ multiply to give the reported outcome.  The
data never pass through the preparation circuit. The same identity applies
componentwise to a commuting family of PPMs.

Unit scalar convolution therefore preserves the logical action of the
offline resource protocol. Establishing its full process distance additionally
requires control of spatial, temporal, mixed, and boundary faults.

\FloatBarrier

\section{Composition of logical readout gadgets}
\label{app:family-level-scaling}

We now consider successive logical measurements on the same code block.
The surgery constructions of
Appendices~\ref{app:sparse-surgery-constructions} and
\ref{app:efficient-fault-tolerance} can be inserted into the alternating
syndrome-extraction schedule of the base code. When the measurements act on
overlapping qubits, however, faults can extend across multiple gadgets, so
the distance of the combined process must be evaluated. We describe the
schedule and numerically verify distance preservation for an example with
two overlapping readouts that retain unmeasured logical information.

\subsection{Scheduling successive logical readouts}
\label{sec:readout-gadget-composition}

In the measurement-based implementation, syndrome extraction on the base
code alternates between $X$- and $Z$-check measurements. An $X$-type surgery
gadget performs fault-tolerant $X$-syndrome extraction together with the
desired logical $X$ measurements. We can therefore insert the gadget by
replacing an $X$-check round with an $X$-surgery round, retaining the adjacent
$Z$-check measurements. For two successive readouts, this gives
\begin{equation}
 \cdots\;\longrightarrow\;
 Z\text{ check}\;\longrightarrow\;
 X\text{ surgery }(M_1)\;\longrightarrow\;
 Z\text{ check}\;\longrightarrow\;
 X\text{ surgery }(M_2)\;\longrightarrow\;
 Z\text{ check}\;\longrightarrow\;\cdots .
 \label{eq:readout-mainline-schedule}
\end{equation}
The $X$ and $Z$ steps together form a full syndrome-extraction cycle.
In the fault complex of Appendix~\ref{app:efficient-fault-tolerance}, they
correspond to shift-$0$ rounds and shift-$1$ intervals, respectively.
Inserting a surgery gadget replaces a base-code round $C_\bullet$ by the
final round $F_\bullet$ of the resource-state preparation.

Figure~\ref{fig:two-readout-fault-complex} illustrates this schedule for the
two-readout example analyzed below. The
blue blocks describe $Z$-check measurements on the base code, while round 3
of each colored resource factory implements an $X$-surgery round. The
preceding rounds prepare the resource state offline, with time running
upward. The two resources can be prepared in parallel and then coupled to
the data in sequence, applying the appropriate frame updates and retaining
both measurement records. 

At the logical level, let $M_1,M_2\subseteq H_1(C_\bullet)$ denote the
measured subspaces in a common logical frame, including any measured
residual observables. If each readout preserves the remaining logical
information after its frame update, the two commuting readouts measure
$M_1+M_2$, leaving $k-\dim_{\mathbb F_2}(M_1+M_2)$ logical qubits. Repeated
measurements of the same observable give the same ideal eigenvalue, while
comparisons of their records can detect faults.

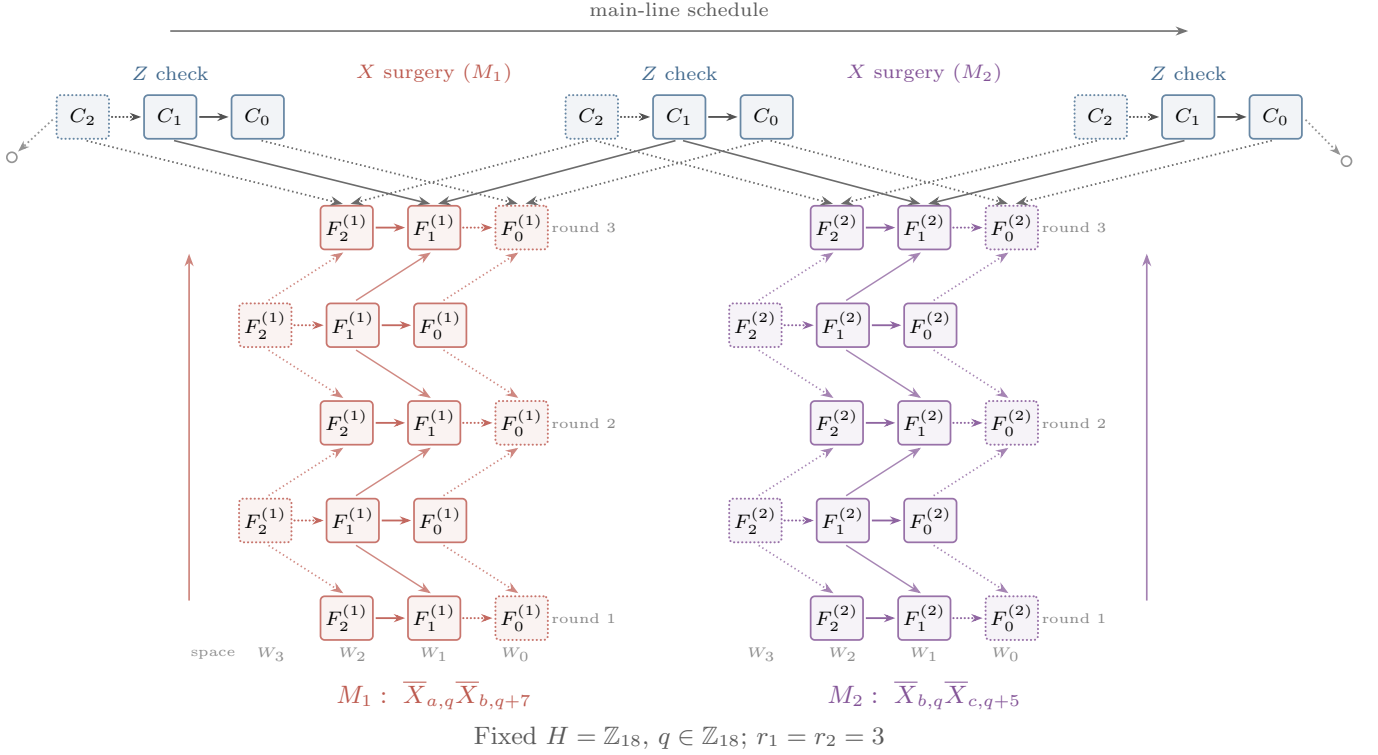
\begin{figure*}[t]
\centering
\resizebox{\textwidth}{!}{\begin{tikzpicture}[x=0.95cm,y=0.82cm]
  \tikzset{
    fcchain/.style={rounded corners=1.5pt, minimum width=6.6mm,
                    minimum height=5.5mm, inner sep=1pt,
                    font=\scriptsize},
    fccchainv/.style={fcchain, draw=basecol!68, fill=basecol!6},
    fccchaine/.style={fccchainv, dashed, dash pattern=on 2.2pt off 1.4pt},
    fcfachainv/.style={fcchain, draw=anccol!70, fill=anccol!6},
    fcfachaine/.style={fcfachainv, dashed, dash pattern=on 2.2pt off 1.4pt},
    fcfbchainv/.style={fcchain, draw=logcol!70, fill=logcol!6},
    fcfbchaine/.style={fcfbchainv, dashed, dash pattern=on 2.2pt off 1.4pt},
    fcphysical/.style={solid, line width=0.62pt},
    fcclassical/.style={dash pattern=on 0pt off 1.35pt, line cap=round,
                        line width=0.72pt},
    fcleg/.style={-{Stealth[length=1.6mm,width=1.2mm]}, black!58,
                  line width=0.58pt, shorten >=1.4pt, shorten <=1.4pt},
    fcvirt/.style={fcleg, densely dotted, line width=0.65pt},
    fcinternal/.style={-{Stealth[length=1.45mm,width=1.05mm]}, black!68,
                       line width=0.66pt, shorten >=0.7pt, shorten <=0.7pt},
    fcinternalvirt/.style={fcinternal, dash pattern=on 0pt off 1.25pt,
                           line cap=round, line width=0.72pt},
    fctimea/.style={fcleg, anccol!70, line width=0.52pt, opacity=0.86},
    fctimeav/.style={fctimea, densely dotted, line width=0.62pt},
    fctimeb/.style={fcleg, logcol!70, line width=0.52pt, opacity=0.86},
    fctimebv/.style={fctimeb, densely dotted, line width=0.62pt},
    fcport/.style={black!42, densely dotted, line width=0.65pt},
    fcportin/.style={fcport, -{Stealth[length=1.45mm,width=1.05mm]},
                     shorten <=3pt, shorten >=1.4pt},
    fcportout/.style={fcport, -{Stealth[length=1.45mm,width=1.05mm]},
                      shorten <=1.4pt, shorten >=3pt},
    fctiny/.style={font=\tiny, inner sep=0.8pt, text=black!58},
  }

  \newcommand{\fcResolvedZero}[6]{%
    \node[#2,fcphysical]  (#1two)  at ({#4-1.16},#5) {$#3_2$};
    \node[#2,fcphysical]  (#1one)  at (#4,#5)        {$#3_1$};
    \node[#2,fcclassical] (#1zero) at ({#4+1.16},#5) {$#3_0$};
    \draw[fcinternal,#6]     (#1two.east) -- (#1one.west);
    \draw[fcinternalvirt,#6] (#1one.east) -- (#1zero.west);
    \node[inner sep=0pt, fit=(#1two)(#1zero)] (#1) {};
  }
  \newcommand{\fcResolvedOne}[6]{%
    \node[#2,fcclassical] (#1two)  at ({#4-1.16},#5) {$#3_2$};
    \node[#2,fcphysical]  (#1one)  at (#4,#5)        {$#3_1$};
    \node[#2,fcphysical]  (#1zero) at ({#4+1.16},#5) {$#3_0$};
    \draw[fcinternalvirt,#6] (#1two.east) -- (#1one.west);
    \draw[fcinternal,#6]     (#1one.east) -- (#1zero.west);
    \node[inner sep=0pt, fit=(#1two)(#1zero)] (#1) {};
  }

  \newcommand{\fcTimeAUp}[2]{%
    \draw[fctimeav] (#1two.north)  -- (#2two.south);
    \draw[fctimea]  (#1one.north)  -- (#2one.south);
    \draw[fctimeav] (#1zero.north) -- (#2zero.south);
  }
  \newcommand{\fcTimeADown}[2]{%
    \draw[fctimeav] (#1two.south)  -- (#2two.north);
    \draw[fctimea]  (#1one.south)  -- (#2one.north);
    \draw[fctimeav] (#1zero.south) -- (#2zero.north);
  }
  \newcommand{\fcTimeBUp}[2]{%
    \draw[fctimebv] (#1two.north)  -- (#2two.south);
    \draw[fctimeb]  (#1one.north)  -- (#2one.south);
    \draw[fctimebv] (#1zero.north) -- (#2zero.south);
  }
  \newcommand{\fcTimeBDown}[2]{%
    \draw[fctimebv] (#1two.south)  -- (#2two.north);
    \draw[fctimeb]  (#1one.south)  -- (#2one.north);
    \draw[fctimebv] (#1zero.south) -- (#2zero.north);
  }

  \newcommand{\fcBackbone}[2]{%
    \draw[fcvirt] (#1two.south)  -- (#2two.north);
    \draw[fcleg]  (#1one.south)  -- (#2one.north);
    \draw[fcvirt] (#1zero.south) -- (#2zero.north);
  }

  \draw[-{Stealth[length=1.6mm,width=1.15mm]}, black!55, line width=0.6pt]
    (0.50,7.20) -- (14.00,7.20)
    node[midway,above=2pt,font=\scriptsize,text=black!65] {main-line schedule};
  \foreach \x in {0.50,7.25,14.00}{
    \node[font=\scriptsize,text=basecol!85] at (\x,6.55) {$Z$ check};}
  \node[font=\scriptsize,text=anccol!85] at (4.00,6.55)
    {$X$ surgery $(M_1)$};
  \node[font=\scriptsize,text=logcol!85] at (10.50,6.55)
    {$X$ surgery $(M_2)$};

  \fcResolvedOne{Cin}{fccchaine}{C}{0.50}{5.88}{black!68}
  \fcResolvedZero{F10}{fcfachainv}{F^{(1)}}{4.00}{4.18}{anccol!76}
  \fcResolvedOne{Cmid}{fccchaine}{C}{7.25}{5.88}{black!68}
  \fcResolvedZero{F20}{fcfbchainv}{F^{(2)}}{10.50}{4.18}{logcol!76}
  \fcResolvedOne{Cout}{fccchaine}{C}{14.00}{5.88}{black!68}

  \fcBackbone{Cin}{F10}
  \fcBackbone{Cmid}{F10}
  \fcBackbone{Cmid}{F20}
  \fcBackbone{Cout}{F20}

  \coordinate (pin) at ($(Cintwo.west)+(-0.58,-0.62)$);
  \draw[fcportout] (Cintwo.west) -- (pin);
  \fill[white] (pin) circle (1.8pt);
  \draw[black!42,line width=0.6pt] (pin) circle (1.8pt);
  \coordinate (pout) at ($(Coutzero.east)+(0.58,-0.68)$);
  \draw[fcportout] (Coutzero.east) -- (pout);
  \fill[white] (pout) circle (1.8pt);
  \draw[black!42,line width=0.6pt] (pout) circle (1.8pt);

  \fcResolvedOne{F1e0}{fcfachaine}{F^{(1)}}{2.92}{2.68}{anccol!76}
  \fcResolvedZero{F11}{fcfachainv}{F^{(1)}}{4.00}{1.18}{anccol!76}
  \fcResolvedOne{F1e1}{fcfachaine}{F^{(1)}}{2.92}{-0.32}{anccol!76}
  \fcResolvedZero{F12}{fcfachainv}{F^{(1)}}{4.00}{-1.82}{anccol!76}
  \fcTimeAUp{F1e0}{F10}
  \fcTimeADown{F1e0}{F11}
  \fcTimeAUp{F1e1}{F11}
  \fcTimeADown{F1e1}{F12}

  \fcResolvedOne{F2e0}{fcfbchaine}{F^{(2)}}{9.42}{2.68}{logcol!76}
  \fcResolvedZero{F21}{fcfbchainv}{F^{(2)}}{10.50}{1.18}{logcol!76}
  \fcResolvedOne{F2e1}{fcfbchaine}{F^{(2)}}{9.42}{-0.32}{logcol!76}
  \fcResolvedZero{F22}{fcfbchainv}{F^{(2)}}{10.50}{-1.82}{logcol!76}
  \fcTimeBUp{F2e0}{F20}
  \fcTimeBDown{F2e0}{F21}
  \fcTimeBUp{F2e1}{F21}
  \fcTimeBDown{F2e1}{F22}

  \draw[-{Stealth[length=1.6mm,width=1.15mm]}, anccol!62, line width=0.6pt]
    (0.75,-1.55) -- (0.75,3.78);
  \draw[-{Stealth[length=1.6mm,width=1.15mm]}, logcol!62, line width=0.6pt]
    (13.45,-1.55) -- (13.45,3.78);

  \foreach \y/\lab in {-1.82/1,1.18/2,4.18/3}{
    \node[fctiny, anchor=west, text=black!48] at (5.52,\y) {round $\lab$};}
  \foreach \y/\lab in {-1.82/1,1.18/2,4.18/3}{
    \node[fctiny, anchor=west, text=black!48] at (12.02,\y) {round $\lab$};}

  \foreach \x/\d in {1.84/3,2.92/2,4.00/1,5.08/0}{
    \node[fctiny, text=black!45] at (\x,-2.38) {$W_{\d}$};}
  \node[fctiny, text=black!45, anchor=east] at (1.40,-2.38) {space};
  \foreach \x/\d in {8.34/3,9.42/2,10.50/1,11.58/0}{
    \node[fctiny, text=black!45] at (\x,-2.38) {$W_{\d}$};}
  \node[font=\small,text=anccol!85] at (4.00,-3.00)
    {$M_1:\ \overline X_{a,q}\overline X_{b,q+7}$};
  \node[font=\small,text=logcol!85] at (10.50,-3.00)
    {$M_2:\ \overline X_{b,q}\overline X_{c,q+5}$};
  \node[font=\small,text=black!65] at (7.25,-3.65)
    {Fixed $H=\mathbb Z_{18}$, $q\in\mathbb Z_{18}$; $r_1=r_2=3$};
\end{tikzpicture}}
\caption{\textbf{Scheduling two logical readouts on the same code block.}
The main line alternates between $Z$-check measurements on the base code
(blue) and $X$-surgery rounds implementing $M_1$ and $M_2$ from
\cref{eq:fixed-subgroup-two-readouts} (red and purple). Each surgery
replaces an ordinary $X$-check round. Both adjacent base-code intervals
connect to the final round 3 of the corresponding resource factory. The
factories can be prepared in parallel, with time running upward; each has
three $X$ rounds and uses the same ordered pair of weight-three temporal
convolutions over $\mathbb Z_{18}$, specified in
\cref{eq:two-readout-temporal-couplings}. The remaining factory layers extend below the main line
without a lower cap. Columns labeled $W_j$ indicate chain degree. Solid
boxes denote qubit spaces: $F_2$ in rounds 1--3 is used for $X$-check
measurement, while $F_0$ on the intervals and $C_0$ on the main line are
used for $Z$-check measurement. Dotted boxes denote detector spaces. Solid
arrows connect physical spaces, and dotted arrows define detector parities.
External boundary completions for the three distance calculations in
\cref{eq:fixed-subgroup-composition-bounds} are omitted.}
\label{fig:two-readout-fault-complex}
\label{fig:two-readout-composition}
\end{figure*}

\FloatBarrier

\subsection{Process distance under composition}
\label{sec:two-readout-joint-protection}

Although each readout may individually preserve the code distance, faults
supported across both gadgets can have cancelling detector syndromes. Such
a fault can change a measurement outcome or a surviving logical observable
without being detected. Figure~\ref{fig:two-readout-mixed-fault} illustrates
the possible supports.

We analyze each fault sector using the boundary conditions and logical
observables defined in Appendix~\ref{app:efficient-fault-tolerance}. Let
$H_W$ be the binary detector matrix of the combined process, with columns
partitioned between the first resource, the shared data and interfaces,
and the second resource. An undetected fault satisfies
\begin{equation}
 H_W e=H_1e_1+H_Ce_C+H_2e_2=0.
 \label{eq:composed-detector-cancellation}
\end{equation}
Here $H_1,H_C,H_2$ are column blocks of $H_W$, rather than chain-degree
labels. A fault satisfying \cref{eq:composed-detector-cancellation}
contributes to the process distance if it also changes one of the logical
observables being protected.

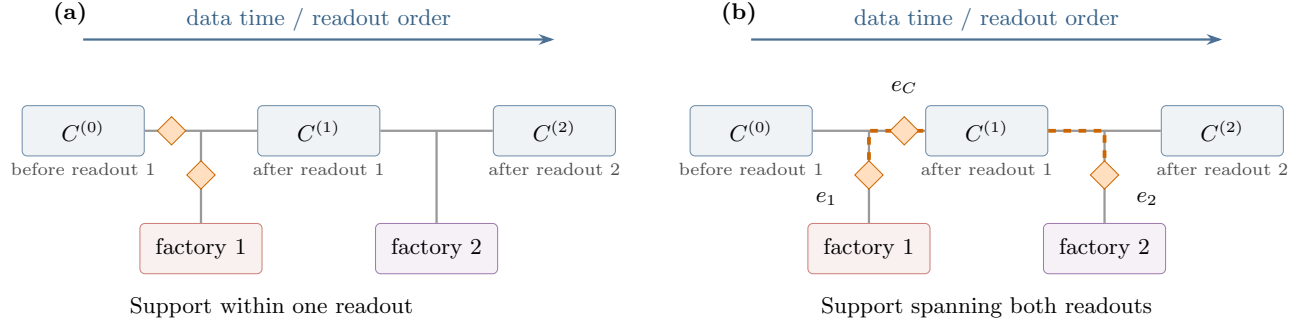
\begin{figure}[htbp]
\centering
\resizebox{0.96\textwidth}{!}{\begin{tikzpicture}[x=0.82cm,y=0.82cm,
  box/.style={rounded corners=2pt,minimum width=17mm,minimum height=7mm,
              inner sep=2pt,font=\small},
  data/.style={box,draw=basecol!65,fill=basecol!7},
  fa/.style={box,draw=anccol!65,fill=anccol!7},
  fb/.style={box,draw=logcol!65,fill=logcol!7},
  link/.style={line width=0.9pt,black!40},
  fault/.style={draw=orange!80!black,fill=orange!25,diamond,
                inner sep=1.6pt,minimum size=4mm},
  lab/.style={font=\small,align=center}]
\foreach \offset/\panel in {0/a,9.3/b}{
 \begin{scope}[xshift=\offset cm]
  \node[lab,anchor=west] at (-0.65,2.0) {\textbf{(\panel)}};
  \draw[-{Stealth[length=2.2mm,width=1.5mm]},basecol!80,line width=0.9pt]
    (0,1.55)--node[above,font=\small,text=basecol!90]
    {data time / readout order} (8,1.55);
  \node[font=\scriptsize,text=black!70] at (0,-0.65) {before readout 1};
  \node[font=\scriptsize,text=black!70] at (4,-0.65) {after readout 1};
  \node[font=\scriptsize,text=black!70] at (8,-0.65) {after readout 2};
  \node[data] (\panel c0) at (0,0) {$C^{(0)}$};
  \node[data] (\panel c1) at (4,0) {$C^{(1)}$};
  \node[data] (\panel c2) at (8,0) {$C^{(2)}$};
  \node[fa] (\panel f1) at (2,-2) {factory 1};
  \node[fb] (\panel f2) at (6,-2) {factory 2};
  \coordinate (\panel p1) at (2,0);
  \coordinate (\panel p2) at (6,0);
  \draw[link] (\panel c0)--(\panel c1)--(\panel c2);
  \draw[link] (\panel p1)--(\panel f1);
  \draw[link] (\panel p2)--(\panel f2);
 \end{scope}
}
\node[fault] at (2,-0.75) {};
\node[fault] at (1.5,0) {};
\node[lab] at (3.2,-3) {Support within one readout};
\begin{scope}[xshift=9.3cm]
 \draw[orange!80!black,line width=1.5pt,dashed]
   (2,-0.75)--(2,0)--(bc1.west);
 \draw[orange!80!black,line width=1.5pt,dashed]
   (bc1.east)--(6,0)--(6,-0.75);
 \node[fault] at (2,-0.75) {};
 \node[fault] at (2.6,0) {};
 \node[fault] at (6,-0.75) {};
 \node[lab,anchor=east] at (1.6,-1.15) {$e_1$};
 \node[lab,anchor=south] at (2.6,0.5) {$e_C$};
 \node[lab,anchor=west] at (6.4,-1.15) {$e_2$};
 \node[lab] at (4,-3) {Support spanning both readouts};
\end{scope}
\end{tikzpicture}}
\caption{\textbf{Fault supports across successive readouts.}
Blue denotes the shared data and red and purple the two resource factories.
Arrows indicate the order of operations on the data, from left to right.
The labels $C^{(0)},C^{(1)},C^{(2)}$ identify the data before, between, and
after the two readouts; factory preparation can proceed in parallel.
Orange indicates possible fault support. (a) A fault confined to one
readout, with boundary conditions inherited from the combined process.
(b) A fault supported on both readouts and the shared data, whose detector
syndromes may cancel as in \cref{eq:composed-detector-cancellation}.
The orange dashed line groups the support components. These sketches
illustrate the supports to be considered in the distance calculation;
they do not specify undetected faults of the construction in
Fig.~\ref{fig:two-readout-fault-complex}.}
\label{fig:two-readout-mixed-fault}
\end{figure}

Individual readout distances alone do not give a lower bound on the
process distance. Even faults confined to one readout must be evaluated
with the boundary conditions and protected logical observables of the
combined process. We therefore check the combined detector constraints
and logical action directly in the numerical example below, allowing
faults on both resources, the shared data, and their interfaces.

\FloatBarrier

\subsection{Numerical study of two-readout composition}
\label{sec:two-readout-finite-example}

We now evaluate the process distance for the two successive readouts in
Fig.~\ref{fig:two-readout-fault-complex}. Both act on a
$\code{288}{76}{12}$ base code and use the same subgroup
$H=\langle t^2\rangle\cong\mathbb Z_{18}$ of its $\mathbb Z_{36}$ symmetry.
We choose three active $H$-orbits $a,b,c$, each containing 18 logical qubits,
and measure the shifted pairs
\begin{equation}
 \begin{aligned}
 M_1&=\operatorname{span}_{\mathbb F_2}
       \{\overline X_{a,q}\overline X_{b,q+7}:q\in\mathbb Z_{18}\},\\
 M_2&=\operatorname{span}_{\mathbb F_2}
       \{\overline X_{b,q}\overline X_{c,q+5}:q\in\mathbb Z_{18}\}.
 \end{aligned}
 \label{eq:fixed-subgroup-two-readouts}
\end{equation}
In \cref{eq:fixed-subgroup-two-readouts}, $M_i$ denotes the active
measurement subspace. The two programs share orbit $b$, but their observables are
independent: the first program pairs $b$ with $a$, and the second pairs $b$
with $c$. Each has rank 18, giving 36 independent commuting measurements.

Only the measurement program changes between readouts; the subgroup and
physical coordinates remain fixed. Direct checks of the logical action
confirm that the combined construction measures the 36 active observables
and two residual observables, leaving 38 logical qubits. This example
therefore tests composition on overlapping logical orbits while retaining
unmeasured quantum information.
Both resource factories have three temporal
layers, with inhomogeneous weight-three temporal convolutions
\begin{equation}
 g_0=1+u+u^6,\qquad g_1=1+u^3+u^{10},\qquad u=t^2.
 \label{eq:two-readout-temporal-couplings}
\end{equation}
The two temporal edges have legs $(I,g_0)$ and $(I,g_1)$, respectively,
ordered from the layer attached to the data toward the free end of the
factory. Both polynomials are units in $\mathbb F_2[\mathbb Z_{18}]$,
and the same ordered pair is used in both factories.

To verify the process distance, we consider the measured $X$-PPM outcomes,
all logical $X$ observables, and the surviving logical $Z$ observables separately. Each family
uses the same interior fault complex, with boundary conditions and logical
observables chosen for that family. Let $H_T$ and $L_T$ be the sparse
detector and logical-observable matrices for family $T$. The protected
distance is
\begin{equation}
 d_T=\min\{\operatorname{wt}(e):H_Te=0,\ L_Te\ne0\}.
 \label{eq:two-readout-semantic-distance}
\end{equation}
The condition $L_Te\ne0$ selects faults that change the protected logical
information, excluding gauge degrees of freedom and logically trivial
spacetime faults~\cite{cowtan2024ssip,delfosse2023spacetime_v2,williamson2026gauging}.
For the composed process, the measured-outcome and surviving-$Z$ maps each
have 38 rows. The all-logical-$X$ map has 76 independent rows and tests
input--output correlations for the full input logical $X$ space.

We use the boundary completions in
Table~\ref{tab:process-distance-boundaries}. The measured-PPM calculation
uses $X$-basis preparation at the input and $Z$-basis readout at the output.
For all logical $X$ observables, both preparation and readout are in the $X$
basis; for surviving $Z$ observables, both are in the $Z$ basis.
We check all three target families for each readout separately and for
the composed process, using their respective external bare-code boundaries
and logical frames. Each individual readout measures $18+2$ active and
residual observables and preserves $54+2$ logical qubits. The
$d_{\mathrm{state}}^Z$ test protects all 76 logical $X$ observables in each
case.

\begin{table}[htbp]
\caption{Exact protected distances for each readout in
Fig.~\ref{fig:two-readout-fault-complex} separately and for their composition.
The combined process uses the alternating schedule of
\cref{eq:readout-mainline-schedule}. Each column uses its own boundary
completions and logical targets, as described in the text. Distances count
binary fault-complex cells; circuit-level distances require a physical fault model.}
\label{tab:two-readout-distance}
\centering
\begin{tabular}{lcccc}
\hline\hline
Protected target & Fault sector & First readout & Second readout & Composition\\
\hline
Measured $X$-PPM outcomes & $Z$ & $12$ & $12$ & $12$\\
All logical $X$ observables & $Z$ & $12$ & $12$ & $12$\\
Surviving $Z$ observables & $X$ & $12$ & $12$ & $12$\\
\hline\hline
\end{tabular}
\end{table}

Exhaustive searches exclude harmful undetected faults of weight below 12
in all nine tests. Bare-code operators give matching upper bounds. In
particular, using the distance notation of
Appendix~\ref{app:efficient-fault-tolerance}, the composed process has
\begin{equation}
 d_{\mathrm{PPM}}^X=d_{\mathrm{state}}^Z=d_{\mathrm{state}}^X=12.
 \label{eq:fixed-subgroup-composition-bounds}
\end{equation}

The closed $X$-fault complex also contains two interface homology classes
with weight-11 representatives in $W_1$. Neither changes any of the 38
surviving $Z$ observables, including the residual observables. Their
$X$-type action also commutes with all logical $X$ observables, including
the measured $X$-PPMs. These faults therefore act trivially on all three families
of protected observables, despite representing nonzero homology classes.
This is consistent with the treatment of additional surgery logicals as
an unprotected sector~\cite[Sec.~2.3.1]{cowtan2024ssip}, and illustrates why
the condition $L_Te\ne0$ is needed in
\cref{eq:two-readout-semantic-distance}. The two-readout construction
preserves distance 12 in the binary-cell metric without additional
syndrome-extraction rounds between readouts. Protection of arbitrarily long
sequences remains a separate question.

\FloatBarrier

\section{Practical design tradeoffs and numerical performance}
\label{app:practical-numerics}

We study how convolution affects distance and local degree in finite-size
surgery constructions. We first give one example construction for each code in
Table~\ref{tab:representative-merged-codes}, then examine spatial convolution for a fixed
$[[576,148,18]]$ code, where scalar choices give different distances at
the same qubit count. We then consider temporal convolution and its
dependence on the symmetry subgroup. We also give the circuit-sampling
methods for Fig.~\ref{fig:convolution-performance}(c) and the resource
accounting for Fig.~\ref{fig:structure-overhead}, before comparing the
resources for magic-state loading. The distance calculations below refer to the specified
static codes or fault sectors.

\subsection{Example constructions for the three base codes}
\label{sec:reproducing-numerical-instances}

We give one example construction for each of the three base codes in
Table~\ref{tab:representative-merged-codes}. The base-check matrices and
pairings below, together with the pivot assignments and convolution
polynomials in Table~\ref{tab:numerical-instance-parameters}, specify these
examples using the constructions of
Appendices~\ref{app:symmetry-logical-coordinates}
and~\ref{app:sparse-surgery-constructions}.

\paragraph*{Circulant convention.}
All indices start at zero. Write $C_q$ for the $q\times q$ cyclic shift
$C_qe_j=e_{j+1\bmod q}$. An exponent $s$ therefore represents the block
$C_q^s$, whose $(r,c)$ entry is one precisely when $r-c=s\pmod q$.
In the cyclic-ring notation, $x$ is represented by $C_P$, and binary
coordinate $bP+j$ is position $j$ in block $b$.

\paragraph*{CPM--PP codes.}
PP256 and PP576 each have $3\times8$ arrays of circulant permutation
blocks. Define their shift matrices $E_X^{(n)}$ and $E_Z^{(n)}$ by
$(H_X)_{ab}=C_P^{(E_X^{(n)})_{ab}}$ and
$(H_Z)_{ab}=C_P^{(E_Z^{(n)})_{ab}}$.
For PP256, $P=32$ and
\begin{equation}
\begin{aligned}
 E_X^{(256)}&=\begin{pmatrix}
 9&0&15&11&16&14&29&6\\
 10&18&23&14&6&18&26&23\\
 23&11&3&12&16&22&3&10
 \end{pmatrix},
 &E_Z^{(256)}&=\begin{pmatrix}
 23&30&3&6&30&9&17&4\\
 11&25&17&3&7&7&20&30\\
 31&29&12&19&24&29&12&28
 \end{pmatrix}.
\end{aligned}
\label{eq:pp256-check-shifts}
\end{equation}
For PP576, $P=72$ and
\begin{equation}
\begin{aligned}
 E_X^{(576)}&=\begin{pmatrix}
 50&69&70&40&5&39&5&0\\
 32&62&69&51&61&18&6&25\\
 33&37&30&47&46&63&23&23
 \end{pmatrix},
 &E_Z^{(576)}&=\begin{pmatrix}
 32&22&29&49&59&48&36&25\\
 40&67&60&70&69&37&69&30\\
 53&17&18&29&66&45&11&3
 \end{pmatrix}.
\end{aligned}
\label{eq:pp576-check-shifts}
\end{equation}
Every entry, including zero, specifies a permutation block; a zero
exponent denotes the identity matrix.

\paragraph*{APM code (APM1152).}
This code has $3\times12$ arrays of $96\times96$ permutation blocks.
In native coordinates, define the affine permutation matrix
\begin{equation}
 A(a,b)_{r,c}=\begin{cases}
 1,&c=ar+b\pmod{96},\\
 0,&\text{otherwise}.
 \end{cases}
 \label{eq:kasai-native-apm}
\end{equation}
The two families of native blocks are $A(a_i,b_i)$ and $A(c_i,d_i)$,
with parameters
\begin{center}
\begin{tabular}{c|cccccc}
 $i$&$0$&$1$&$2$&$3$&$4$&$5$\\ \hline
 $(a_i,b_i)$&$(61,15)$&$(1,24)$&$(89,62)$&$(25,22)$&$(85,93)$&$(25,78)$\\
 $(c_i,d_i)$&$(5,41)$&$(85,77)$&$(73,66)$&$(1,0)$&$(1,72)$&$(37,9)$
\end{tabular}
\end{center}
Their matrices after the common row and column reordering below are
$F_i$ and $G_i$, respectively. Each then factors as a $3\times3$
permutation tensored with a $32\times32$ circulant.
To specify this basis, write a native position
$z\in\mathbb Z_{96}$ in Chinese-remainder coordinates
$(a,y)=(z\bmod3,z\bmod32)$, and number the second coordinate by
\begin{equation}
 y_0=0,\qquad y_{j+1}=13y_j+3\pmod{32}.
 \label{eq:kasai-orbit-order}
\end{equation}
The $32$ values are distinct. Order each parent block by $(a,j)$,
with index $32a+j$, and use the same reordering for check rows.
Thus parent block $b$ splits into the three cyclic blocks $3b+a$ used
in Table~\ref{tab:numerical-instance-parameters}, with $x$ shifting $j$.

Let $J_3e_a=e_{2-a\bmod3}$ be the reflection of the three coordinates.
For example, $A(61,15)$ becomes $F_0=I_3\otimes C_{32}^{3}$ in this
basis. The twelve $96\times96$ blocks are
\begin{equation}
\begin{array}{c|cccccc}
 i&0&1&2&3&4&5\\ \hline
 F_i&I_3\otimes C_{32}^{3}&I_3\otimes C_{32}^{8}
 &J_3\otimes C_{32}^{10}&C_3^{2}\otimes C_{32}^{18}
 &I_3\otimes C_{32}^{5}&I_3\otimes C_{32}^{26}\\[2pt]
 G_i&J_3\otimes C_{32}^{9}&C_3\otimes C_{32}^{21}
 &I_3\otimes C_{32}^{6}&I_3\otimes I_{32}
 &I_3\otimes C_{32}^{24}&I_3\otimes C_{32}^{9}
\end{array}
\label{eq:kasai-tensor-blocks}
\end{equation}
With these blocks, the check matrices are
\begin{equation}
 H_X=\left(\begin{array}{cccccc|cccccc}
 F_0&F_1&F_2&F_3&F_4&F_5&G_0&G_1&G_2&G_3&G_4&G_5\\
 F_5&F_0&F_1&F_2&F_3&F_4&G_5&G_0&G_1&G_2&G_3&G_4\\
 F_4&F_5&F_0&F_1&F_2&F_3&G_4&G_5&G_0&G_1&G_2&G_3
 \end{array}\right),
 \label{eq:kasai-check-X}
\end{equation}
\begin{equation}
 H_Z=\left(\begin{array}{cccccc|cccccc}
 G_0^T&G_5^T&G_4^T&G_3^T&G_2^T&G_1^T&F_0^T&F_5^T&F_4^T&F_3^T&F_2^T&F_1^T\\
 G_1^T&G_0^T&G_5^T&G_4^T&G_3^T&G_2^T&F_1^T&F_0^T&F_5^T&F_4^T&F_3^T&F_2^T\\
 G_2^T&G_1^T&G_0^T&G_5^T&G_4^T&G_3^T&F_2^T&F_1^T&F_0^T&F_5^T&F_4^T&F_3^T
 \end{array}\right).
 \label{eq:kasai-check-Z}
\end{equation}

\begin{table}[htbp]
\centering
\small
\setlength{\tabcolsep}{8pt}
\renewcommand{\arraystretch}{1.2}
\caption{Pivot assignments and surgery parameters for one example construction
per base code. A range $a{:}b$
includes both endpoints. The first three coordinate groups partition the
base qubits into $X$-pivot, $Z$-pivot, and information
blocks.
Each pair-program triple $(a,b,s)$ pairs blocks $L_a$ and $L_b$ with
relative shift $s$, including all $P$ cyclic translations.
Polynomial exponents refer to
the convolution cyclic order listed below. Spatial coefficients use the
orientation $f z_1+g z_2$ of Appendix~C.3. Temporal legs are ordered as
(earlier round, later round), and $h_i$ is the polynomial on interval $i$.}
\label{tab:numerical-instance-parameters}
\begin{tabular}{lccc}
\hline\hline
 & PP256 & PP576 & APM1152 \\
\hline
Base code & $[[256,68,12]]$ & $[[576,148,18]]$ & $[[1152,580,12]]$ \\
Base cyclic block size $P$ & $32$ & $72$ & $32$ \\
Base check-array shape over $R$ & $3\times8$ & $3\times8$ & $9\times36$ \\
$X$-pivot blocks $S_X$ & $\{0{:}2\}$ & $\{0,1,7\}$ & $\{9{:}17\}$ \\
$Z$-pivot blocks $S_Z$ & $\{3{:}5\}$ & $\{3,4,5\}$ & $\{21{:}23,30{:}35\}$ \\
Information blocks $L$ & $\{6,7\}$ & $\{2,6\}$ & $\{0{:}8,18{:}20,24{:}29\}$ \\
Convolution cyclic order & $32$ & $72$ & $32$ \\
\hline
Pair program $(a,b,s)$ & $(0,1,13)$ & $(0,1,0)$ &
$\begin{gathered}
(6,11,10),\ (14,0,5)\\
(10,8,28),\ (3,9,23)\\
(16,4,12),\ (7,5,30)\\
(1,2,20),\ (13,15,15)\\
(17,12,25)
\end{gathered}$ \\
Spatial layers $m$ & $2$ & $2$ & $2$ \\
Spatial $f$ & $1+x+x^3$ & $1+x^4$ & $1$ \\
Spatial $g$ & $1$ & $1+x+x^3$ & $1+x^3+x^{15}$ \\
Merged code & $[[576,34,12]]$ & $[[1296,74,18]]$ & $[[2304,290,12]]$ \\
Listed merged checks $(m_X,m_Z)$ & $(352,192)$ & $(792,432)$ & $(1440,576)$ \\
Temporal rounds $r$ & $3$ & $4$ & $3$ \\
Temporal legs on interval $i$ & $(h_i,1)$ & $(h_i,1)$ & $(h_i,1)$ \\
$h_1$ & $1+x^{17}+x^{29}$ & $1+x^3+x^4$ & $1+x^{17}+x^{29}$ \\
$h_2$ & $1+x^{21}+x^{25}$ & $1+x^2+x^5$ & $1+x^{21}+x^{25}$ \\
$h_3$ & --- & $1+x+x^6$ & --- \\
\hline\hline
\end{tabular}
\end{table}

\paragraph*{Example pairings.}
Let $L$ list the information blocks in increasing order, as in the table,
and let $e_b$ denote a unit column at block $b$ over the cyclic ring.
A triple $(a,b,s)$ defines the pair-program column
$e_{L_a}+x^s e_{L_b}$, with all $P$
translations included. The attachment $\gamma_1$ consists of the $Z$-pivot unit columns
$e_b$ for $b\in S_Z$ and the listed pair-program columns. We take
$P_0=C_0$, $\gamma_0=I$, and $\partial_P=H_Z\gamma_1$.

\FloatBarrier

\subsection{Spatial convolution: distance and local-degree cost}
\label{sec:numerical-degree-distance}

We compare ordinary and convolved couplings for the two-layer PP576
construction in Fig.~\ref{fig:convolution-performance}(a), holding the
measurement program and the $1296$ physical qubits fixed. In addition to
the static $X$ distance, Table~\ref{tab:spatial-degree-comparison} reports
the maximum check weight and the maximum number of checks incident on a
qubit in each CSS sector.

\begin{table}[!htbp]
 \centering
 \small
 \setlength{\tabcolsep}{4pt}
 \renewcommand{\arraystretch}{1.2}
 \caption{Spatial distance and local-degree cost for two-layer PP576
 constructions, all with $n=1296$ and $k=74$.
 Overall maximum check weights and incidences ($X/Z$) appear to the
 left of the double rule; the nine component families are grouped by
 their role in the merged code to its right. Check columns give maximum weights; qubit columns give separate
 maximum $X/Z$ incidences. Superscripts $(1)$ and $(2)$ denote the
 attached and outer layers, respectively; unsuperscripted auxiliary
 components lie between them. Bold component entries attain the overall maximum
 check weight or qubit incidence in their respective CSS sector. Distances are attained static $X$ distances
 for selected scalars, not optima over the scalar classes.}
 \label{tab:spatial-degree-comparison}
 \begin{tabular}{cccc||ccc|cc|cccc}
  \hline\hline
  & & \multicolumn{2}{c||}{Overall maxima ($X/Z$)}
  & \multicolumn{3}{c|}{$X$-check weight}
  & \multicolumn{2}{c|}{$Z$-check weight}
  & \multicolumn{4}{c}{Qubit incidence ($X/Z$)} \\
  \cline{3-4}\cline{5-7}\cline{8-9}\cline{10-13}
  $(|f|,|g|)$ & $d_X$ & Check weight & Incidence
  & $C_2$ & $P_1^{(1)}$ & $P_1^{(2)}$
  & $C_0$ & $P_0$
  & $C_1$ & $P_0^{(1)}$ & $P_1$ & $P_0^{(2)}$ \\
  \hline
  $(1,1)$ & $7$ & $9/9$ & $5/6$
  & $8$ & $\mathbf{9}$ & $7$
  & $\mathbf{9}$ & $7$ & $4/3$ & $\mathbf{5}/2$ & $2/\mathbf{6}$ & $\mathbf{5}/1$ \\
  $(1,3)$ & $14$ & $9/9$ & $5/6$
  & $8$ & $\mathbf{9}$ & $\mathbf{9}$
  & $\mathbf{9}$ & $\mathbf{9}$ & $4/3$ & $\mathbf{5}/2$ & $4/\mathbf{6}$ & $\mathbf{5}/3$ \\
  $(2,3)$ & $18$ & $10/10$ & $5/6$
  & $8$ & $\mathbf{10}$ & $9$
  & $9$ & $\mathbf{10}$ & $4/3$ & $\mathbf{5}/3$ & $\mathbf{5}/\mathbf{6}$ & $\mathbf{5}/3$ \\
  \hline\hline
 \end{tabular}
\end{table}
\FloatBarrier

Table~\ref{tab:spatial-degree-comparison} shows that $(1,3)$ convolution
raises $d_X$ from $7$ to $14$ without increasing the qubit count, maximum
check weight, or maximum qubit incidence. 
The outer $P_1^{(2)}$ checks and interlayer $P_0$ checks both grow from
weight seven to nine, filling the slack below the ordinary check-weight
ceiling. On the qubits, the interlayer $P_1$ maximum $X$ incidence grows
from two to four and the outer $P_0^{(2)}$ maximum $Z$ incidence from one
to three, while the overall maxima remain five and six. Thus adding
translated couplings need not raise the maximum degree, although it
increases the total number of couplings.
Allowing maximum check weight ten with $(2,3)$ convolution then reaches
the base-distance ceiling $18$, still at the ordinary qubit-incidence
ceiling. 
These comparisons demonstrate the ability of convolution to dramatically 
grow the distance
with only a modest increase in local degree.

\FloatBarrier
\subsection{Distance growth in convolved repetition}
\label{sec:classical-qc-round-scaling}

Temporal convolution can build protection much faster than ordinary
repetition over the first few rounds. We study this effect in the classical
temporal detector code defined below, which
produces the dashed curves in Fig.~\ref{fig:convolution-performance}(b).
The same calculation explains why the available subgroup size limits this
advantage, connecting temporal protection to the programmability tradeoff
in Sec.~\ref{sec:reducing_structure}. Here we isolate the classical code;
the spatial sectors and base-code boundary enter the full-complex distance
calculations separately.

We start by explaining the precise classical distance calculation for the two examples in
Fig.~\ref{fig:convolution-performance}(b). Here, the
couplings lie in $H\cong\mathbb Z_{36}\subset\mathbb Z_{72}$. We take the interval-to-round attachments to be $(1,f_i)$, with $w=|f_i|$ and
\begin{equation}
\begin{array}{c|ccc}
 w & f_1 & f_2 & f_3\\ \hline
 3 & 1+x^{12}+x^{23} & 1+x^8+x^{17} & 1+x^7+x^{13}\\
 5 & 1+x^2+x^9+x^{12}+x^{17} & 1+x+x^{11}+x^{14}+x^{21} & ---
\end{array}
\label{eq:temporal-main-figure-polynomials}
\end{equation}
Here $x$ generates the order-$36$ subgroup, and each $r$-round point uses the first
$r-1$ intervals. 
The exact classical distances are $1,4,13,28$ for $w=3$ and $1,6,23$ for
$w=5$, established by explicit codewords and completed exclusions of all
lower weights. 

\paragraph{Rapid initial growth.}
More generally, we expect suitably chosen convolutions to give an initial
regime of approximately exponential distance growth, with the convolution
weight $|f|=w$ setting the base, as $w^r$. The intuition is that each
additional round can spread an error over roughly $w$ times as many
positions, before cancellations and the finite orbit size limit this growth.

To make this intuition precise, let $H$ be the cyclic subgroup supporting
the couplings, with order $P$, and write
$R=\mathbb F_2[x]/(x^P-1)$. In the forward recurrence convention, an undetected record-error
pattern obeys $u_t=f_tu_{t-1}$ across $r$ rounds, with $r-1$
invertible couplings $f_t\in R$. Its minimum binary Hamming weight is
\begin{equation}
 d_{\mathrm{cl}}(r)=\min_{u_0\ne0}\sum_{t=0}^{r-1}|u_t|,
 \qquad u_t=\left(\prod_{j=1}^{t}f_j\right)u_0.
 \label{eq:classical-qc-distance-definition}
\end{equation}
Ordinary repetition copies the same pattern
into every round and has $d_{\mathrm{cl}}(r)=r$. Convolution instead changes
its support from round to round, potentially increasing the total weight.

If every coupling has weight $w>1$, one occupied position can spread to
at most $w$ positions in the next round. Starting with $u_0=1$ therefore
gives the explicit word $(1,f_1,f_2f_1,\ldots)$ and the upper bound
\begin{equation}
 d_{\mathrm{cl}}(r)
 \leq\sum_{t=0}^{r-1}\left|\prod_{j=1}^{t}f_j\right|
 \leq\frac{w^r-1}{w-1}.
 \label{eq:classical-qc-geometric-envelope}
\end{equation}
This is the temporal counterpart of the classical-word construction in
Appendix~\ref{sec:convolution-distance-upper-bounds}. 

Different intervals use different polynomials for a reason. Reusing one
polynomial gives $f^2(x)=f(x^2)$ over $\mathbb F_2$, so two applications
have weight at most $w$, rather than $w^2$. At three rounds the same
single-position word then gives $d_{\mathrm{cl}}\leq1+2w$.

\paragraph{The limit imposed by subgroup size.}
The relevant orbit size is $P=|H|$, since short logicals are only spread to a single subgroup orbit. 
Since each round contains only $P$ positions, combining this observation
with the geometric bound immediately gives
\begin{equation}
 d_{\mathrm{cl}}(r)\leq
 \min\left\{\frac{w^r-1}{w-1},\,rP\right\}.
 \label{eq:classical-qc-group-ceiling}
\end{equation}
At fixed subgroup size, exponential growth cannot persist as rounds are
added: distance is at most linear in $r$. Smaller subgroups limit the
space available for spreading, explaining why finer programmability can
require more rounds of protection. Sharper linear bounds depend on the
cyclic group, as illustrated by the reference lines below.

\paragraph{Numerical evidence for the crossover.}
Figure~\ref{fig:classical-qc-crossover} follows random sequences of
weight-three or weight-five unit polynomials for up to $128$ rounds.
We sample distinct-exponent supports uniformly, conditioned on
$\gcd(f_t,x^P-1)=1$, independently between intervals; shorter sequences
are prefixes of the same longer sequence. For each sequence, we enumerate
all $2^P-1$ nonzero initial patterns, giving exact distances. These samples
are not selected for high distance.

The figure plots $d_{\mathrm{cl}}(r)/(rP)$, so a plateau corresponds to
linear distance growth. The curves initially rise and then approach
group-dependent ceilings. In these samples, increasing the convolution
weight accelerates the initial growth and the approach to this regime,
without removing the linear ceiling.
The practical advantage comes from reaching the target distance
within the first few rounds, while the subgroup orbit still provides room
for rapid spreading, and using a sufficiently large subgroup. This is the regime used by the short constructions
in the main text.

\begin{figure*}[t]
 \centering
 \includegraphics[width=\textwidth]{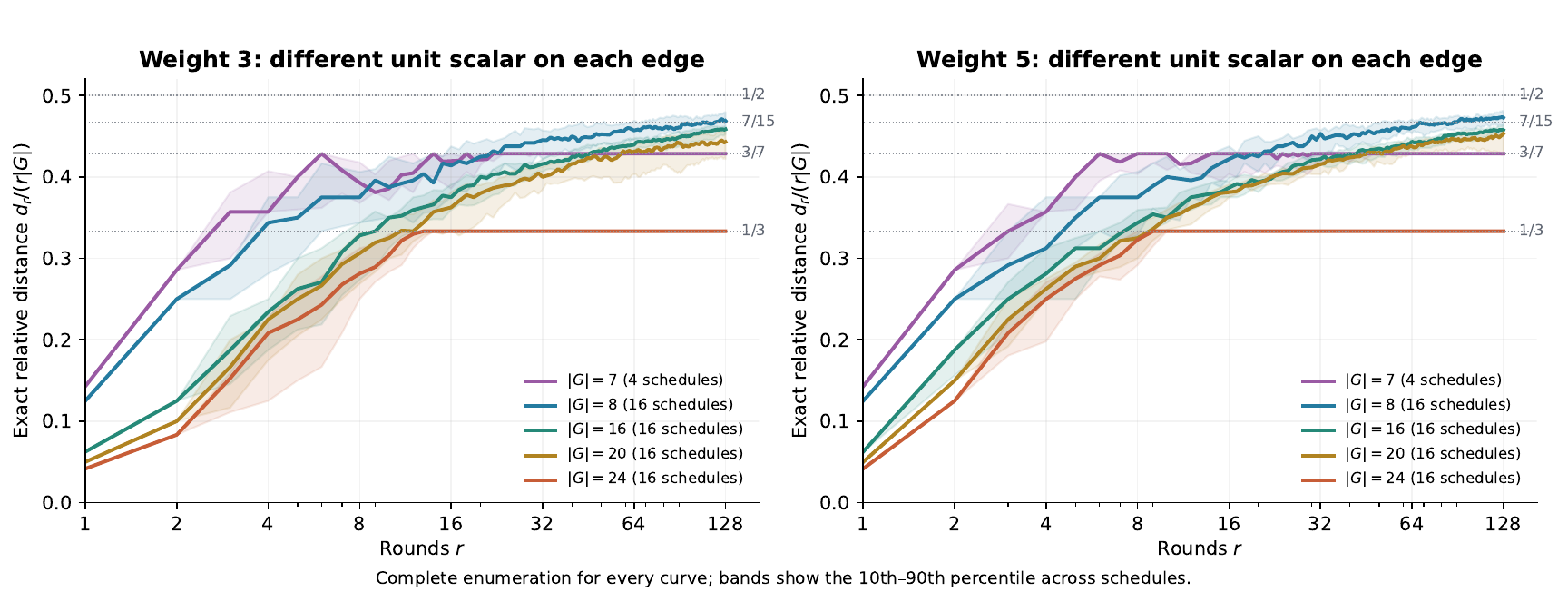}
 \caption{\textbf{Distance growth in convolved repetition.}
 Exact relative distance $d_{\mathrm{cl}}(r)/(rP)$ for random sequences of
 invertible convolutions of weight three (left) and five (right).
 Curves show medians and bands the 10th--90th percentiles over 16 sequences
 per group and weight, except $P=7$, which uses four. Every distance is
 obtained by enumerating all nonzero initial patterns. A plateau means
 linear growth with rounds. Horizontal lines give sharper upper bounds:
 $1/2$ for $P=8,16$, $3/7$ for $P=7$, $7/15$ for $P=20$, and $1/3$ for
 $P=24$. These follow by averaging over a nonzero-word orbit under all
 invertible convolutions: each round permutes this set, so the average
 total weight is $r$ times its mean word weight and bounds the minimum
 from above. The lines use the smallest such mean, divided by $P$.
 These are classical-code distances, without spatial sectors or
 base-code boundaries.}
 \label{fig:classical-qc-crossover}
\end{figure*}

\FloatBarrier

\subsection{Circuit-derived logical-error estimates}
\label{sec:circuit-numerical-methods}

Figure~\ref{fig:convolution-performance}(c) compares complete circuit-level
simulations in the measurement-based implementation. These experiments
exercise logical measurement and preservation of the unmeasured logical
sector, providing a circuit-level performance test of the construction.
They use block readout rather than the pairwise measurement programs of
Table~\ref{tab:numerical-instance-parameters}.

\paragraph*{Constructions and noise model.}
For PP256, both experiments use the same two-layer spatial gadget: the
attachment selects physical blocks $3,4,5,6$, with blocks $3,4,5$ supplying
the pivot coordinates and block $6$ the readout orbit. The spatial
coefficients are $f=1+x+x^3$ and $g=1$ over $\mathbb F_2[\mathbb Z_{32}]$.
The repetition experiment uses twelve merged rounds with identity temporal
couplings. The convolved experiment uses three rounds, with temporal legs
$(1+x^{17}+x^{29},1)$ and $(1+x^{21}+x^{25},1)$, ordered as
(earlier round, later round). Both circuits include the bare-code boundary
records, preparation and terminal readout, and the corresponding input and
output checks. Their CNOT schedules are fixed before the reported sampling.

The PP576 experiment reads out one full orbit of $72$ logicals while
preserving the other orbit. It uses a two-layer spatial gadget with
$f=1+x$ and $g=1+x^{10}+x^{24}$ over $\mathbb F_2[\mathbb Z_{72}]$.
Its four-round temporal complex has legs $(1,h_i)$ with
$h_1=h_3=1+x^4+x^{24}$ and $h_2=1+x^{12}+x^{28}$.
This is a separate construction from the PP576 examples in
Fig.~\ref{fig:convolution-performance}(a,b).

After each CNOT, a two-qubit depolarizing channel applies a uniformly
chosen nonidentity two-qubit Pauli with total probability $p_{2q}=10^{-3}$.
Preparation and measurement faults have probability $p_{\mathrm{SPAM}}=10^{-3}$:
a $Z$ error follows an $X$-basis preparation or precedes an $X$ measurement,
and an $X$ error is used for the corresponding $Z$-basis operations.
Although no explicit idle-error channel is included, preparation and
measurement faults on the resource qubits contribute effective data noise
in the circuit interpretation of MBQC. Here their combined outcome-flip
probability is approximately $2p_{\mathrm{SPAM}}=2\times10^{-3}$, before
including two-qubit gate faults.
We construct the detector error model
(DEM) from each noisy Stim circuit~\cite{gidney2021stim},
and sample detector events and observable flips from that DEM. 

\paragraph*{Decoding and failure definition.}
For our fast and reliable decoding, we use a three-tier belief-propagation-based~\cite{mueller2025improved,zhao2026towards} policy, with identical tuning for all three experiments.
The first tier uses a single-basis detector family only, while the second and third tiers use the joint detector model. The gurobi integer-programming solver~\cite{gurobi2026manual} is used for final decoding of shots that fail to converge.

A shot fails if the final recovery disagrees with any included logical
observable. Scoring includes the two additional measured residual parities
as well as the requested readouts and surviving logical information. Thus
these are complete-experiment failure rates. None of our reported shots remained unresolved.

\begin{table}[!htbp]
\centering
\small
\caption{Samples and final failures used in Fig.~\ref{fig:convolution-performance}(c).
The failure columns label the input basis of the experiment, not the
Pauli type of the logical error.}
\label{tab:circuit-sampling-counts}
\begin{ruledtabular}
\begin{tabular}{lrrrr}
Construction & Rounds & Shots per basis & $X$-basis failures & $Z$-basis failures \\
PP256 repetition & $12$ & $20\,000\,000$ & $3$ & $34$ \\
PP256 convolved & $3$ & $20\,000\,000$ & $1$ & $14$ \\
PP576 convolved & $4$ & $100\,000\,000$ & $0$ & $0$
\end{tabular}
\end{ruledtabular}
\end{table}
\FloatBarrier

A flipped logical $Z$ outcome signals an error with a logical $X$ component;
a flipped logical $X$ outcome signals an error with a logical $Z$ component.
The larger failure count in the $Z$-basis experiments therefore should not
be read as a predominance of logical $Z$ errors. Since each shot is scored
against all included observables, these aggregate counts do not resolve
which logical observable failed.

With equal sample counts in the two bases, the plotted rate
$(k_X+k_Z)/(N_X+N_Z)$ estimates the mean of the two basis-specific failure
probabilities. Figure error bars are two-sided $95\%$ Wilson intervals
on pooled counts, used as an approximate summary of the two equal-sized
strata. For zero observed failures, the estimate remains zero and only the
interval upper endpoint is shown. These sampling results do not establish
an exact circuit distance or an optimal-decoder failure rate.

\subsection{Subgroup programmability and resource estimates}
\label{sec:subgroup-numerical-methods}

Figure~\ref{fig:structure-overhead} compares the resources needed for
$72$ pairwise measurements on $144$ active logicals of PP576 as the
allowed symmetry subgroup $H\leq\mathbb Z_{72}$ is reduced. At each
subgroup order, the active logicals form $144/|H|$ subgroup orbits.
A sampled program pairs these orbits and assigns a relative cyclic shift
to each pair, giving $72/|H|$ translated measurement families. The total
number of target measurements is fixed; this study changes their pairing
and batch size rather than puncturing measurements from a fixed program.

\paragraph*{Spatial search and certification.}
Both search methods preserve $H$ symmetry and impose maximum
$(X\text{-check weight},X\text{ incidence},Z\text{-check weight},Z\text{ incidence})$
no larger than $(11,6,11,6)$. The scalar search varies the number of layers
and their scalar couplings, increasing the layer count when the tested
choices fail to reach the target. The search is random rather than an exhaustive optimization over
all couplings.

The matrix search also permits paired maps satisfying
Eq.~\eqref{eq:general-convolution-square}. Its homotopy modifications have
the form $S_1=fI+h\partial_P$, $S_0=fI+\partial_Ph$, with
$h:P_0\to P_1$, so the commuting square is preserved. Short logical
witnesses guide further modifications, including symmetry-preserving
ancilla additions and check splitting to satisfy the degree cap~\cite{zheng2025high}. The red
curve represents this combined design procedure.

These designs have verified distances $18$ with matching upper and lower bounds. 
We also verify commutation, the intended measured subspace,
and the surviving logical dimension $k=74$. 

\paragraph*{Resource accounting and sample coverage.}
For each matching and method, we use the smallest  merged code found. Each plotted point is the largest cost among the sampled
matchings at that subgroup order. The physical footprint is
\begin{equation}
 n_{\mathrm{phys}}=n_{\mathrm{merged}}+m_X+m_Z,
 \qquad \text{space per participating logical}=n_{\mathrm{phys}}/144,
\end{equation}
where $m_X,m_Z$ count listed checks, including dependent generators, with
one measurement ancilla per check and no ancilla reuse.

The series contains three matchings at each displayed subgroup
order. The matrix optimization is
stochastic, and poorly converged solutions with cost higher than the scalar value are omitted from the plotted curve. 

\FloatBarrier

\paragraph*{Temporal protection and spacetime accounting.}
Panel~(c) multiplies the spatial footprint by the number of syndrome
rounds $r$ needed to reach the target temporal distance $18$. Candidate
scalar convolutions are obtained using the detector code of
Appendix~\ref{sec:classical-qc-round-scaling}: the search samples up to
$5000$ schedules at each subgroup order and round count, with independent
weight-three unit couplings and verified $d \geq 18$ lower bounds.
The scalar and matrix designs use the same round count at fixed $|H|$,
as summarized in Table~\ref{tab:subgroup-round-counts}.

\begin{table}[!htbp]
\centering
\small
\caption{Syndrome rounds for the complexes in in Fig.~\ref{fig:structure-overhead}(c)
to reach temporal distance $18$. The counts are identical for both scalar and
matrix designs.}
\label{tab:subgroup-round-counts}
\renewcommand{\arraystretch}{1.15}
\setlength{\tabcolsep}{5pt}
\begin{ruledtabular}
\begin{tabular}{l*{9}{c}}
Subgroup order $|H|$ & $1$ & $2$ & $3$ & $4$ & $6$ & $8$ & $9$ & $12$ & $18$ \\
\hline
Syndrome rounds $r$ & $18$ & $18$ & $18$ & $10$ & $9$ & $6$ & $6$ & $5$ & $4$
\end{tabular}
\end{ruledtabular}
\end{table}
\FloatBarrier

For $|H|\leq3$, no weight-three unit exists, and identity repetition gives
$r=18$. 

The plotted spacetime cost is $r n_{\mathrm{phys}}/144$. The dashed
$a/|H|$ curve is an unweighted least-squares fit to the scalar costs at
$|H|=4,6,8,9,12,18$, giving $a=1258.1$. Its extension to smaller subgroups
is an extrapolation.

\FloatBarrier

\subsection{Magic-state loading resource accounting}
\label{sec:magic-loading-accounting}

Table~\ref{tab:pp256-surface-injection-resources} compares three concrete loading constructions at source distance seven. We count allocated code and attachment qubits, excluding check-measurement ancillas and resource preparation. For a batch of $B$ states, the per-state count is
\begin{equation}
 n_{\mathrm{load}}/B=(n_{\mathrm{source}}+n_{\mathrm{destination}}+n_{\mathrm{interface}})/B.
\end{equation}
The interface includes any intermediate code and all surgery bridges. Both bridges are allocated in the two-stage constructions, so these totals are not optimized peak hardware counts.

For our lifted direct injection construction, sixteen rotated $[[49,1,7]]$ patches couple to a $[[256,68,12]]$ block through $384$ attachment qubits. Thus
\begin{equation}
 n_{\mathrm{load}}=16\cdot49+256+384=1424,
 \qquad n_{\mathrm{load}}/16=49+16+24=89.
\end{equation}
The stated static construction has $d_X=7$, $d_Z\geq12$, maximum check weight $10$, and maximum total qubit degree $7$. 
These static-code parameters are distinct from the protection of the temporally convolved loading schedule discussed below.

For the LP transistor construction of Ref.~\cite{zheng2026logical} (Sec.~V.2), each batch loads $66$ states using two $33$-patch stacks, a $1122$-qubit destination, a $1749$-qubit transistor, and bridges of $396$ and $99$ qubits. Consequently,
\begin{equation}
 n_{\mathrm{load}}=66\cdot85+1122+1749+396+99=8976,
 \qquad n_{\mathrm{load}}/66=85+17+34=136.
\end{equation}
For the mitten construction of Ref.~\cite{bhardwaj2026high} (Appendix~E), the distance-seven source contributes $85$ qubits per state, the magic port contributes $3\cdot7-1=20$, and the bridges contribute $6+1$. Choosing the $[[200,40,12]]$ destination gives
\begin{equation}
 n_{\mathrm{load}}=40(85+20+6+1)+200=4680,
 \qquad n_{\mathrm{load}}/40=117.
\end{equation}
Both reference constructions use unrotated patches, whose product structure supports their LP interfaces. The patch geometry and fraction of the destination loaded differ from our example, so total-space ratios do not isolate the loading mechanism. The comparison also does not optimize direct attachments for the reference codes.

The time entries count syndrome rounds in the sequential joint-measurement stages, excluding preparation and terminal readout. Zheng et al. use two stages scaling with source distance (Sec.~V.2 and Appendix~D of Ref.~\cite{zheng2026logical}); seven rounds per stage gives the $14$-round baseline. We assume the same repetition schedule for the two switches in the mitten protocol, rather than infer a measured injection latency from Ref.~\cite{bhardwaj2026high}. 
The three-round value for our construction comes from optimizing $(1,f_i)$ temporal convolutions with $|f_i|=3$. Looking at the distance with respect to the measured $X$ (eq. $Z$)-PPMs, we confirm the fault-distance becomes limited by the input distance of 7. These round counts therefore compare joint-measurement stages rather than total loading latency.

\section{Asymptotic construction for symmetric surgery}
\label{app:pk-transverse-products}

In this section, we explore the implications of performing symmetric surgery
with symmetric codes through asymptotic constructions.

Our starting point is the Panteleev--Kalachev (PK) family of lifted-product
constructions, which combine expander graphs and local codes to obtain
sparse complexes with strong distance guarantees~\cite{panteleev2022almostlinear,panteleev2022asymptotically}.
Expansion also supplies short fillings: small undetectable ancillary
errors can be removed while adding at most a constant multiple of their
weight to the error on the base code.
For joint spatial and temporal protection, we use product-expanding local
codes and high-dimensional cubical expanders to extend this control to
errors supported across several directions~\cite{kalachev2025extendable_doi,dinur2024cubical_doi}.
These constructions let us prove that symmetric surgery can achieve a
distance proportional to the base-code distance with only polylogarithmic
overhead in both space and total spacetime resources, under the symmetry
and sparsity assumptions stated below.

Appendix~\ref{si-pk:sec:two-page-proof} first proves the
squared-logarithmic spatial-overhead bound for a general sparse starter.
Appendix~\ref{sec:pk-temporal-distance} then gives a jointly chosen base
code and protection factors that support flexible surgery programs with
polylogarithmic spatial and spacetime overhead.
Appendix~\ref{sec:asymptotic-supporting-results} collects the supporting
constructions and filling estimates used in both proofs.

The starter and attachment conventions are those of
Appendix~\ref{app:sparse-surgery-constructions}; the protected fault
sectors are those of Appendix~\ref{app:efficient-fault-tolerance}.
All products are balanced over the same commutative ring
$R=\mathbb F_2[G]$, with the group specified in each result, and all
weights count expanded binary coordinates. The bounds concern static
code distance and phenomenological fault-complex distance, respectively.

\subsection{Distance guarantee from two transverse factors}
\label{si-pk:sec:two-page-proof}

In this section, we present a construction of a merged code
with an asymptotic distance guarantee, that leverages the symmetry structure of the
base code and logical operator, in order to achieve sub-linearly scaling
space overhead.
The key connection to our primary construction is that in both cases
the ancilla complex is attained by a lifted product of the starter complex
$P_1 \rightarrow P_0$ and a transverse classical code.
For our practical constructions, we have used lifts of short repetition codes
to achieve the target distance. That is not what is proven here.
The purpose of this proof is simply to show that symmetries are also useful for better scaling bounds than bare repetition alone.
The construction below is not practical.
Nevertheless, it helps to justify that symmetries and the general structure we study is also asymptotically efficient.
This also provides confidence that the procedure will continue to support good performance
at larger distances. In particular, we show below how the Panteleev-Kalachev construction
can be applied to build a different complex $V$, such that the merged code
with $C$ attached to $P \otimes_{R} V$
has a distance guarantee.
The spatial overhead required
is only $O(\log^2 d)$ in the limit of a sufficiently large group
$|G| = \Theta(d)$, and furthermore this assumes no expansion or additional structure in the base code or starter complex
beyond the group symmetry.

\begin{theorem}[Squared-logarithmic spatial overhead]
\label{si-pk:thm:two-factor-spatial-distance}
Let $C_2\to C_1\to C_0$ have a logically faithful starter
$D=\partial_P:P_1\to P_0$ with attachment $\gamma$, and set
$d=d_X(C_\bullet;J)<\infty$ as in
Eq.~\eqref{eq:surviving-base-x-distance}. Assume based free modules over
$R=\mathbb F_2[G]$, with $G$ abelian of power-of-two order $g\ge2$,
and uniformly bounded physical row and column weights. Assume also
that $\operatorname{coker}D\to H_0(C)$ is injective.

There is a bounded-degree merged code, measuring the same logical
subspace, with $d_X=d$ and $O(m)$ times the starter size in added
physical qubits and checks, where
\[
 m=O\!\left(\max\{1,\log^2g,(d/g)^2\}\right).
\]
No expansion or classical-distance assumption on the starter is needed.
If $g=\Theta(d)$ and the starter has size $O(\text{base size})$, this
is $O(\log^2d)$ spatial overhead.
\end{theorem}

\begin{proof}
We start by defining an ancilla complex $V$ built from a product of two
classical codes.
Then, we show how to remove the connecting ancillary error using a short transverse filling,
then use transverse distance to eliminate the remaining ancillary support.

\paragraph*{1. Build a transverse code with short fillings.}
Fix a sufficiently large Tanner graph degree $w$ and compatible binary
local codes of length $w$, independently of $g,m,d$.
Use two PK factors with $\Theta(\sqrt m)$ orbit
coordinates each. Consider an invertible maximal minor and one extra
column in each, obtaining split surjections $A',B'$ with rank-one
kernels and normalized information coordinates. Their lifted product is
\[
 V_0\xrightarrow{H_1}V_{-1}\xrightarrow{H_2}V_{-2},\qquad
 H_1z=(A'z,z(B')^{\mathsf T}),\quad
 H_2(u,v)=u(B')^{\mathsf T}+A'v.
\]
Each module has $O(m)$ orbit coordinates. The properties we use are
\begin{equation}
 \begin{gathered}
 \ker H_2=\im H_1,\qquad \ker H_1=Rq,\quad q_1=1,
 \qquad d(\ker H_1)=d_H=\Omega(gm),\\
 H_2b=0,\ |b|<\tau\ \Longrightarrow\
 \exists z:\ H_1z=b,\quad |z|\le\rho|b|,
 \qquad \tau=\Omega_w(g\sqrt m),\quad\rho=w.
 \end{gathered}
 \label{si-pk:eq:compact-transverse}
\end{equation}
PK's block filling estimate converts to $|z|\le w|b|$ in binary
weight; Lemma~\ref{si-pk:lem:explicit-rho} makes this conversion explicit
\cite[Lemma~1 and Proposition~1]{panteleev2022asymptotically}.
Shortening preserves this bound: a filling outside the retained
rectangle would contain a constituent kernel word below its distance.
Split surjectivity gives exactness; $d_H\ge d_A d_B/g$.
Proposition~\ref{prop:pk-transverse-package} below supplies these inputs
and the additional coordinate-minor requirement, with parameter
dependence explicit in Appendix~\ref{si-pk:sec:construction-parameters}.

\paragraph*{2. Couple the base to the first transverse coordinate.}
Form $E=P\otimes_R V$ and couple the base code only to the first
transverse coordinate, chosen so that $q_1=1$. An X error has coordinates
$(e,a_0,a_1)$ on
$C_1\oplus(P_0\otimes_R V_0)\oplus(P_1\otimes_R V_{-1})$.
The added X stabilizers and zero-syndrome equations are
\begin{equation}
 z\longmapsto(\gamma_1z_1,Dz,H_1z),\qquad
 \partial_1e=\gamma_0(a_0)_1,\quad
 H_1a_0=Da_1,\quad H_2a_1=0.
 \label{si-pk:eq:compact-cycles}
\end{equation}
The last equation is imposed by physical Z checks; both ancillary
qubit sectors are retained. Since $V$ is homotopy equivalent to $R$
through its first coordinate, the mapping-cone exact sequence
identifies the merged logical space with the intended bare quotient
\cite[Eq.~(8)]{ide2025homological}. The assumed injectivity
excludes additional logical classes.

\paragraph*{3. Fill the connecting error, then eliminate the remainder.}
Let $a=(a_0,a_1)$, let $w_D$ bound the column weight of $D$, and set
\[
 \tau_{\rm clean}=\min\{\tau,d_H/\max\{1,w_D\rho\}\}=\Omega(g\sqrt m).
\]
For an undetectable error with $|a|<\tau_{\rm clean}$, the equation $H_2a_1=0$
and \eqref{si-pk:eq:compact-transverse} give
$H_1z=a_1$ with $|z|\le\rho|a_1|$.
Apply the transverse filling separately to each free $P_1$ coordinate;
the physical weights add. Multiplying by the added X stabilizers
specified by $z$ gives
\[
 (e,a_0,a_1)\sim(e+\gamma_1z_1,\ a_0+Dz,\ 0).
\]
The remaining ancillary component is a transverse kernel word because
$H_1(a_0+Dz)=Da_1+DH_1z=0$. But
\[
 |a_0+Dz|\le |a_0|+w_D\rho|a_1|
 \le\max\{1,w_D\rho\}|a|<d_H.
\]
Thus it vanishes in every free $P_0$ coordinate. We have removed
\emph{all} ancillary support using at most $\rho|a|$ added generators.
This supplies Proposition~\ref{prop:static-distance-repair}'s removal
condition with $\lambda=\max\{1,w_\gamma\rho\}$, where $w_\gamma$
bounds the column weight of $\gamma_1$. Only the bounded column weight of $D$ entered the argument.

Before the optional replication in the next step, this gives
\[
 d_X(F)\ge\min\{\tau_{\rm clean},d_X(C_\bullet;J)/\lambda\},
 \qquad \lambda=\max\{1,w_\gamma\rho\}.
\]
The removal argument applies to any transverse complex with the stated
properties and any fixed filling cost $\rho$; the construction used
here has $\rho=w$.

\paragraph*{4. Preserve the full distance and count resources.}
Choose $\sqrt m$ a sufficiently large constant times
$\max\{1,\log g,d/g\}$, so that $\tau_{\rm clean}\ge d$.
To absorb the constant cleaning cost, replace every ancillary qubit
in both sectors by $k$ copies, for a fixed integer
$k\ge\max\{1,w_\gamma\rho\}$. Each X check acts on all its copies;
old Z checks act on first copies, and adjacent-copy Z equality checks
force an undetectable X error to be constant within each block.
This preserves commutation, the logical quotient, and bounded degrees.
The physical weight is now $|e|+k|a|$.

Every undetectable error below weight $d$ therefore cleans to a bare
representative satisfying
\[
 |e+\gamma_1z_1|
 \le |e|+w_\gamma\rho|a|
 \le |e|+k|a|<d.
\]
It is trivial in the surviving quotient. Conversely, a surviving bare
logical of weight $d$ remains nontrivial after attachment. Hence
$d_X=d$.

The modules $V_0,V_{-1},V_{-2}$ each have $O(m)$ orbit coordinates,
so the added qubits and checks number $O(m)$ times the starter size.
Tensoring over $R$ shares the group coordinate and introduces no extra
factor of $g$. Constant repetition changes only the prefactor.
The choice above gives the stated $m$ and completes the proof.
\end{proof}

With additional compatible product-expansion structure inherited from the
base code and preserved by programming, we expect one new transverse factor can
suffice, improving $O(\log^2d)$ to $O(\log d)$ spatial overhead.
Here we focus on the general construction: beyond logical correctness and
sparsity, it uses only the stated group structure of the code and target
logical measurements, without an expansion assumption on the base code
or starter. The one-factor improvement requires a separate joint filling
hypothesis; it does not follow from ordinary base-code expansion alone.

\subsection{Joint spatial and temporal distance guarantees}
\label{sec:pk-temporal-distance}

We now present a construction with a full spacetime distance guarantee.
This involves selecting a base code and the factors providing spatial and temporal
protection together. The surgery program is chosen afterward, and it needs
only to respect the symmetry, have bounded physical degrees, and implement
the intended logical measurements. The proof removes short ancillary
errors and reduces the remaining fault to one protected by the base code.

Use $E=P\otimes_R V$ and $F=\operatorname{Cone}(E\to C)$ as in
Appendix~\ref{app:sparse-surgery-constructions}. The spatial complex $V$
has degrees $0,-1,-2$ and the temporal complex $T$ has degrees $0,1,2$.
All products share the group coordinate. Write $N_X$ for the total number
of binary cells of a complex $X$. Distances in this subsection count
binary cells, rather than elementary circuit faults.

\begin{theorem}[Base codes supporting flexible surgery programs]
\label{thm:joint-pk-spacetime}
Let $G$ be a cyclic group and write $g=|G|$ for its order.
Fix bounds on the physical row and column weights of the starter
differential and attachment maps, and on the size ratio $N_P/N_C$.
For sufficiently large power-of-two values of $g$, the following
construction is possible. Write $\ell=\Theta(\log g)$ and
$d_*=g\ell$ for its distance scale.

\begin{enumerate}
\item \textbf{Choose the base code and protection factors.}
Choose a binary PK-derived base code $C$, a spatial complex $V$,
and a temporal complex $T$ jointly, all using the same cyclic lift
group $G$. Each of $V$ and $T$ is formed by
taking the lifted product of two sparse classical-code complexes over
$R=\mathbb F_2[G]$. The base code has
$N_C=\Theta(g\ell^2)$ cells and distance $\Omega(d_*)$.

\item \textbf{Program the measurements.}
Choose any logically faithful starter $P=(P_1\to P_0)$ with free
$G$-action on its physical coordinates, and an
attachment $\gamma:P\to C$ that are $R$-linear and satisfy the sparsity
and size bounds above. The chosen $C,V,T$ work for every such program;
the starter need not inherit their product structure.

\item \textbf{Add spatial and temporal protection.}
Form the spatial auxiliary, attach it to the base code, and then
thicken the merged complex in time:
\[
 E=P\otimes_R V,\qquad
 F=\operatorname{Cone}(E\to C),\qquad
 W^\circ=F\otimes_R T.
\]
Here $W^\circ$ is the fault complex before adding the boundary copies
that specify preparation and readout. Complete these boundaries using
the original base code $C$, as specified below.
\end{enumerate}

\noindent\textbf{Distance guarantees.}
The merged code has $d_X(F)=\Omega(d_*)$. The completed fault complexes
protect the measured $X$ outcomes, all logical $X$ observables, and
the surviving logical $Z$ observables, including residual logicals:
\[
 \begin{aligned}
 d_{\mathrm{PPM}}^X&=\Omega(d_*),
     &&\text{measured $X$ outcomes},\\
 d_{\mathrm{state}}^X&=\Omega(d_*),
     &&\text{surviving $Z$ observables},\\
 d_{\mathrm{state}}^Z&=\Omega(d_*),
     &&\text{all logical $X$ observables}.
 \end{aligned}
\]
Here the superscripts on the state distances denote the Pauli type
of the fault, while the superscript on the outcome distance denotes
the measured observable.

\medskip\noindent\textbf{Resource guarantees.}
Let $N_W$ be the largest cell count among the three boundary completions.
The cell counts obey
\[
 N_F=O(N_C+\ell^2N_P),\qquad
 N_W=O\bigl(\ell^2(N_C+\ell^2N_P)\bigr).
\]
Since $N_P=O(N_C)$, the spatial overhead and total spacetime overhead,
both measured in binary cells relative to the original base code, are
\[
 \underbrace{\frac{N_F}{N_C}}_{\text{spatial overhead}}
   =O(\log^2d_*),\qquad
 \underbrace{\frac{N_W}{N_C}}_{\text{total spacetime overhead}}
   =O(\log^4d_*).
\]
\end{theorem}

Now, we explain the key ideas and general approach we use to prove the theorem.
The goal is to show that a short undetectable fault cannot change either a
measurement outcome or a surviving logical observable. We do this by
removing (or deforming) parts of the fault without changing its logical effect, until
what remains is protected by the base code. The essential requirement is
that each removal step increases the weight of the fault by at most a constant
factor. Then, the fault-complex inherits a reduced but constant fraction of the
base code distance.

The PK construction provides the needed removal operations, called
\emph{fillings}~\cite{panteleev2022almostlinear}.
Importantly, we choose the base code and spatial and temporal
complexes together so that these operations also work for errors that
extend across space and time. The starter $P_{\bullet}$ then specifies which logical
measurements to perform. Its sparse couplings ensure that removing one
part of an error creates only a bounded amount of error on the other
parts. This is why the same protection complexes can accommodate
different programs without requiring the starter to inherit their
product structure. Appendix~\ref{sec:joint-pk-construction} establishes
the joint construction, and Lemma~\ref{lem:joint-pk-starter} proves
that its filling estimates survive coupling to the starter.

We apply this approach separately to faults that change the measured
outcomes and faults that change any logical $X$ observable or a surviving
logical $Z$ observable. Each test includes boundary copies of the original base code
to represent its preparation and readout conditions. For surviving
$Z$ observables, the closed completion retracts to the interior without
increasing fault weight. For measured outcomes and logical $X$
correlations, the boundary equations identify the interior coboundary
that can be removed. Each case reduces to a distance bound supplied by
the base code or the joint PK construction. The exact boundary
definitions and logical-effect tests are given in
Appendix~\ref{sec:joint-pk-boundary-definitions}.

These bounds count binary cells and phenomenological faults in the
completed fault complexes. The temporal factor $T$ is a two-dimensional
complex, so the spacetime bound concerns total resource size and does
not specify a number of syndrome-extraction rounds. Physical scheduling
and circuit-level fault propagation require a separate analysis.
The construction establishes asymptotic scaling without estimating
practical numerical constants.

\begin{proof}[Proof of Theorem~\ref{thm:joint-pk-spacetime}]
We call $W^\circ$ the \emph{interior}: it includes both the base-code
and ancillary coordinates, but not the additional preparation and
readout copies of $C$. Write $B=C\otimes_R T$ for its base-code part
and $Q=(E\otimes_R T)[1]$ for its ancillary part. The shift $[1]$
places degree-$j$ ancillary coordinates in degree $j+1$ of the fault
complex. The boundary copies are defined in
\eqref{eq:joint-pk-boundaries}; we explain their roles in each case below.
In each removal below, we add a stabilizer or a fault equivalence:
this changes where the error is supported, but preserves its syndrome
and its effect on every protected logical observable. All additions
are binary, so adding the same error twice cancels it.
Choose a constant $a>0$ such that every filling guarantee used below
applies to errors of weight less than $a d_*$. Decrease $a$ if needed
so that $a d_*$ is also below the base-code distance and the degree-one
distance of $B$ used below.
Let $\Lambda\ge1$ bound the cumulative increase in error weight during the
removal steps.
Starting with a fault of weight less than $a d_*/\Lambda$ therefore keeps
every intermediate error within the required filling thresholds.

\medskip\noindent\textbf{Spatial protection: move the error onto the base code.}
Write a spatial error as $(e,\alpha)$, where $e$ is on the base code and
$\alpha$ is on the auxiliary system. Proposition~\ref{prop:joint-pk-input}(1)
supplies the properties of $V$ used in the removal argument of
Appendix~\ref{si-pk:sec:two-page-proof}.
That argument uses only the filling cutoff and a constant filling cost,
so it applies with the $\rho$ supplied here, which need not equal the
earlier spatial construction's Tanner degree $w$.
It selects a combination of added $X$ stabilizers whose
ancillary support is exactly $\alpha$. Multiplying by them cancels $\alpha$,
while adding some error $\Delta e$ on the base code:
\[
 (e,\alpha)\sim(e+\Delta e,0),\qquad |\Delta e|=O(|\alpha|).
\]
The remaining error has zero syndrome and the same surviving logical
effect, but weight $O(|e|+|\alpha|)$. Below the base-code distance it is
trivial, proving $d_X(F)=\Omega(d_*)$. We use Step 3 of
Appendix~\ref{si-pk:sec:two-page-proof}, before its optional ancillary
replication; a constant-factor reduction in distance is allowed here.

\medskip\noindent\textbf{Surviving $Z$ observables: cancel the ancillary part of an $X$ fault.}
For this test, $W^{\rm cl}$ denotes the \emph{closed} boundary
completion of $W^\circ$, used to protect surviving $Z$ observables.
It adds two copies of $C$, one in the original degrees and one shifted
by one degree, joined as in \eqref{eq:joint-pk-boundaries}.
First reduce a fault of $W^{\rm cl}$ to one on $W^\circ$.
This includes faults on both added copies:
the reduction in Appendix~\ref{sec:joint-pk-boundary-definitions}
preserves their effect on surviving $Z$ observables without increasing
weight. It therefore suffices to analyze the resulting interior fault.

Write an undetected fault on $W^\circ_1=B_1\oplus Q_1$ as $(b,q)$,
where $b$ is the base-code component and $q$ is the ancillary component.
Its zero-syndrome equations are
$\partial_Qq=0$ and $\partial_Bb+\Gamma q=0$, where $\Gamma$ is
the attachment map. The first equation lets us apply
Lemma~\ref{lem:joint-pk-starter}: there is a choice of generators
$h$ with $\partial_Qh=q$ and $|h|=O(|q|)$.
In the full fault complex these generators have support
$(\Gamma_2h,\partial_Qh)$, so adding them gives
\[
 (b,q)+(\Gamma_2h,\partial_Qh)=(b+\Gamma_2h,0).
\]
Thus the ancillary error disappears, leaving only
$b'=b+\Gamma_2h$ on $B_1$. Each generator selected by $h$ acts on
at most a constant number of base-code cells. Since $|h|=O(|q|)$,
the added error $\Gamma_2h$ has weight $O(|q|)$, and hence
$|b'|=O(|b|+|q|)$.

The remaining error $b'$ still has zero syndrome, because we added
a fault equivalence. Proposition~\ref{prop:joint-pk-input}(2) says
that an undetected error on $B_1$ with a nontrivial logical effect
must have weight $\Omega(d_*)$. For a sufficiently short initial
fault, $b'$ is below this distance and cannot flip any surviving
$Z$ observable. Every step preserved those flips, so the original
fault cannot flip them either. This proves
$d_{\rm state}^X=\Omega(d_*)$.

\medskip\noindent\textbf{Measured $X$ outcomes: first cancel the base-code part of a $Z$ fault.}
For measured outcomes, $W^{\rm PPM}$ adds one shifted boundary copy
of $C$ to $W^\circ$; the superscript identifies the PPM-outcome test.
Write a degree-two $Z$ fault as $(\beta,q;\zeta)$: its three components lie on
the base-code part $B^2$, the ancillary part $Q^2$, and $C^1$ on the
boundary copy, respectively. The boundary component has degree one
because that copy is shifted. Let $j:C\to B$ insert the base code at
the distinguished temporal coordinate. Undetectability gives
$\delta_B\beta=0$ and $j^*\beta=\delta_C\zeta$.
The latter equation says that the restriction of $\beta$ to the base
code is logically trivial. The temporal construction preserves these
logical classes by Proposition~\ref{prop:joint-pk-input}(1), so
$\beta$ itself can be removed: it is a coboundary.

Proposition~\ref{prop:joint-pk-input}(2) now supplies generators $h$
with $\delta_Bh=\beta$ and $|h|=O(|\beta|)$. Their full fault
equivalence has three components,
$(\delta_Bh,\Gamma_2^\dagger h;j^*h)$. Adding it cancels $\beta$,
but changes both the ancillary error and the boundary error:
\[
 (\beta,q;\zeta)\sim
 (0,q';\zeta'),\qquad
 q'=q+\Gamma_2^\dagger h,\quad \zeta'=\zeta+j^*h.
\]
Both changes have weight $O(|\beta|)$, since the attachment and
restriction maps have bounded physical weights.

With the base-code component now zero, the remaining detection
equations say $\delta_Qq'=0$ and $\delta_C\zeta'=0$.
Lemma~\ref{lem:joint-pk-starter} gives $q'=\delta_Qk$.
Generators supported on $Q$ contribute only $(0,\delta_Qk;0)$,
so a second removal leaves
\[
 (0,q';\zeta')\sim(0,0;\zeta').
\]
This step changes neither the base-code component nor the boundary
error. We are left with an undetected $Z$ error on the original base
code, of weight $O(|\beta|+|q|+|\zeta|)$ and with the same measured
outcome flips as the original fault. A nonzero flip requires weight
at least $d_Z(C)$, proving $d_{\mathrm{PPM}}^X=\Omega(d_*)$.

\medskip\noindent\textbf{All logical $X$ observables: remove interior faults between two attachments.}
Use the two temporal representatives $e_0,e_1$ and connecting chain
$t_{01}$ defined in Appendix~\ref{sec:joint-pk-boundary-definitions}.
Write an undetected degree-two fault on $W^{\rm copy}$ as
$(\beta,q;\zeta_0,\zeta_1)$. Its boundary equations are
\[
 \delta_B\beta=0,\qquad
 j_a^*\beta=\delta_C\zeta_a\quad(a=0,1),
\]
where $j_a(x)=x\otimes e_a$. Since $j_0^*$ induces an isomorphism
on cohomology, $\beta$ is a coboundary. As in the measured-outcome
case, Proposition~\ref{prop:joint-pk-input}(2) supplies a short
primitive $h$, and the full fault equivalence is
\[
 (\delta_Bh,\Gamma_2^\dagger h;j_0^*h,j_1^*h).
\]
It removes $\beta$ at constant cost, including both boundary changes.
The remaining ancillary cocycle is removed by
Lemma~\ref{lem:joint-pk-starter}, leaving
$(0,0;\zeta'_0,\zeta'_1)$ with $\delta_C\zeta'_a=0$ and
\[
 |\zeta'_0+\zeta'_1|
 \le |\zeta_0|+|\zeta_1|+(\|j_0^*\|+\|j_1^*\|)|h|
 =O(|\beta|+|q|+|\zeta_0|+|\zeta_1|).
\]
The observable cycle associated to a bare logical representative $x$
is $(x\otimes t_{01};x;x)$, with its first component in $B_2$.
After removal its pairing is
$\langle\zeta'_0+\zeta'_1,x\rangle$.
A nonzero pairing therefore requires a bare-code cocycle of weight
at least $d_Z(C)$. This proves $d_{\rm state}^Z=\Omega(d_*)$.
The constant loss includes the two attachment norms; an exact lower
bound of $d_Z(C)$ is not asserted.

\medskip\noindent\textbf{Resources.}
Each of $V,T$ has $O(\ell^2)$ orbit coordinates. Forming $E$ therefore
costs $O(\ell^2N_P)$ cells, attaching $C$ adds $N_C$, and temporal
thickening multiplies the count by $O(\ell^2)$. The boundary copies
add $O(N_C)$. Products over $R$ share the group coordinate and introduce
no extra factor of $g$. Since $\ell=\Theta(\log d_*)$, the claimed
spatial and total spacetime overheads follow.
\end{proof}

\subsection{Supporting constructions and proofs}
\label{sec:asymptotic-supporting-results}

\subsubsection{The PK construction input}
\label{si-pk:sec:pk-input}
\label{si-pk:sec:construction-parameters}

We now establish the construction inputs used in the preceding proofs.
Taking the lifted product of two two-term factors gives a three-term
transverse complex: its middle map supplies consistency checks, while
its top map has both large kernel distance and short fillings of small
valid syndromes. We first state this package, then give the full-factor
expansion input, prove that coordinate shortening preserves it, and
establish compatible invertible coordinate minors.
Appendix~\ref{sec:joint-pk-construction} supplies the additional joint
construction needed for spatial and temporal protection.

In this two-factor construction, $G$ is a finite abelian $2$-group of
order $g\ge2$. The graph degree $w$ and local constituent matrices are
chosen once and kept fixed as $g$ and the transverse size grow.

\begin{proposition}[Two-factor transverse construction]
\label{prop:pk-transverse-package}
For $\sqrt m$ at least a sufficiently large fixed multiple of
$\max\{1,\log g\}$, there is a based free, bounded-degree complex
\[
 V_0\xrightarrow{H_1}V_{-1}\xrightarrow{H_2}V_{-2}
\]
with $O(m)$ orbit coordinates in each degree, whose only homology is
$H_0(V)=\ker H_1=Rq$, with first coordinate $q_1=1$.
Evaluation at the first coordinate is a chain-homotopy equivalence to $R$ in degree
zero. There are constants $c_1,c_2>0$, depending only on the fixed
local construction, such that
\[
 d(\ker H_1)\ge c_1gm,\qquad
 H_2b=0,\quad |b|<c_2g\sqrt m
 \quad\Longrightarrow\quad
 \exists z:\ H_1z=b,\quad |z|\le w|b|.
\]
The row and column weights are bounded independently of $g,m$.
\end{proposition}
The proposition follows from the full-factor input, rectangular
shortening, and Lemma~\ref{si-pk:lem:transverse} below. It is a property
of the transverse complex alone; using it in a fault complex still
requires identifying the relevant fault sector and boundary conditions.

Write $R=\mathbb F_2[G]$. The parameter $m$ measures transverse overhead;
it does not denote the base-code length. In the input below, each
factor has $\Theta(\sqrt m)$ orbit coordinates, so their product has
$\Theta(m)$ orbit coordinates. Constant rounding of dimensions is
implicit.

\begin{proposition}[Compatible full factors: construction input]
\label{si-pk:prop:pk-input}
For $\sqrt m$ at least a sufficiently large constant times
$\max\{1,\log g\}$, there are two sparse $R$-linear maps $A,B$, each
with $\Theta(\sqrt m)$ rows and columns and more columns than rows,
with the following simultaneous properties.
\begin{enumerate}
\item Each map has an invertible maximal square coordinate minor.
\item Their classical kernel distances satisfy
\[
 d_A,d_B\ge c_{\mathrm{dist}}g\sqrt m.
\]
\item Every boundary $(AY,YB^{\mathsf T})$ of total physical weight
less than $\tau_{\mathrm{PK}}$ has a filling $Z$ with the same boundary
and
\[
 |Z|\le\rho\bigl(|AY|+|YB^{\mathsf T}|\bigr),
 \qquad \tau_{\mathrm{PK}}\ge c_{\mathrm{PK}}g\sqrt m.
\]
\end{enumerate}
Here $c_{\mathrm{dist}},c_{\mathrm{PK}}>0$ and the physical degree
bounds depend on the fixed local construction, independently of
$g,m$ and the starter. One may take $\rho=w$, where $w$ is its fixed Tanner graph degree.
The transpose in these array formulas only indicates placement of
coordinates; any group-ring adjoint used to define $B$ already includes
the inversion of group elements.
\end{proposition}

\begin{proof}[Construction ingredients]
Fix local rate intervals strictly below and above $1/2$, respectively.
First fix a spectral constant $L$ from the random cyclic-lift
bound~\cite[arXiv Theorem~2]{agarwal2019grouplifts}, for the fixed
$O(\sqrt w)$ base-graph spectral constant below, and put $\mu=L+1$.
Then choose a sufficiently large fixed degree $w$ and binary local
matrices $h,h'$ using
\cite[Proposition~1]{panteleev2022asymptotically}, including the
weight-enumerator conditions of
Lemma~\ref{si-pk:lem:coordinate-minors}.
The degree and local codes are independent of $g$ and $m$.

Let $N_0=\Theta(\sqrt m)$ be even and sufficiently large.  The
graph construction in the proof of
\cite[Proposition~1]{panteleev2022almostlinear} supplies a simple
$w$-regular graph $\Gamma_0$ on $N_0$ vertices whose nonleading
eigenvalues have absolute value $O(\sqrt w)$. Enlarge their bound
$\lambda_0$ to at least $\sqrt w$. For sufficiently large $w$ it
satisfies $\lambda_0\le w/\log w$, the additional hypothesis for
the cited additive $O(\sqrt w)$ bound on new lift eigenvalues.
For $N_0\ge C_w\max\{1,\log g\}$, this graph has a $G$-lift
$\widehat\Gamma$ with nonleading eigenvalues bounded by $L\sqrt w$.
Here is the extension from cyclic to arbitrary finite abelian $G$.
For a nontrivial character $\chi$ of $G$ with image of order $k_\chi$, independently
uniform edge voltages induce independently uniform $k_\chi$th roots of
unity in its character block.  This block has the distribution of one
nontrivial block of a random cyclic $k_\chi$-lift.  The cyclic-lift
theorem bounds its failure probability by
$k_\chi\exp(-cN_0/w^2)$.  A union bound over characters therefore
bounds the probability of any bad block by
$g^2\exp(-cN_0/w^2)$.  Increasing $C_w$ makes this probability less
than one.  No independence between character blocks is needed.
The theorem's additional size condition
$w\le\sqrt{N_0/(3\log N_0)}$ is absorbed in the fixed lower bound
on $N_0$.

Set $N=gN_0$ and $a=N/\sqrt w$.  Spectral mixing gives
$|E(S,T)|\le\mu\sqrt w\sqrt{|S||T|}$ whenever
$|S|,|T|\le a$.  Thus the precise graph hypothesis of
\cite[Proposition~1]{panteleev2022asymptotically} holds, and its
locally minimal one-cycle distance is at least
\[
 a/(2w)=gN_0/(2w^{3/2})=\Omega_w(g\sqrt m).
\]
Its Tanner factors are the wide map $A$ from $h$ and $B=T^\dagger$,
where $T$ is the tall map from $h'$.  The local code and dual-code
distance conditions imply the expansion bounds of
\cite[Lemma~3]{panteleev2022almostlinear}; their kernel-distance
consequence, noted before Proposition~1 of that reference, gives
$d(\ker A),d(\ker T^\dagger)=\Omega_w(gN_0)$.
Both Tanner factors have $\Theta_w(N_0)$ rows and columns.
The binary local-descent argument in
Lemma~\ref{si-pk:lem:explicit-rho} then supplies the stated filling
cutoff and cost $\rho=w$.

Finally choose independent local coordinate permutations as in
Lemma~\ref{si-pk:lem:coordinate-minors}.  The PK local properties
hold for every such permutation choice, and the lemma gives simultaneous
full augmentation ranks with positive probability on the fixed base
graph.  Since $G$ is an abelian $2$-group, nonsingular augmented minors
are invertible over $R$.  This proves all three properties simultaneously.
\end{proof}

\paragraph*{Construction parameters and the explicit filling constant}
Choose a fixed $w$-regular graph construction and full-row-rank binary
local matrices $h\in\mathbb F_2^{r\times w}$ and
$h'\in\mathbb F_2^{r'\times w}$, with fixed rates
$r/w<1/2<r'/w$. The full factors are the wide Tanner map $A$
and $B=T^\dagger$, where $T$ is the tall Tanner map from $h'$.
The local pair must satisfy the PK product-expansion and distance
conditions, together with the coordinate-minor construction below.
These choices precede taking the asymptotic family.

\begin{center}
\begin{tabular}{c p{0.77\linewidth}}
Parameter & Meaning and dependence\\\hline
$w$ & Fixed graph degree and local code length; chosen sufficiently large.\\
$g$ & Order of the symmetry group already acting on the starter.\\
$m$ & Transverse product size in orbit coordinates; each factor has
$\Theta_w(\sqrt m)$ coordinates.\\
$\rho$ & Binary filling-cost bound: choose $\rho=w$, unchanged by shortening.\\
$\tau$ & Shortened filling threshold
$\min\{\tau_{\mathrm{PK}},d_A/w,d_B/w\}$.\\
\end{tabular}
\end{center}

\begin{lemma}[Binary filling cost at most the local degree]
\label{si-pk:lem:explicit-rho}
For the full factors above, every top boundary $b$ with
$|b|<\tau_{\mathrm{PK}}$ has a filling $z$ with
$|z|\le r'|b|\le w|b|$, in the physical binary bases.
\end{lemma}
\begin{proof}
In the local-system representation of $A\otimes T^*$, each top
correction block is $\mathbb F_2^{r'}$
\cite[Sec.~2.3]{panteleev2022asymptotically}. Write $|\cdot|_{\rm blk}$ for the number
of nonzero local blocks. Choose any $c$ with $\partial c=b$.
The second assertion of \cite[Lemma~1]{panteleev2022asymptotically}
gives $d_{\rm blk}(c,\ker\partial)\le |b|_{\rm blk}$ below the
locally minimal distance. Subtract a nearest kernel word to obtain a
filling $z$ with $|z|_{\rm blk}\le |b|_{\rm blk}$. Each filling block
has at most $r'$ binary coordinates, so
\[
 |z|\le r'|z|_{\rm blk}\le r'|b|_{\rm blk}\le w|b|.
\]
The assumed physical cutoff implies the block cutoff since
$|b|_{\rm blk}\le|b|$. The block counts are per lifted cell, not per
group orbit, so there is no factor of $g$ in this conversion.
\end{proof}

Rectangular shortening below $d_A/w,d_B/w$ keeps this filling inside
the retained coordinates (Lemma~\ref{si-pk:lem:shortening}); it introduces
no extra factor in $\rho$. Taking $\rho=w$ throughout is therefore valid.

\paragraph*{Dependence on the fixed local construction.}
All asymptotic constants may depend on $w$, the chosen local codes,
and the fixed graph-expansion constants. In particular, the theorem
does not give a uniform resource prefactor as $w$ varies. The cited
existence argument selects sufficiently large $w$; it does not certify
a small numerical degree. Indeed, with its auxiliary choices, one
of its sufficient inequalities, $s>2\lambda''$ in Eq.~(12), becomes
\[
 w^{1/6}>16\mu^2(\ln w+2),
\]
where $\mu\sqrt w$ is the graph edge-expansion parameter
\cite[proof of Proposition~1, Eq.~(12)]{panteleev2022asymptotically}.
This illustrates the conservatism of that proof, not a necessary
lower bound on degrees of useful constructions. A practical small-$w$
claim requires explicit local matrices and verified expansion and
rank conditions; no such numerical instance is certified here.

\subsubsection{A transverse complex with cheap valid-syndrome fillings}
\label{si-pk:sec:transverse}

\paragraph*{Shortening preserves the filling estimate}
Retain an invertible maximal minor and one extra column in each full
factor. Denote the restricted maps by $A',B'$. They are surjective and
have free rank-one kernels. Ordering the extra column first gives
\[
 \ker A'=Rq_A,\qquad \ker B'=Rq_B,
 \qquad (q_A)_1=(q_B)_1=1.
\]
This follows by solving for the minor coordinates given the first
coordinate. Column restriction preserves sparsity and cannot decrease
classical kernel distance.

\begin{lemma}[Rectangular shortening preserves short fillings]
\label{si-pk:lem:shortening}
Every boundary of the shortened product has a filling of weight at
most $\rho$ times its weight, provided its weight is less than
\[
 \tau=\min\{\tau_{\mathrm{PK}},d_A/\rho,d_B/\rho\}
       =\Omega(g\sqrt m).
\]
\end{lemma}
\begin{proof}
Extend a shortened input array by zero to the full rectangle. Its
boundary has unchanged physical weight, so
Proposition~\ref{si-pk:prop:pk-input} supplies a full-product filling $Z$.
Outside the retained columns, the first boundary component is zero;
hence every such column of $Z$ lies in $\ker A$. Since
$|Z|\le\rho|\text{boundary}|<d_A$, each is zero. Likewise every row
outside the retained rows lies in $\ker B$ and is zero because
$|Z|<d_B$. Thus $Z$ is supported on the retained rectangle and fills
the shortened boundary at the same cost.
\end{proof}

\paragraph*{The two-factor transverse complex}
Take the lifted tensor product of the shortened two-term factors
\cite[Sec.~1.2 and Appendix~B]{panteleev2022asymptotically}. If the shortened factors have $m_A,m_B$ columns, define
\begin{align*}
 V_0&=R^{m_Am_B},\\
 V_{-1}&=R^{(m_A-1)m_B}\oplus R^{m_A(m_B-1)},\\
 V_{-2}&=R^{(m_A-1)(m_B-1)},
\end{align*}
where $m_A,m_B=\Theta(\sqrt m)$. When these arrays are written as
vectors, coordinate $(1,1)$ is the first coordinate. On arrays, set
\[
 H_1Z=(A'Z,Z(B')^{\mathsf T}),\qquad
 H_2(U,V)=U(B')^{\mathsf T}+A'V.
\]
These maps are sparse, and $H_2H_1=0$ in characteristic two.

\begin{lemma}[Transverse exactness, distance, and filling]
\label{si-pk:lem:transverse}
The transverse complex satisfies
\[
 \ker H_2=\im H_1,\qquad
 \ker H_1=Rq,\qquad q=q_Aq_B^{\mathsf T},\quad q_{1,1}=1.
\]
Its kernel distance $d_H$ and valid-syndrome fillings obey
\[
 d_H\ge d_A d_B/g=\Omega(gm),
\]
\[
 H_2b=0,\quad |b|<\tau
 \quad\Longrightarrow\quad
 \exists z:\ H_1z=b,\quad |z|\le\rho|b|.
\]
\end{lemma}
\begin{proof}
Both factors are split surjections with free rank-one kernels. Each
splits as its kernel plus a contractible two-term complex, so their
tensor product has only the top kernel $Rq$, with $q_{1,1}=1$.
This proves middle exactness; Lemma~\ref{si-pk:lem:shortening} then gives
short fillings of every valid syndrome.

For the weight bound, a nonzero top word is $tq_Aq_B^{\mathsf T}$.
Its first row has weight at least $d_B$, hence at least $d_B/g$
nonzero group-ring entries. Each corresponding column is a nonzero
$A'$ kernel word of weight at least $d_A$. Thus $d_H\ge d_Ad_B/g$.
\end{proof}

\subsubsection{Compatible invertible coordinate minors}
\label{si-pk:app:coordinate-minors}

\begin{lemma}[Coordinate minors without changing the PK local properties]
\label{si-pk:lem:coordinate-minors}
The local matrices and coordinate orderings in the PK construction
can be chosen so that the wide Tanner factor $A$ and the adjoint
$B=T^{\dagger}$ of the tall Tanner factor both have invertible maximal
coordinate minors over $R=\mathbb F_2[G]$.
\end{lemma}
\begin{proof}
Choose full-row-rank local matrices $h,h'$ of width $w$, with
row counts respectively below and above $w/2$ by fixed fractions.
Thus $A$ is wide and $T$ is tall. The two local spaces
\[
 U=\im h^{\mathsf T},\qquad U'=\ker h'
\]
have dimensions at most $\eta w$ for a fixed $\eta<1/2$. We choose
the local pair with the product-expansion and distance properties
of \cite[proof of Proposition~1]{panteleev2022asymptotically}, and with the following
additional property: for each nonzero binary word $x$ and a uniform
coordinate permutation $\pi$,
\begin{equation}
 \Pr[x\in\pi U],\ \Pr[x\in\pi U']
 \le w^3 2^{-(1-\eta)w}.
 \label{si-pk:eq:local-membership}
\end{equation}
To see that these requirements are simultaneous, let $A_j$ count
weight-$j$ words in a uniform $r$-dimensional subspace. Then
$\mathbb E A_j\le 2^{r-w+1}\binom wj$. Markov's inequality and a union
bound over $1\le j\le w$ show that
$A_j\le w^3 2^{r-w}\binom wj$ for all $j$ except with probability
at most $2/w^2$. Apply this to both local spaces. The PK local
conditions hold with probability tending to one for these fixed rate
intervals, so the intersection is nonempty for sufficiently large
fixed $w$. The enumerator bound implies \eqref{si-pk:eq:local-membership}.
In particular, both spaces can also have distance at least $\delta w$
for a fixed $\delta>0$.

Let the base graph have $N=\Theta(\sqrt m)$ vertices, degree $w$,
and nontrivial spectral bound $\lambda_0=O(\sqrt w)$. Independently
permute the local coordinates for $h$ and $h'$ at each base vertex,
and repeat each choice throughout that vertex's lift orbit. Such
permutations preserve the local product and distance properties.
Augmentation, which sends every group element to $1$, recovers the
Tanner matrices on this base graph.

A nonzero row dependence of the augmented $A$ determines a nonempty
edge set whose pattern at every active vertex lies in its permuted
$U$. A nonzero kernel word of the augmented $T$ gives the same type
of event with $U'$. Full local row rank ensures that a nonzero row
dependence gives a nonempty edge set. If $S$ is its active vertex set,
local distance and spectral mixing give
\[
 \delta w|S|\le 2|E(S)|
 \le\left(w\frac{|S|}{N}+\lambda_0\right)|S|.
\]
Hence $|S|\ge(\delta-\lambda_0/w)N\ge\delta N/2$ after fixing $w$
sufficiently large.
For a fixed $S$, there are at most $2^{w|S|/2}$ possible edge sets.
Independence of the local permutations and
\eqref{si-pk:eq:local-membership} therefore bound the probability of a
witness with exactly $S$ as its active
vertex set by
\[
 2^{w|S|/2}\bigl(w^3 2^{-(1-\eta)w}\bigr)^{|S|}
 =2^{-[(1/2-\eta)w-3\log_2w]|S|}.
\]
There are at most $2^N$ choices of $S$. Taking $w$ sufficiently large
makes the total failure probability $2^{-\Omega(N)}$. A further
union bound handles both $A$ and $T$, so both required augmentation
ranks hold simultaneously for sufficiently large $N$.

Finally, write $G=\prod_i C_{2^{a_i}}$. Its group ring is
\[
 R\cong\mathbb F_2[t_1,\ldots,t_b]/
 (t_1^{2^{a_1}},\ldots,t_b^{2^{a_b}}),\qquad b=\text{number of cyclic factors}.
\]
The augmentation ideal is nilpotent. A square minor with nonsingular
augmentation therefore has unit determinant in $R$ and is invertible.
Full row rank of augmented $A$ and full column rank of augmented $T$
supply such minors for $A$ and $B=T^{\dagger}$, respectively.
All choices preserve the previously fixed PK local properties.
\end{proof}

\subsubsection{Boundary definitions for the joint distance theorem}
\label{sec:joint-pk-boundary-definitions}

For Theorem~\ref{thm:joint-pk-spacetime}, use
$E=P\otimes_R V$, $F=\operatorname{Cone}(E\to C)$ and
$W^\circ=F\otimes_R T$.
Logical faithfulness means that inclusion identifies
$H_1(F)=H_1(C)/J$, where $J=\im\gamma_{*,1}$ includes all measured
classes, including residual measurements. In particular,
$\gamma_{*,0}:H_0(P)\to H_0(C)$ is injective. Neither $P_0=C_0$ nor
an expansion assumption on $P$ is required. The base-code family has $O(\log^2 g)$ physical orbits.

\paragraph*{Boundary copies of the base code.}
Let $e_0\in T_0$ be the distinguished physical coordinate. Choose
$e_1\ne e_0$ representing the same generator of $H_0(T)$, so that
$e_0+e_1=\partial_Tt_{01}$ for some $t_{01}\in T_1$.
This is equality of classes in the free $R$-module $H_0(T)$:
a nontrivial group translate of $e_0$ represents a different class.
Require that the $R$-linear maps $R\to T_0$, $1\mapsto e_a$, have
uniformly bounded physical row and column weights. This is possible:
a single physical basis vector $t_{01}$ with nonzero boundary gives
$e_1=e_0+\partial_Tt_{01}$ of bounded weight. More generally the
argument allows any such sparse homologous representatives, without
a bound on the weight of the connecting chain. The second representative
need not be a single coordinate; no separation in extraction rounds is
claimed by this algebraic construction.
Write $j_a(x)=x\otimes e_a$ in $B=C\otimes_R T$, and let
$i_a:C\to W^\circ$ be its composition with the base inclusion.
Put $i=i_0$. The three tests of the protected observables use
\begin{align}
 W^{\rm cl}&=\operatorname{Cone}
       (C\xrightarrow{(i,1)}W^\circ\oplus C),\nonumber\\
 W^{\rm PPM}&=\operatorname{Cone}(C\xrightarrow{i}W^\circ),\qquad
 W^{\rm copy}=\operatorname{Cone}
       (C\oplus C\xrightarrow{(i_0,i_1)}W^\circ).
 \label{eq:joint-pk-boundaries}
\end{align}
These protect surviving $Z$ observables, measured $X$ outcomes, and
all logical $X$ observables, respectively. For a measured representative
$x$, choose $y\in W^\circ_2$ with $\partial y=ix$; its answer cycle is
$(y;x)$. The test of all logical $X$ observables uses
$(x\otimes t_{01};x;x)$, with representatives of a basis of the full
$H_1(C)$, including measured and residual classes. It has interior
support: $\partial(x\otimes t_{01})=i_0x+i_1x$ for a cycle $x$.
These copied observables test input--output $X$ correlations, including
for measured classes; they do not assert preservation of the quantum
information removed by measurement. All degrees are retained, including
$W^\circ_1=F_1T_0\oplus F_0T_1\oplus F_{-1}T_2$.

The superscripts distinguish the closed ($\rm cl$), measured-outcome
($\rm PPM$), and copied-observable ($\rm copy$) boundary tests.
The interior $W^\circ$ contains all base-code and ancillary coordinates
before these boundary copies are added. Explicitly, their degree-$j$
coordinate spaces are
\[
 \begin{aligned}
 W^{\rm cl}_j&=W^\circ_j\oplus C_j\oplus C_{j-1},\\
 W^{\rm PPM}_j&=W^\circ_j\oplus C_{j-1},\\
 W^{\rm copy}_j&=W^\circ_j\oplus C_{j-1}\oplus C_{j-1}.
 \end{aligned}
\]
Thus a degree-$j$ chain of $W^{\rm cl}$ is written $(e,c,\eta)$,
with $e\in W^\circ_j$ on the interior, $c\in C_j$ on the unshifted
boundary copy, and $\eta\in C_{j-1}$ on the shifted boundary copy.
In particular, a degree-one fault has $c\in C_1$ and $\eta\in C_0$;
these are different sets of boundary coordinates.
Here $i$ is the coordinate inclusion of the original base code into
the interior at the distinguished temporal coordinate. The differential is
\[
 \partial(e,c,\eta)=
 (\partial_{W^\circ}e+i(\eta),\ \partial_Cc+\eta,\ \partial_C\eta).
\]
For $W^{\rm cl}$, the reduction to the interior is the chain map
$r(e,c,\eta)=e+i(c)$, with section $s(e)=(e,0,0)$.
The homotopy $H(e,c,\eta)=(0,0,c)$ satisfies
$1+sr=\partial H+H\partial$, so the reduction preserves logical
classes. Thus $W^{\rm cl}$ is a mapping cylinder of $i$, with a
contractible pair of added copies. Since $i$ inserts physical
coordinates, $r$ does not increase weight. Taking the retained $Z$ observables on $W^{\rm cl}$ to be
the pullbacks of the interior observables under $r$ makes their
pairings with faults identical before and after reduction.

\paragraph*{Identifying the protected classes.}
The root inclusion $j:C\to B=C\otimes_R T$ induces an isomorphism
on homology and cohomology, since the distinguished coordinate
$R\to T$ is a chain-homotopy equivalence. Logical faithfulness then
identifies $H_1(W^\circ)$ with $H_1(C)/J$. Thus the degree-one
classes of $W^{\rm cl}$, through its explicit retraction to $W^\circ$,
are exactly the surviving $X$ error classes, paired against
$J^\perp\subseteq H^1(C)$.

For $W^{\rm PPM}$, the long exact sequence of the cone contains
\[
 H_2(W^{\rm PPM})\longrightarrow H_1(C)
 \xrightarrow{i_*}H_1(W^\circ).
\]
The first map sends the class of $(y;x)$ to $[x]$, and its image
is $\ker i_*=J$. Therefore an answer cycle exists for every measured
class. Fix a lift $(y;x)$ for each chosen independent measured
representative. Different lifts may differ by an interior cycle;
the distance argument in Theorem~\ref{thm:joint-pk-spacetime} works for every such choice. In particular,
a cocycle supported only on the boundary, $(0;\zeta)$, evaluates on
this answer cycle as $\langle\zeta,x\rangle$.

For $W^{\rm copy}$, the assignment
$[x]\mapsto[(x\otimes t_{01};x;x)]$ is an injection from the full
$H_1(C)$ into $H_2(W^{\rm copy})$. It is well-defined: replacing
$x$ by $x+\partial_Cu$ changes the cycle by the boundary of
$(u\otimes t_{01};u;u)$. It is injective because the projection to
the two shifted copies sends its class to $([x],[x])$ in
$H_1(C)\oplus H_1(C)$. Thus no measured or residual logical class
is lost from this observable test.

\paragraph*{Why the measured-outcome boundary permits removal.}
In the decomposition $W^\circ=B\oplus Q$, the inclusion of $C$
is $(j,0)$. An undetected degree-two cochain
$(\beta,q;\zeta)$ on $W^{\rm PPM}$ obeys
\[
 \delta_B\beta=0,\qquad
 \delta_Qq+\Gamma_3^\dagger\beta=0,\qquad
 \delta_C\zeta+j^*\beta=0.
\]
Hence $j^*[\beta]=0$. Since $j^*:H^2(B)\to H^2(C)$ is an
isomorphism, $[\beta]=0$. This is an existence assertion about a
primitive; it does not assert that a homotopy inverse has bounded
weight. The short primitive used in the proof is supplied separately
by Proposition~\ref{prop:joint-pk-input}(2).
A degree-one cochain $(h,0;0)$ has coboundary
\[
 (\delta_Bh,\Gamma_2^\dagger h;j^*h),
\]
whereas $(0,k;0)$ has coboundary $(0,\delta_Qk;0)$.
These formulas justify both removals, including their boundary effects.
For $W^{\rm copy}$ the same equations hold with one boundary equation
for each $j_a$, and the first equivalence has boundary components
$(j_0^*h,j_1^*h)$. Both maps have bounded norms by the attachment
hypothesis. Thus the removal includes faults on the interior as well
as on either boundary copy.

\subsubsection{Weight control under chain equivalences}

\begin{lemma}[Weight bounds under local chain equivalences]
\label{lem:pk-bounded-equivalence}
Let $X,Y$ be based binary chain complexes and let
$f:X\to Y$, $\bar f:Y\to X$ be chain maps satisfying
\[
 1_X+\bar f f=\partial_X H+H\partial_X.
\]
Write $\|L\|$ for the maximum column weight of a binary map $L$.
Fix a degree $k$ and constants $\alpha,\beta,L\ge0$ with
$\|f_k\|\le \alpha$, $\|\bar f_{k+1}\|\le \beta$, and $\|H_k\|\le L$,
where $\alpha\ge1$.
Suppose every degree-$k$ cycle of $Y$ of weight below $\tau$
has a filling of weight at most $\rho$ times its weight.
Then every degree-$k$ cycle $z$ of $X$ with $|z|<\tau/\alpha$
has a filling of weight at most $(\alpha\beta\rho+L)|z|$.
The same assertion holds when both filling hypotheses and conclusions
are restricted to cycles that are boundaries.
Moreover, writing $d_k$ for homological distance,
\[
 d_k(X)\ge d_k(Y)/\|f_k\|,
\]
whenever $H_k(X)\ne0$.

The corresponding statements for cochains follow by using $\bar f^\dagger$
as the forward map, $f^\dagger$ as the return map, and $H^\dagger$
as the homotopy. Explicitly, a degree-$k$ cochain filling bound
$(\tau,\rho)$ for $Y$ transfers to $X$ with cutoff $\tau/\alpha'$ and cost
$\alpha'\beta'\rho+L'$, where
\[
 \alpha'=\max\{1,\|\bar f_k^\dagger\|\},\qquad
 \beta'=\|f_{k-1}^\dagger\|,\qquad
 L'=\|H_{k-1}^\dagger\|.
\]
Thus bounded row as well as column weights are required when both
chain and cochain estimates are transferred.
\end{lemma}
\begin{proof}
For a cycle $z$, the chain-map property makes $fz$ a cycle and
$|fz|\le \alpha|z|<\tau$. Choose $v$ with $\partial_Yv=fz$
and $|v|\le\rho|fz|$. The homotopy identity gives
\[
 z=\partial_X(\bar f v+Hz),\qquad
 |\bar f v+Hz|\le(\alpha\beta\rho+L)|z|.
\]
If $z$ is a boundary, then $fz$ is a boundary too, so the same
argument uses only the valid-boundary hypothesis.
Since $\bar f_*f_*=1$, a nonzero class of $X$ maps to a nonzero class of
$Y$; consequently $d_k(Y)\le|fz|\le\|f_k\||z|$.
Taking adjoints of the homotopy identity yields
$1_X+f^\dagger \bar f^\dagger=H^\dagger\delta_X+\delta_XH^\dagger$,
which proves the cochain assertions by the same argument.
\end{proof}

\noindent
For a fixed number of tensor factors, tensoring these maps with
identities preserves their weight bounds. For equivariant maps on based free $R$-modules, the same bounds hold
for balanced products over $R$, with physical binary weights. Tensoring
the homotopy identity is valid in characteristic two: the cross terms cancel.
Composing a fixed number of such equivalences therefore changes
cutoffs and filling costs by constants only. In particular, no inverse
of a global Tanner matrix is needed for this metric transfer.

\subsubsection{Joint factors for spatial and temporal protection}
\label{sec:joint-pk-construction}

The temporal theorem requires simultaneous estimates for $C\otimes T$
and $V\otimes T$. We supply these by choosing the base-code factors and
four transverse factors jointly. This is a separate construction from
arbitrary coordinate shortening above; temporal compatibility does not
follow from the spatial properties alone. All weights below are physical
binary weights, and adjoints include inversion of group elements.
The extension field used to select local codes is a construction tool:
the resulting $C,V,T$ are complexes over $R=\mathbb F_2[G]$ and the
starter and attachment maps may be arbitrary bounded binary $R$-linear
maps satisfying the logical-faithfulness hypotheses of
Theorem~\ref{thm:joint-pk-spacetime}.
Here we use cyclic groups to match the cyclic-lift realization of the
commuting-direction geometry; the separate spatial theorem permits
arbitrary abelian $2$-groups.
Throughout this joint construction and its technical details, citations
to Kalachev--Panteleev~\cite{kalachev2025extendable_doi} use arXiv v2;
citations to Dinur--Lin--Vidick~\cite{dinur2024cubical_doi} use arXiv v3.

\begin{proposition}[Joint PK construction]
\label{prop:joint-pk-input}
For sufficiently large cyclic groups $G$ of power-of-two order $g$,
let $R=\mathbb F_2[G]$ and $\ell=\Theta(\log g)$.
There exist a binary base code $C$ and two-factor complexes $V,T$,
in degrees $0,-1,-2$ and $2,1,0$, respectively, with the following properties.
All physical row and column weights are bounded. Each transverse factor
has $O(\ell)$ orbit coordinates, and $N_C=\Theta(g\ell^2)$.
For constants $c,\rho>0$ independent of $g$, put $\tau=cg\ell$.
\begin{enumerate}
\item Both base-code Pauli distances are $\Omega(g\ell)$.
Evaluation at the first coordinate identifies $V$ up to chain homotopy
with $R$ in degree zero, with $\ker\partial_0^V=Rq$ and $q_1=1$.
Its kernel distance is $\Omega(g\ell^2)$, and each chain-degree-$-1$
cycle of weight below $\tau$ has a filling of weight at most $\rho$
times its own weight. Inclusion of the distinguished physical
coordinate $R\to T$ is a chain-homotopy equivalence.
\item $B=C\otimes_R T$ has degree-one homological distance
$\Omega(g\ell)$. Every degree-two coboundary $\beta$ with
$|\beta|<\tau$ has a degree-one primitive $h$ satisfying
$\delta_Bh=\beta$ and $|h|\le\rho|\beta|$.
\item In $U=V\otimes_R T$, every cycle of weight below $\tau$ in
chain degree $-1$ or $0$, and every cocycle of weight below $\tau$
in cochain degree $0$ or $1$, has a filling of weight at most $\rho$
times its own weight. Although $H_0(U)\cong R$ is nonzero, the
cutoff implies that every degree-zero cycle below $\tau$ is a boundary:
$d_0(U)\ge\tau$. Dually the degree-zero cohomological distance is
also at least $\tau$; neither assertion requires $U$ to be acyclic.
\end{enumerate}
The constants and the choice of $C,V,T$ are independent of the surgery
program. The constant $\rho$ need not equal the Tanner degree $w$ of
the earlier spatial construction.
\end{proposition}

\begin{proof}[Construction and proof]
\leavevmode\par
\paragraph*{Choose all local codes together.}
Choose six local families: two for the base $A\otimes M^\dagger$,
two for $V$, and two for $T$. The products
$A\otimes M^\dagger\otimes T$ and $V\otimes T$ determine a fixed
finite list of tuple types, including subtuples, duals and local flag
endpoints, with at most four directions in each tuple.
After choosing $K=\mathbb F_{2^{s_K}}$, separate binary restriction
adds at most $s_K^{s-1}$ relative Frobenius twists for each type of
arity $s$. Random constituent codes satisfy all these requirements
simultaneously for sufficiently large fixed $s_K$
\cite{kalachev2025extendable_doi}. The parameter order and union bound
are given in Appendix~\ref{sec:joint-pk-technical-details}.
Choose $A,M$ to be wide, full-row-rank Tanner maps at fixed
nondegenerate rates, with unit maximal minors and nonzero kernels.
Their tensor product therefore has nonzero encoded homology.

\paragraph*{Establish binary local robustness.}
Lemma~\ref{lem:joint-pk-binary-flags} transfers the selected expansion
to separately restricted binary coefficient packets and their
intermediate flags, including the duals, uniformly over external
coefficient spaces. Lemma~\ref{lem:joint-pk-local-degrees} then gives
robustness in every local degree, with constants depending only on the
fixed rates and arity.

\paragraph*{Choose physical attachment coordinates.}
For each transverse direction take a tall near-square map
$N=[N_0;\varpi]$, with $N_0$ square and a unit, and with $\varpi$ the
extra output row. Use $N^\dagger$ for a spatial factor
and $N$ for a temporal factor. The extra coordinate normalizes the
kernel or cokernel without applying $N_0^{-1}$ to a physical error.
Lemma~\ref{lem:joint-pk-square-unit} below supplies these square
units by an Euler-orientation determinant argument compatible with
the selected local expansion conditions.

After restricting the local coefficients to binary, retain one bit
of the extra output packet at one chosen endpoint vertex orbit. This
leaves a single physical binary attachment coordinate.
The binary robustness established above applies to this flag and its dual.

\paragraph*{Lift and obtain the metric estimates.}
Choose the random cyclic lifts used in
Proposition~\ref{si-pk:prop:pk-input}, with $\ell=\Theta(\log g)$
vertices per base graph and normalized spectral bound $O(w^{-1/2})$
\cite[arXiv Theorem~2]{agarwal2019grouplifts}.
Apply Lemma~\ref{lem:joint-pk-global-fillings} to the ordinary cubical
representatives of these binary products. It gives both chain and
cochain cutoffs $\Omega(g\ell)$ with constant filling costs: the
$t$-direction product has $\Theta(g\ell^t)$ cells and component
fraction $\Theta(\ell^{1-t})$.
Table~\ref{tab:joint-pk-degree-dictionary} lists the degrees used;
all filling applications are in interior cubical degrees. In particular,
the nonzero homology and cohomology in paper degree zero of $U$ occur
in cubical degree two. Their distances follow directly from
Lemma~\ref{lem:joint-pk-global-fillings}, followed by the bounded
local transfers below. The algebraic equivalences $V\simeq R$ and
$T\simeq R$ are used only to identify homology; their global inverses
need not preserve weight. The two-factor estimates similarly give
both base Pauli distances.
For each normalized spatial factor $N^\dagger$, its local image
$U_{\rm bit}\subset U_+$ has block distance at least $\delta w$
for a fixed $\delta>0$, by the one-direction subtuple condition.
A nonzero vector in $\ker N^\dagger$ assigns matching edge packets
to adjacent active vertices. Local injectivity and distance force each
active vertex to meet at least $\delta w$ active edges. To see the
resulting global bound explicitly, let $n_\Gamma=g\ell$ be the number
of vertices of the underlying lift, and let $S_L,S_R$ be the active
sets in the two halves of its double cover. With $s_a=|S_a|$ and
normalized absolute spectral bound $\lambda<\delta$, mixing gives
\[
 \delta(s_L+s_R)
 \le 2\left(\frac{s_Ls_R}{n_\Gamma}
              +\lambda\sqrt{s_Ls_R}\right).
\]
Both sets are nonempty. Since $s_L+s_R\ge2\sqrt{s_Ls_R}$, this implies
$\sqrt{s_Ls_R}\ge n_\Gamma(\delta-\lambda)=\Omega(g\ell)$.
Each active vertex requires a nonzero binary coefficient, giving
$d(\ker N^\dagger)=\Omega(g\ell)$. Half-edge scalings preserve
these supports. These classical distances give
$d(\ker\partial_0^V)=\Omega(g\ell^2)$, since their balanced-product
kernel distance is at least the product of those distances divided by
$g$: unfolding a balanced kernel word gives a binary matrix whose
nonzero rows and columns belong to the two factor codes, while
repeating each physical cell exactly $g$ times.
Thus $V$ has $O(\ell^2)$ orbit coordinates, corresponding to
$m=\Theta(\ell^2)$ in the spatial removal argument of
Appendix~\ref{si-pk:sec:two-page-proof}. Here its filling cutoff
$\Omega(g\ell)$ comes from Lemma~\ref{lem:joint-pk-global-fillings},
with a fixed cost $\rho=O(1)$. Step~3 of that argument then gives
$\tau_{\rm clean}=\Omega(g\ell)$.

\paragraph*{Pass to adjoint factors.}
We also need a local conversion, since the descent theorem does not
directly concern arbitrary adjoint Tanner matrices. For a Tanner map
$M=hI$, let $L$ be its half-edge space. The edge and local-code exact
sequences are
\[
 0\to E\xrightarrow{I}L\xrightarrow{\Pi}E\to0,
 \qquad
 0\to Y\xrightarrow{h^\dagger}L\xrightarrow{k}Z\to0,
\]
where $\Pi=I^\dagger$ and $k$ checks the dual local code. Choose
local splittings $I\lambda+j\Pi=1$ and
$h^\dagger v+uk=1$. Exactness gives $\Pi I=0$ and $kh^\dagger=0$;
the splittings consequently satisfy
\[
 \lambda I=1_E,\quad \Pi j=1_E,\quad
 vh^\dagger=1_Y,\quad ku=1_Z,\quad \lambda j=0,\quad vu=0.
\]
The two-term maps
$M^\dagger=\Pi h^\dagger$ and $kI$ are chain-homotopy equivalent:
the chain maps, in degrees one and zero, are
\[
 (\lambda h^\dagger,kj),\qquad (vI,\Pi u),
\]
with homotopies $vj$ and $\lambda u$. These identities follow by
substitution using the splitting and exactness identities above.
Every map acts within an edge or one fixed-width local block, so
its physical row and column weights are bounded, also for binary
intermediate flags. Apply the conversion separately in each factor,
then tensor the resulting maps and homotopies.
Lemma~\ref{lem:pk-bounded-equivalence} transfers the ordinary cubical estimates to the mixed
orientations used here, with constant losses in cutoff and filling
cost. The homotopy terms contribute only bounded multiples of the
original fault weight.

\paragraph*{Transfer to the actual binary base.}
Define the actual binary base as
$C=\operatorname{Res}(A\otimes_{K[G]}M^\dagger)$, where $\operatorname{Res}$
denotes restriction of scalars from $K$ to $\mathbb F_2$.
The unit minors split $A$ as a surjection and $M^\dagger$ as an
injection with free cokernel. Consequently,
\[
 H_1(C)=\operatorname{Res}\!\left(
 \ker A\otimes_{K[G]}\operatorname{coker}M^\dagger\right)\ne0,
 \qquad H_0(C)=H_2(C)=0.
\]
Since $B=C\otimes_R T\simeq C$, also $H_2(B)=H^2(B)=0$;
the latter follows by binary linear duality. Thus every degree-two
cocycle of this particular $B$ is a coboundary. A short primitive
still requires the metric estimate in item~2, not merely this
homology calculation.
Separability of $K/\mathbb F_2$ makes $C$ a cellwise
chain summand of
$C^{\rm big}=\operatorname{Res}(A)\otimes_{\mathbb F_2[G]}\operatorname{Res}(M^\dagger)$:
there are bounded binary maps $\iota,p$ with $p\iota=1$.
They act within the coefficient packets and commute with the binary
temporal factors. A nonzero homology class stays nonzero under
$\iota$; a nonzero cohomology class stays nonzero under $p^\dagger$,
since $\iota^\dagger p^\dagger=1$. Thus both distances transfer with
a fixed loss. For a short valid coboundary, lift by $p^\dagger$,
fill in the larger product, and
project the primitive by $\iota^\dagger$; this transfers the filling
bound with a fixed loss as well. This establishes the estimates for
the original binary $C\otimes T$. The starter and its attachment
are chosen afterward as bounded $\mathbb F_2[G]$-linear maps,
subject to the logical-faithfulness hypotheses of
Theorem~\ref{thm:joint-pk-spacetime}.

All local parameters are fixed before $g$ grows. Taking the minimum
of the finitely many cutoffs and enlarging their constant costs proves
Proposition~\ref{prop:joint-pk-input}. The conclusions hold uniformly over the
bounded programs because those programs enter only through forward
attachment maps and $D,D^\dagger$ in Lemma~\ref{lem:joint-pk-starter}.

\end{proof}

\begin{table}[htbp]
\caption{Degrees used in Proposition~\ref{prop:joint-pk-input}.
The cubical representatives are obtained by the bounded local
equivalences in the proof; the estimates for $C$ also use the binary
summand transfer. Cubical chain degree counts edge directions.
The complex $V$ has a shift of $-2$, so a degree of $U=V\otimes T$
is increased by two in the cubical convention. The stronger kernel
bound for $V$ follows from the two-factor distance argument in the proof.}
\label{tab:joint-pk-degree-dictionary}
\centering
\small
\setlength{\tabcolsep}{8pt}
\begin{tabular}{llll}
\hline\hline
Estimate & Paper degree & Dimension $t$ & Cubical degree\\
\hline
$C$: both distances & $1$ & $2$ & $1$\\
$V$: chain filling & $-1$ & $2$ & $1$\\
$V$: kernel distance & $0$ & $2$ & $2$\\
$C\otimes T$: chain distance & $1$ & $4$ & $1$\\
$C\otimes T$: cochain filling & $2$ & $4$ & $2$\\
$U$: chain fillings & $-1,0$ & $4$ & $1,2$\\
$U$: cochain fillings & $0,1$ & $4$ & $2,3$\\
\hline\hline
\end{tabular}
\end{table}

\subsubsection{Technical details for the joint construction}
\label{sec:joint-pk-technical-details}

The local coefficient field is fixed throughout the construction. The
following arguments explain why selecting this field does not introduce
a growing degree, why the attachment coordinate can be a single binary
coordinate, and why the estimates survive the adjoint factors.

\begin{lemma}[A square Tanner map with a unit augmentation]
\label{lem:joint-pk-square-unit}
Let $\mathcal H$ be a finite loopless $w$-regular multigraph, with $w$
even, and put $r_0=w/2$. At each vertex choose an $r_0\times w$
matrix over a field $K$ of characteristic two, every $r_0$-column minor
of which is nonsingular. If $|K|>2$, nonzero scalings of its half-edge
columns can be chosen so that the square Tanner matrix $N_0$ has
nonsingular augmentation. These scalings can be repeated throughout
each lift orbit. For every lift by a cyclic group of power-of-two order,
$N_0$ is consequently a unit over $K[G]$.

Deleting any collection of local rows gives a wide Tanner matrix with a
unit maximal coordinate minor. Adding one independent local row at one
chosen endpoint vertex orbit gives a tall map $N=[N_0;\varpi]$. After scalar restriction to
$\mathbb F_2$, retaining a single binary coordinate of this extra output
gives a map with exactly one more binary output orbit than input orbits.
\end{lemma}
\begin{proof}
Orient an Euler circuit in each component of $\mathcal H$. Every vertex
then receives exactly $r_0$ edges. Introduce one independent indeterminate
$s_{v,e}$ for every half-edge, multiplying the local column at $(v,e)$.
In the determinant of the augmented Tanner matrix, group terms according
to which endpoint receives each edge column. The coefficient of the
monomial that assigns every edge to its head in the Euler orientation is
the product of the corresponding $r_0$-column local determinants.
It is nonzero. Thus the determinant polynomial is nonzero and has degree
at most one in each indeterminate. A nonzero polynomial of degree at most
one in each variable cannot vanish on $(K^\times)^a$ when $|K^\times|>1$:
this follows by induction on the number of variables, viewing it as a
linear polynomial in the last variable. Nonzero half-edge scalings with
nonzero determinant therefore exist, regardless of the number of vertices.

For $G=C_{2^a}$, the ideal $(x-1)$ in
$K[G]=K[x]/((x-1)^{2^a})$ is nilpotent. Nonzero augmented determinant
therefore implies a unit determinant. The chosen scalings change local
codes only by coordinatewise invertible maps, so preserve all their
support-based expansion estimates and nested inclusions.

Every subset of the rows of the nonsingular augmented matrix is linearly
independent. An ordinary pivot minor of that rectangular matrix has
nonzero determinant over $K$, and its lifted minor is a unit by the same
argument. Restriction of scalars preserves the square unit $N_0$.
Adding one binary output row preserves its full column rank and gives
free cokernel of rank one, represented by the extra physical output.
No inverse of $N_0$ is applied to a physical fault.
\end{proof}

For the base-code factors, choose the smaller local row spaces as nested
subspaces of the square local row spaces in this lemma. Thus their rates
can stay a fixed distance below $1/2$, while the lemma supplies their unit
maximal minors. For transverse factors use the square row space and its
one-row extension. All endpoint row spaces, their duals, and their
required tuples are included in the joint local selection.

\paragraph*{Block weights for the binary local codes.}
For each direction $i$, let $V_i$ be a finite-dimensional binary space
and let $L_i\subset V_i^{[w]}$. A word on the $s$-dimensional grid has
one coefficient in $\bigotimes_i V_i\otimes E$ at every grid point,
where $E$ is any additional binary coefficient space. Its block weight
counts nonzero grid points. Write $|y_i|_i$ for the number of nonzero
lines parallel to direction $i$. We call a family uniformly
$\eta$-expanding in block weight if every word in the sum of its
directional code spaces has a decomposition $y=\sum_i y_i$ satisfying
\[
 w\sum_i |y_i|_i\le\eta^{-1}|y|_{\rm blk},
\]
for every choice of $E$. This definition does not depend on the dimensions
of the coefficient spaces.

\begin{lemma}[Binary restriction and intermediate flags]
\label{lem:joint-pk-binary-flags}
Fix an arity bound $s\le4$. Suppose the required full-field local tuples,
all their subtuples, and their relative Frobenius twists have product
expansion at least $\varepsilon>0$. Their separately restricted binary
codes are uniformly expanding in block weight, with constants depending
only on $s,\varepsilon$. The same conclusion holds after taking arbitrary
binary subcodes in any of their directions.

Consequently, if the two nested endpoint row spaces are
$U_0\subset U_+$, every intermediate binary space
$U_0\subset U_{\rm bit}\subset U_+$ has the needed expansion both
for $U_{\rm bit}$ and for its binary dual, provided the joint selection
includes the endpoints $U_+$ and $U_0^\perp$ in every required tuple.
The constants are independent of the extension degree and of $w$.
\end{lemma}
\begin{proof}
Separability decomposes a tensor product of separately restricted
$K=\mathbb F_{2^{s_K}}$ coefficient packets into Frobenius-twisted $K$
components. The decomposition is an invertible change within each grid
cell. Further expanding $E$ in a basis, the block support is therefore
exactly the union of the supports of the resulting scalar components.

For a sum-code word $y$, take the $\varepsilon$-closure $M$ of this
union. Closure is monotone in its initial set: repeatedly add a line
when at least $\varepsilon w$ of its cells are present. Moreover,
an $\varepsilon$-closed set is $\rho'$-closed for every
$\rho'\ge\varepsilon$. Thus the same $M$ is closed for the expansion
constant of every scalar tuple, and each component supported in $M$
is generated by lines contained in $M$. The support-closure results
\cite[Lemmas 17--19]{kalachev2025extendable_doi} show that
\[
 |M|\le c_s|y|_{\rm blk},\qquad
 c_s=\bigl((2^s+1)/\varepsilon\bigr)^s,
\]
and every scalar component can be generated using directional codewords
on lines contained in this same $M$. Combine all component contributions
on each physical line. There are at most $|M|/w$ such lines in any one
direction. Hence $w\sum_i|y_i|_i\le s|M|$, independently of the number
of scalar components. This proves the first assertion with
$\eta_s=(s c_s)^{-1}$. It is essential here to use one common support
closure, rather than summing the weights of scalar components.

For completeness the subcode argument
\cite[Lemma 10]{kalachev2025extendable_doi} also works with these
block alphabets. It suffices to replace $L_1$ by $L'_1\subset L_1$.
Choose $L_1=L'_1\oplus A$ and apply the old decomposition to $y$.
Project its first directional summand line by line as $y_1=y'_1+y_A$.
Neither projection activates a new line, and
$|y_A|_{\rm blk}\le w|y_1|_1$. Since $y$ belongs to the new sum code,
the tensor-intersection identity gives
$y_A\in A\otimes(L_2\boxplus\cdots\boxplus L_s)\otimes E$.
On every first-coordinate slice, apply the uniform expansion of the
remaining tuple, with external coefficient space $V_1\otimes E$.
Adding those decompositions to the other directional summands eliminates
$y_A$ at a cost of at most $|y_A|_{\rm blk}/\eta_{\rm rest}$ in the
quantity $w\sum_i|y_i|_i$. Thus the new expansion constant is at least
\[
 \frac{\eta_s}{1+\eta_{\rm rest}^{-1}}.
\]
Repeat for at most $s$ directions, retaining the minimum constants over
all subtuples. One explicit uniform choice is obtained by setting
\[
 \xi_0=\min_{1\le j\le s}(j c_j)^{-1},\qquad
 \xi_{r+1}=\frac{\xi_r^2}{1+\xi_r}\quad(0\le r<s).
\]
At each replacement both the current tuple and its remaining subtuple
have constant at least $\xi_r$, so $\xi_s$ works after all replacements.
This proves uniformity over all external spaces too.
Finally $U_{\rm bit}\subset U_+$ and
$U_{\rm bit}^{\perp}\subset U_0^{\perp}$, so the two endpoint choices
give the primal and dual claims. Binary duals are taken in trace-dual
packet bases; changing such bases preserves block support.
\end{proof}

The half-edge scalings used in
Lemma~\ref{lem:joint-pk-square-unit} are compatible with this selection.
A coordinate permutation permutes the binary edge packets, and a
nonzero field scaling acts invertibly within each packet. Tensoring these maps gives a
support-preserving isomorphism of each local complex; the inverse
transpose gives the corresponding isomorphism for its dual. The
endpoint inclusions defining a flag are preserved as well. These
are linkwise statements: scalings at opposite ends of an edge may
differ, so no global isomorphism to the unscaled complex is asserted.

\begin{lemma}[Local fillings in every degree]
\label{lem:joint-pk-local-degrees}
Let $j_i:W_i\hookrightarrow V_i^{[w]}$ be injective binary maps whose
image tuples and subtuples satisfy Lemma~\ref{lem:joint-pk-binary-flags}.
In the tensor product of the two-term cochain complexes $j_i$, every
local coboundary $y$ of positive degree has a primitive of block weight
at most $|y|_{\rm blk}/(\kappa w)$, where $\kappa>0$ depends only on
the arity and the block expansion constants. This remains true after
tensoring with every external coefficient space $E$.
\end{lemma}
\begin{proof}
At the top degree this is precisely block product expansion: each
nonzero primitive face is one active directional line. We induct on the
number of directions to treat lower degrees, retaining arbitrary $E$.
Split a coboundary $y$ into faces omitting the last direction ($y_0$)
and faces containing it ($y_1$). The first is a coboundary in the smaller
complex with external coefficient space $W_{\rm last}\otimes E$.
Fill it by $u_0$ with
$|u_0|\le |y_0|/(\kappa_0 w)$ and subtract $\delta u_0$.
The remainder consists of $w$ closed slices of the smaller complex,
with external space $V_{\rm last}\otimes E$.

Each injection $j_i$ has cohomology only in degree one, so their tensor
product is exact below its top degree. Therefore these slices are
coboundaries in every lower-degree case under consideration; in degree
zero they vanish. Fill the slices using the induction hypothesis.
The last-direction part of $\delta u_0$ has weight at most $w|u_0|$.
The total primitive consequently has weight at most
\[
 \frac{|y_0|}{\kappa_0w}
 +\frac{|y_1|+w|u_0|}{\kappa_0w}
 \le\frac{1+\kappa_0^{-1}}{\kappa_0w}
       (|y_0|+|y_1|).
\]
Taking the minimum of $\kappa_0^2/(1+\kappa_0)$ and the top-degree
constant closes the induction. This argument is adapted from
\cite[Lemma~5.13]{dinur2024cubical_doi}. Uniformity over external
coefficient spaces is supplied here directly by
Lemma~\ref{lem:joint-pk-binary-flags}, so the separate collective-robustness
step of Proposition~5.14 in that reference is not needed.
No coefficient dimension occurs in the bound.
\end{proof}

\begin{corollary}[Uniform local robustness]
\label{cor:joint-pk-local-robustness}
Under Lemma~\ref{lem:joint-pk-local-degrees}, a cochain $z$ that is
minimal modulo local coboundaries, below the top local degree, obeys
$|\delta z|_{\rm blk}\ge\kappa w|z|_{\rm blk}$, uniformly over
external binary coefficient spaces.
\end{corollary}
\begin{proof}
Choose a primitive $u$ of $\delta z$ with
$|u|_{\rm blk}\le |\delta z|_{\rm blk}/(\kappa w)$.
Local exactness makes $z-u$ a coboundary (zero in degree zero).
Minimality gives $|z|_{\rm blk}\le |u|_{\rm blk}$.
\end{proof}

\paragraph*{Order of parameters and simultaneous local selection.}
Fix the maximum arity four and rate intervals compactly contained in
$(0,1)$. The random-code theorem
\cite[Theorem~2 and Corollary~3]{kalachev2025extendable_doi}
supplies an endpoint product-expansion constant depending only on
fixed upper bounds for the rates and dual rates, and on arity.
The preceding lemmas and corollary give a local robustness constant independent
of $w,s_K$. First choose a sufficiently large fixed even $w$ for the
spectral requirement below. Then choose a sufficiently large fixed $s_K$.
Draw a largest flag space uniformly, then each smaller space uniformly
inside its predecessor. Every endpoint is uniform among subspaces of
its dimension, by invariance under invertible changes of coordinates.
Choose the direction families independently. Each required tuple uses
at most one endpoint from each family; relative Frobenius twists preserve
these marginals and independence. Dependence between different tuples
is harmless for the union bound.
For arity $s\le4$, Theorem~2 bounds failure by
$w^s2^{w^s-s_K+1}$. There are at most $s_K^{s-1}$ relative twists
per endpoint tuple. Thus the joint failure probability is at most
$C(w)s_K^3 2^{-s_K}$, where $C(w)$ can be of order
$w^4 2^{w^4}$ times the fixed number of endpoint choices.
Taking $s_K>w^4+O(\log w+\log s_K)$ suffices for this part of the
selection. The resulting fixed alphabet and binary degree constants
can be extremely large; this is an existence proof, not a practical
parameter estimate. Include all $r_0$-column nonsingular-minor conditions
for each square local family in Lemma~\ref{lem:joint-pk-square-unit}.
At fixed $w$ these are finitely many additional events, with total
failure probability $O_w(2^{-s_K})$. The union bound selects the flag
and minor properties together, without conditioning the flag distribution
on a minor event. Half-edge scalings from
Lemma~\ref{lem:joint-pk-square-unit} preserve every selected property.
Only after these constants have been fixed do $g$ and $\ell$ grow.

\paragraph*{The lift geometry and the cutoff.}
Use the abelian-lift geometry of
\cite[Section 3.5.2]{dinur2024cubical_doi}, with cyclic lift group $G$
and $\ell=\Theta(\log g)$ vertices in an even-degree nonbipartite
loopless base graph. Use the same random cyclic-lift existence theorem as in
Proposition~\ref{si-pk:prop:pk-input}
\cite[arXiv Theorem~2]{agarwal2019grouplifts}. Its size and spectral
hypotheses were checked there. This gives $\ell=\Theta(\log g)$
and normalized spectral parameter $\lambda=O(w^{-1/2})$; no explicit
polynomial-time lift construction is required here. For the adjacency
matrix of each underlying lift $\widehat\Gamma$, this means
\[
 \left\|w^{-1}A_{\widehat\Gamma}|_{\mathbf1^\perp}\right\|_2
 \le\lambda<1.
\]
Thus the lift is connected and nonbipartite, and the bound controls
both signs of its nonleading eigenvalues.
Decompose the base edges into oriented $2$-factors, giving permutation
ports and their inverses; the voltages lift these ports to permutations.
Ports retain their labels even if two permutations coincide, so
incidence counts are unchanged.
The physical Tanner graph is the bipartite double cover, with its
endpoint bit retained among the coordinates. The spectral parameter
below belongs to the underlying nonbipartite lift, not to the absolute
spectrum of its double cover, which has the eigenvalue $-1$.
The double cover has $2\ell$ vertex orbits and $w\ell$ edge orbits;
$w/2$ local rows at each vertex therefore still give a square $N_0$.
Add the extra local row at one chosen endpoint vertex orbit and retain
one binary output bit, as in Lemma~\ref{lem:joint-pk-square-unit}.
Half-edge scalings change only the local coefficient maps, leaving
these permutation ports unchanged.

For $t\le4$ directions, write
$\mathcal S=G\times S_1\times\cdots\times S_t$, where $S_i$ is the
factor's base-vertex set and $|S_i|=\ell$. This is the coordinate set
called $G$ in Ref.~\cite{dinur2024cubical_doi}; it excludes the cubical
endpoint bits. A move in direction $i$ changes only $S_i$ and multiplies
the shared group coordinate by the corresponding voltage. Moves from
different directions commute because our lift group $G$ is abelian.
The direction-$i$ graph on $\mathcal S$ consists of connected copies
of $\widehat\Gamma_i$, each with $g\ell$ vertices. Define
\[
 r_{\rm comp}=\min_{i,\,\mathcal K}
       \frac{|\mathcal K|}{|\mathcal S|}
       =\ell^{1-t},
 \qquad
 |X(k)|=\binom tk2^{t-k}w^k|\mathcal S|,
\]
where $\mathcal K$ ranges over directional connected components on
$\mathcal S$. This is exactly the normalization of $r$ in
Ref.~\cite[Definition~6.4]{dinur2024cubical_doi}. The $2^t$ endpoint
choices appear in $|X(0)|$, not in the denominator defining
$r_{\rm comp}$. In particular, $r_{\rm comp}|X(k)|=\Theta(g\ell)$
for fixed $w,t$.

\paragraph*{Global distance and fillings from binary links.}
We state the extension of the cubical estimates separately from the
local-code selection.

\begin{lemma}[Cubical estimates with binary coefficient packets]
\label{lem:joint-pk-global-fillings}
Let $X$ be a $t$-direction cubical complex on the lift geometry above,
with $2\le t\le4$. In each direction $i$, fix an edge packet $V_i$.
At a factor vertex $u$, let the local cochain map be an injection
$j_{i,u}:W_{i,u}\hookrightarrow V_i^{[w]}$; it depends only on that
factor vertex, including its endpoint bit. Assume that every link
tuple and its dual have the uniform local robustness supplied by
Corollary~\ref{cor:joint-pk-local-robustness}, with constant at least
$0<\kappa\le1$. Suppose each direction graph is a disjoint union of
components of size at least $r_{\mathrm{comp}}|\mathcal S|$ on the
coordinate set $\mathcal S$ defined above,
with normalized nonleading eigenvalues bounded in magnitude by
$\lambda$. Require $\lambda C_1<1$, in every degree used below, where for $k<t$ we use
\[
 C_1=t^2 2^{t^2+3t}\kappa^{-(k+1)},\qquad C_2=2^{2t}w^tC_1.
\]
These are the uniform-robustness constants of
Ref.~\cite[Proposition~7.1, arXiv v3]{dinur2024cubical_doi}.
Let $\widetilde X$ replace each local image by its binary orthogonal
complement in the same edge packets. In block weight:
\begin{enumerate}
\item For $0\le k<t$, both $X$ and $\widetilde X$ satisfy
\[
 d_{\rm coloc}(k)\ge
 \frac{1-\lambda C_1}{C_2}\,r_{\mathrm{comp}}|X(k)|.
\]
Here $C_1$ may depend on $k,t,\kappa$, but not on packet dimensions.
\item For $1\le k<t$, every cocycle $x$ with
$|x|<d_{\rm coloc}(k)$ has a primitive $y$ with
$\delta y=x$ and $|y|\le L_t|x|$, where $L_t=(2w)^t$.
\item For $1\le k<t$, every cycle $x$ with
\[
 |x|<(2wt)^{-t}\widetilde d_{\rm coloc}(t-k)
\]
has a filling $z$ with $\partial z=x$ and
\[
 |z|\le C_{\rm fill}(t,w)|x|,\qquad
 C_{\rm fill}(t,w)=2L_t(t2^{2t}w^{t+1})^{t-1}.
\]
For $k=t$, every nonzero cycle has weight at least
$(2wt)^{-t}\widetilde d_{\rm coloc}(0)$.
\end{enumerate}
Here $d_{\rm coloc}(k)$ is the minimum block weight of a nonzero
cocycle that cannot be reduced by a coboundary from any one vertex's
upward neighborhood. In degree zero there are no incoming
coboundaries, so every cochain is locally co-minimal and
$d_{\rm coloc}(0)$ is the minimum weight of a nonzero zero-cocycle.
The tilde denotes this distance for $\widetilde X$.
\end{lemma}
\begin{proof}
\emph{Local coefficients and co-minimal distance.}
For a face $f$ missing direction $i$, write $\pi_i(f)$ for its
direction-$i$ factor vertex. Along any other direction this vertex
is unchanged. Consequently, the two paths around a square use the
same local maps on different tensor factors; they commute and cancel
in characteristic two. An upward link at $f$ is exactly the tensor
complex of $(j_{i,\pi_i(f)})_{i\notin\operatorname{type}(f)}$, with
the active edge packets as an external coefficient space. This is
the packet version of Ref.~\cite[Lemma 5.2]{dinur2024cubical_doi}.
The local exactness statements of its Lemmas 5.3--5.4 follow by
tensoring the surjections $j_{i,u}^{\mathsf T}$ over $\mathbb F_2$;
their top kernels are the tensor products of the local dual images.

We detail the extension of its Claims 7.4--7.6. Let $x$ be a nonzero
locally co-minimal $k$-cocycle and $\mathcal A$ its active faces.
For $v\in X(j)$, $j\le k$, let $\mathrm{Nb}_v(k)$ and
$\mathrm{Op}_v(k)$ be the $k$-faces sharing a $(k+1)$-coface with
$v$ that, respectively, meet but do not contain $v$, or are disjoint
from $v$ (Definition~6.7 of that reference).

\emph{(a) Link minimality (Claim 7.4).}
A primitive supported above $v$ changes only faces above $v$.
Extending it by zero therefore shows that $x|_{X_{\ge v}(k)}$ is
minimal in its link. For $j=k$, link degree zero has no incoming
coboundaries.

\emph{(b) Packet robustness and cancellation (Claim 7.5).}
Corollary~\ref{cor:joint-pk-local-robustness} and $\delta x=0$ give
\[
 \kappa w|x|_{X_{\ge v}(k)}
 \le |\delta_{\rm link}(x|_{X_{\ge v}(k)})|
 \le |x|_{\mathrm{Nb}_v(k)}+|x|_{\mathrm{Op}_v(k)}.
\]
The second inequality counts nonzero packets: every nonzero output
needs an outside contribution, and each outside face contributes to
at most one such output. Varying packet dimensions does not change
this incidence bound.

\emph{(c) Descent through faces (Claim 7.6).}
Lemma~6.9's covering multiplicities obey $a_{k,j}\le2^t$.
Iterating (b) until $j=0$, where $\mathrm{Nb}_v(k)$ is empty, gives
\[
 |x(f)|\le\sum_{j=0}^k\sum_{v\preceq f,\,\dim v=j}
 \frac{2^{t(k-j)}}{(\kappa w)^{k-j+1}}
 |x|_{\mathrm{Op}_v(k)}.
\]

\emph{(d) Geometric averaging.}
Let $W^{(k,j)}$ choose a $j$-subface, take one step in a missing
direction, and choose an incident $k$-face (Definition~6.6).
Claim~6.11 converts (c), summed over $f\in\mathcal A$, into
\[
 |\mathcal A|\le\sum_{j=0}^k b_j
 \langle1_{\mathcal A},W^{(k,j)}1_{\mathcal A}\rangle,\qquad
 b_j=\frac{\binom kj\binom{t-j}{k-j}(t-j)2^{(t+1)(k-j)}}
 {\kappa^{k-j+1}}.
\]
Lemma~6.12 bounds each inner product by
\[
 \lambda|\mathcal A|+
 \binom tj2^{t-j-1}w^j\frac{|\mathcal A|^2}
 {r_{\rm comp}|X(k)|}.
\]
Only the directional walks enter this bound, so the double-cover
geometry and its endpoint bits are already included. Since
$\sum_jb_j\le C_1$ and
$\sum_jb_j\binom tj2^{t-j-1}w^j\le C_2$, we obtain
\[
 (1-\lambda C_1)|\mathcal A|
 \le C_2\frac{|\mathcal A|^2}{r_{\rm comp}|X(k)|}.
\]
This proves item 1. Applying the same argument to the dual link
injections proves it for $\widetilde X$. All coefficient-dependent
steps were confined to (a)--(b); the bounded adjoint conversions are
applied separately after this lemma.

\emph{Short cochain fillings.}
A nonzero cocycle below $d_{\rm coloc}(k)$ is not locally co-minimal.
Thus adding a coboundary supported in the upward neighborhood of
one vertex strictly reduces its integer block weight. At most
$|x|$ such steps reduce it to zero. Each primitive occupies at most
\[
 \max_v|X_{\ge v}(k-1)|
 =\binom{t}{k-1}w^{k-1}\le L_t
\]
faces, so the sum of the primitives proves item 2. This also bounds
cohomological distance. The argument uses local minimality rather
than the global soundness constant.

\emph{Short chain fillings.}
Use the local-view construction of
Ref.~\cite[Proposition 8.1]{dinur2024cubical_doi}.
Its auxiliary differential $\Delta$ compares views on incident
faces, whereas $\partial_L$ applies the local chain boundary.
The exactness and commutation identities in Lemmas 8.4--8.5 remain
valid with packets: for a fixed upper face the inactive coefficient
spaces agree, and the comparison maps are identities on those
spaces. Local fillings use only the exactness just established.
Put $a=wt$, $b=2^{2t}w^t$, and $s=t-k-1$.
Starting from a cycle $x$, the construction gives local views and
local primitives satisfying
\[
 |x^{(0)}|\le2^t|x|,\qquad
 |z^{(j)}|\le a^j|x^{(0)}|,\qquad
 |x^{(j+1)}|\le a^{j+1}|x^{(0)}|.
\]
If it does not terminate earlier, its last obstruction is
$x^{(t-k)}$. Decode its local tensor-kernel coefficients to obtain
a cochain $\widetilde x\in C^{t-k}(\widetilde X)$. Decoding and
re-encoding are inverse isomorphisms on each nonzero coefficient
block and preserve block support. They commute with comparisons
because $\pi_i$ is constant across every other direction. The
calculation of Claim 8.8 therefore gives
\[
 \widetilde\delta\widetilde x=0,\qquad
 |\widetilde x|=|x^{(t-k)}|
 \le a^{t-k}2^t|x|<\widetilde d_{\rm coloc}(t-k).
\]
Apply item 2 in the dual complex to obtain
$\widetilde\delta\widetilde u=\widetilde x$ with
$|\widetilde u|\le L_t|\widetilde x|$. Re-encoding gives $u^{(s)}$
with $\Delta u^{(s)}=x^{(s+1)}$ and $\partial_Lu^{(s)}=0$.
Thus $\widehat z^{(s)}=z^{(s)}+u^{(s)}$ has zero final obstruction
and still has local boundary $x^{(s)}$.

The reconstruction in Claim 8.7 gives
\[
 \begin{aligned}
 |\widehat z^{(s)}|&\le(a^s+L_ta^{s+1})|x^{(0)}|,\\
 |\widehat z^{(j)}|&\le a^j|x^{(0)}|+b|\widehat z^{(j+1)}|.
 \end{aligned}
\]
The second line applies for $0\le j<s$: geometric exactness supplies
a comparison primitive of weight at most $b|\widehat z^{(j+1)}|$,
whose local boundary corrects $z^{(j)}$. Stitching the consistent
vertex views costs at most $2^tw^t$. Hence
\[
 \begin{aligned}
 |z|&\le b\left[\sum_{j=0}^s(ab)^j+L_ta(ab)^s\right]|x|\\
 &\le 2L_t(ab)^{t-1}|x|.
 \end{aligned}
\]
Here $s\le t-2$, $ab\ge2$, and $L_ta\ge2$. An earlier termination
omits the last correction and satisfies the same bound. This proves
item 3 without the inverse global soundness factor in Claim 8.8.
If $k=t$, decoding the initial views instead gives a dual degree-zero
cocycle of weight at most $2^t|x|$. Below the stated cutoff it must
vanish, since degree zero has no incoming coboundaries. Injectivity
of re-encoding then forces $x=0$, proving the top-degree assertion.
\end{proof}

For the selected local families, $C_1$ depends only on the fixed
rates and arity. Choose $w$ so that $\lambda C_1<1/2$ in every
required degree, then choose $s_K$ as above. The lemma gives cutoffs
$\Omega(r_{\mathrm{comp}}|X(k)|)=\Omega(g\ell)$ in both orientations.
If $D_{\rm pkt}$ bounds the binary dimension of a face coefficient
space, then
\[
 |x|_{\rm blk}\le |x|_{\rm bit}
       \le D_{\rm pkt}|x|_{\rm blk}.
\]
Thus the block cutoffs are also valid binary cutoffs, and the filling
costs grow by at most $D_{\rm pkt}$. This is a constant once $w,s_K,t$
are fixed. Lemma~\ref{lem:pk-bounded-equivalence} then transfers the
chain and cochain estimates through the local adjoint conversions;
the binary summand maps in the proof of
Proposition~\ref{prop:joint-pk-input} transfer them to the actual base.

\subsubsection{Passing the filling bounds through the starter}
\begin{lemma}[The starter costs only a constant]
\label{lem:joint-pk-starter}
Let $U$ be supplied by Proposition~\ref{prop:joint-pk-input} and
$d_*=g\ell$. For any bounded-degree two-term starter
$P=(P_1\xrightarrow{D}P_0)$ over based free $R$-modules, put
$Q=(P\otimes_R U)[1]$, with $(X[1])_j=X_{j-1}$.
Every short $Q_1$ cycle has a constant-cost $Q_2$ filling, and every
short degree-two $Q$ cocycle is a coboundary. Their cutoffs are
$\Omega(d_*)$ and the constants depend on $P$ only through its degree
bound.
\end{lemma}
\begin{proof}
Write a $Q_1$ cycle as $(a,b)\in P_0U_0\oplus P_1U_{-1}$. Its
equations are $\partial_Ub=0$ and $\partial_Ua+Db=0$.
Fill $b=\partial_Uv$ and then fill the cycle $a+Dv=\partial_Uu$.
The primitive $(u,v)$ has weight $O(|a|+|b|)$.
All operations are performed rowwise in the free $P$ coordinates.

Dually, a degree-two cocycle has components
$a\in P_1^*U^0$ and $b\in P_0^*U^1$ satisfying
$\delta_Ub=0$ and $\delta_Ua+D^\dagger b=0$.
Fill $b=\delta_Uv$ and subtract that coboundary. The remaining
$a+D^\dagger v$ is a short $U^0$ cocycle, hence a coboundary.
This clears the second sector as well. Bounded $D,D^\dagger$ keep
all intermediate weights below the required cutoffs after a fixed
reduction of the initial threshold.
Explicitly, if the four $U$ estimates have cutoff $\tau$ and cost
$\rho$, and $w_P$ bounds both $D$ and $D^\dagger$, one may use
cutoff $\tau/\max\{1,w_P\rho\}$ and filling cost
$\rho(1+w_P\rho)$. These constants do not depend on the number of
starter orbits.
\end{proof}

\fi

% ==========================================================================
\ifdefined\appendixonly
\clearpage
\bibliography{refs}

\fi

\end{document}